%% file: main_arxiv_spanning_tree_packing.tex
\documentclass[11pt]{article}
\usepackage{amsmath,amsfonts,amssymb,amsthm}
\usepackage{fullpage}
\usepackage{algorithm}
\usepackage{algorithmic}
\usepackage{epstopdf}
\usepackage{graphicx}
\usepackage{hyperref}
\usepackage{cite}
\usepackage{color}
\usepackage[sans]{dsfont}
\usepackage{mathtools}

\newtheorem{theorem}{Theorem}[section]
\newtheorem{corollary}[theorem]{Corollary}
\newtheorem{proposition}[theorem]{Proposition}
\newtheorem{heuristic algorithm}{Heuristic Algorithm}
\newtheorem{lemma}[theorem]{Lemma}
\newtheorem{definition}{Definition}
\newtheorem{claim}[theorem]{Claim}

\newtheorem{observation}[theorem]{Observation}

\newtheorem{fact}{Fact}

\newtheorem{approximation algorithm}{Approximation Algorithm}
\newcommand{\poly}{\mathsf{poly}}
\newcommand{\dilation}{\mbox{\tt d}}
\newcommand{\congestion}{\mbox{\tt c}}
\newcommand{\congedge}{\eta}
\newcommand{\congest}{${\mathsf{CONGEST}}$}
\newcommand{\NonTreeCover}{\mathsf{NonTreeCover}}

\newcommand{\TreeCover}{\mathsf{TreeCover}}
\newcommand{\GraphCover}{\mathsf{CycleCover}}
\def\cA{{\mathcal A}}
\def\cC{{\mathcal C}}

\newcommand{\depth}{\mathsf{Depth}}
\newcommand{\ceil}[1]{\ensuremath{\left\lceil#1\right\rceil}}
\newcommand{\floor}[1]{\ensuremath{\left\lfloor#1\right\rfloor}}
\newcommand{\notS}{\overline{S}}
\newcommand{\cset}{\mathcal{C}}
\newcommand{\fset}{\mathcal{F}}
\newcommand{\tset}{\mathcal{T}}
\newcommand{\set}[1]{\{#1\}}
\newcommand{\pset}{{\mathcal{P}}}
\newcommand{\rset}{{\mathcal{R}}}
\newcommand{\qset}{{\mathcal{Q}}}
\newcommand{\mset}{{\mathcal{M}}}
\newcommand{\tqset}{\tilde{\mathcal{Q}}}

\newenvironment{proofof}[1]{\noindent{\bf Proof of #1.}}%
{\hfill\stopproof}
\renewenvironment{proof}{\par \smallskip{\bf Proof:}}{\hfill\stopproof}

\newcommand{\expect}[2][]{\text{\bf E}_{#1}\left [#2\right]}
\newcommand{\prob}[2][]{\text{\bf Pr}_{#1}\left [#2\right]}
\newcommand{\eps}{\epsilon}

\newenvironment{properties}[2][0]
{
	\begin{enumerate} \setcounter{enumi}{#1}}{\end{enumerate}}
	
\newcounter{note}[section]

\newcommand{\mynote}[2][red]{\textcolor{#1}{\sc\bf{[#2]}}}

\def\stopproof{\square}
\def\square{\vbox{\hrule height.2pt\hbox{\vrule width.2pt height5pt \kern5pt
\vrule width.2pt} \hrule height.2pt}}

\makeatletter
\newtheorem*{rep@theorem}{\rep@title}
\newcommand{\newreptheorem}[2]{%
	\newenvironment{rep#1}[1]{%
		\def\rep@title{#2 \ref{##1}}%
		\begin{rep@theorem}}%
		{\end{rep@theorem}}}
\makeatother

\newreptheorem{theorem}{Theorem}

\newcommand{\mail}{\operatorname{mail}}
\newcommand{\Mail}{\operatorname{Mail}}

\newcommand{\opt}{\text{OPT}}
\newcommand{\dist}{\operatorname{dist}}
\newcommand{\out}{\mathsf{out}}

\newcommand{\diam}{\operatorname{diam}}

\begin{document}

\begin{titlepage}
\title{On Packing Low-Diameter Spanning Trees}
\author{Julia Chuzhoy\thanks{Toyota Technological Institute at Chicago. Email: {\tt cjulia@ttic.edu}. Part of the work was done while the author was a Weston visiting professor in the Department of Computer Science and Applied Mathematics, Weizmann Institute of Science. Supported in part by NSF grant CCF-1616584.}  \and Merav Parter\thanks{The Weizmann Institute of Science, Israel. Email: {\tt merav.parter@weizmann.ac.il}. Supported in part by an ISF grant (no. 2084/18).} \and Zihan Tan\thanks{Computer Science Department, University of Chicago. Email: {\tt zihantan@uchicago.edu}. Part of the work was done while the author was visiting  the Department of Computer Science and Applied Mathematics, Weizmann Institute of Science. Supported in part by NSF grant CCF-1616584.} 
}
\maketitle

\thispagestyle{empty}
\input{abstract}
\end{titlepage}

\pagenumbering{gobble}
\tableofcontents
\newpage 
\pagenumbering{arabic}

\section{Introduction}
\input{intro}
\subsection{Our Results}
Our graph-theoretic results consider two main settings: in the first setting, the input graph is $k$-edge connected, and has diameter at most $D$; in the second setting, the input graph is $(k,D)$-connected. 
We only consider unweighted graphs, that is, all edge lengths are unit. Graphs are allowed to have parallel edges, unless we explicitly state that the graph is simple.
Throughout the paper, we use the term \emph{efficient algorithm} to refer to a sequential algorithm whose running time is polynomial in its input size.

\paragraph{Packing Trees into Low-Diameter Graphs.}
We prove the following two theorems that allow us to pack low-diameter trees into low-diameter graphs. 

\begin{theorem}\label{thm:small-depth-congestion-two}
	There is an efficient randomized algorithm, that, given any positive integers $D,n,k$, and an $n$-vertex $k$-edge-connected graph $G$ of diameter at most $D$, computes a collection $\tset'=\{T'_1,\ldots,T'_{\floor{k/2}}\}$ of $\floor{k/2}$ spanning trees of $G$, such that  each edge of $G$ appears in at most two of the trees in $\tset'$, and,  with high probability, each tree $T'_i\in \tset'$ has diameter $O((101k\ln n)^{D})$.
\end{theorem}

As we show later, the diameter bound of Theorem~\ref{thm:small-depth-congestion-two} is close to the best possible. 
Unfortunately, the trees in the packing provided by Theorem \ref{thm:small-depth-congestion-two} may share edges.
Next, we generalize the classical result of Karger \cite{karger1999random}  to obtain a packing of completely edge-disjoint trees of small diameter, in the following theorem.

\begin{theorem}
	\label{thm:Karger_diameter-main}
	There is an efficient randomized algorithm that, given an $n$-vertex $k$-edge-connected graph $G$ of diameter at most $D$, such that $k>1000\ln n$, computes a collection $\{T_1,\ldots,T_{r}\}$ of $r=\Omega(k/\ln n)$ edge-disjoint spanning trees of $G$, such that with probability $1-1/\poly(n)$, each resulting tree $T_i$ has diameter $O(k^{D(D+1)/2})$.
\end{theorem}
We note that while the diameter bound in Theorem \ref{thm:Karger_diameter-main} is slightly weaker than that obtained in Theorem \ref{thm:small-depth-congestion-two}, and the number of the spanning trees is somewhat lower, its advantage is that the resulting trees are guaranteed to be edge-disjoint. Moreover, the algorithm in Theorem \ref{thm:Karger_diameter-main} is very simple: we construct $r$ graphs $G_1,\ldots,G_r$ with $V(G_i)=V(G)$ for all $i$, by sampling every edge of $G$ into one of these graphs independently. We then compute a spanning tree $T_i$ in each such graph $G_i$, and show that its diameter is suitably bounded. As such, this algorithm is easy to use in the distributed setting.

Lastly, we show that our upper bounds are close to the best possible if $k\gg D$, by proving the following lower bound.

\begin{theorem}
	\label{thm:diameterlowerbound}
	For all positive integers $n,k,D,\eta,\alpha$ such that $k/(4D\alpha \eta)$ is an integer and $n\geq 3k\cdot \left (\frac{k}{2D\alpha\eta}\right)^D$, there exists a $k$-edge connected simple graph $G$ on $n$ vertices of diameter at most $2D+2$, such that, for any collection $\tset=\{T_1,\ldots,T_{k/\alpha}\}$ of $k/\alpha$ spanning trees of $G$ that causes edge-congestion at most $\eta$, some tree $T_i\in \tset$ has diameter at least $\frac 1 4 \cdot \left(\frac{k}{2D\alpha\eta} \right)^D$.
\end{theorem}

Note that, in particular, any collection $\tset$ of $\Omega(k)$ trees that are either edge-disjoint, or cause a constant edge-congestion, must contain a tree of diameter $\Omega\left(\left(\frac{k}{cD}\right )^D\right )$ for some constant $c$. Even if we are willing to allow a polylogarithmic edge-congestion, and to settle for $\Theta(k/\poly\log n)$ trees, at least one of the trees must have diameter $\Omega\left(\left(\frac{k}{D\poly\log n}\right )^D\right )$. 
Moreover, we show that the lower bound from Theorem \ref{thm:diameterlowerbound} continues to hold even for the weaker notion of \emph{edge-independent} trees\footnote{A collection $\tset$ of spanning trees is edge-independent, iff all trees in $\tset$ are rooted at the same vertex $v^*$, and for every vertex $v\in V(G)$, if we denote by $\pset(v)$ the collection of paths that contains, for each tree $T\in \tset$, the unique path connecting $v$ to $v^*$ in $T$, then all paths in $\pset(v)$ are edge-disjoint.
}, introduced in \cite{itai1988multi}.

\paragraph{Packing Trees into $(k,D)$-connected Graphs.}

We next consider $(k,D)$-connected graphs and show an algorithm that computes a tree packing, that is near-optimal in both the number of trees and in the diameter.
\begin{theorem}
	\label{thm:packing_spanning_trees_in_(k,D)_with_congestion}
	There is an efficient randomized algorithm, that, given any positive integers $D,k,n$ with $k\leq n$, and a $(k,D)$-connected $n$-vertex graph $G$,  computes a collection $\tset=\{T_1,\ldots,T_{k}\}$ of $k$ spanning trees of $G$, such that, for each $1\le \ell\le k$,  tree $T_{\ell}$ has diameter at most $O(D\log n)$, and with probability at least $1-1/\poly(n)$, each edge of $G$ appears in $O(\log n)$ trees of $\tset$.
\end{theorem}

\paragraph{Improved Distributed Algorithms for Highly Connected Graphs.}
We present several applications of low-diameter tree packing in the standard \congest\ model of distributed computation \cite{Peleg:2000}. By the proof of Theorem \ref{thm:Karger_diameter-main} and the $O(\log n)$-approximation algorithm for edge connectivity by \cite{ghaffari2013distributed}, we obtain the following result.

\begin{theorem}\label{lem:lambda}
	There is a randomized distributed algorithm, that, given an $n$-vertex graph $G$ of constant diameter $D=O(1)$ and an integer $\lambda$, with high probability solves the problem of $O(\log n)$-approximate verification of $\lambda$-edge connectivity in $G$ in $\poly(\lambda\cdot \log n)$ rounds. 
\end{theorem}
This improves upon the state of the art bound of $O(\sqrt{n})$ for graphs with constant diameter $D \geq 3$, and $\lambda\leq n^c$ for some positive constant $c< 1/(2D^2)$. 
From now on, we restrict our attention to $k$-edge connected graphs with a constant diameter $D=O(1)$. 
We employ the modular approach for distributed optimization introduced by Ghaffari and Haeupler in \cite{ghaffari2016distributed} which is based on the notion of \emph{low-congestion shortcuts}. Roughly speaking, these shortcuts augment vertex-disjoint connected subgraphs by adding nearly-edge disjoint subsets of ``shortcut'' edges (that is, edges that reduce the diameter of each subgraph). 
Using our tree packing construction, we provide improved shortcuts for highly connected graphs of small diameter. This immediately leads  to  $o(\sqrt{n})$-round algorithms for several classical graph problems. For example, we prove the following:
\begin{theorem}\label{thm:mstdist}
	There is a randomized distributed algorithm, that, given a $k$-edge connected weighted $n$-vertex graph $G$ of diameter $D$, such that the nodes know an $O(\log n)$ approximation of $k$, computes an MST of $G$ in $\widetilde{O}(\min\{\sqrt{n/k}+n^{D/(2D+1)},n/k\})$ rounds with high probability. 
\end{theorem}
If the nodes do not know an $O(\log n)$-approximation of the value of $k$, then such an approximation can be computed in $\poly(k \log n)$ rounds for $D=O(1)$ using Theorem \ref{lem:lambda}, w.h.p. For general graphs (of an arbitrary connectivity) with diameter $D=3,4$, Kitamura et al. \cite{kitamura2019low} showed nearly optimal constructions of MST's (based on shortcuts) with round complexities of $\widetilde{O}(n^{1/4})$ and  $\widetilde{O}(n^{1/3})$ respectively.
Turning to lower bounds, we slightly modify the construction of Lotker et al. \cite{LotkerPP06} to obtain a lower bound of $\Omega((n/k)^{1/3})$ rounds for computing an MST in $k$-edge connected graphs of diameter $4$, assuming that $k=O(n^{1/4})$.

Finally, we consider the basic task of \emph{information dissemination}, where a given source vertex $s$ is required to send $N$ bits of information to the designated target vertex $t$ in a $k$-edge connected $n$-vertex graph. This problem was first addressed in \cite{ghaffari2013distributed}, who showed a lower bound of $\Omega(\min\{N/\log^2 n,n/k\})$ rounds, provided that the diameter of the graph is $\Theta(\log n)$. Using our low-diameter tree packing we obtain the first improved upper bounds for sublogarithmic diameter. We also show a new lower bound for simple store-and-forward algorithms, for the regime where $D=o(\log n)$. 
\begin{theorem}\label{lem:dist-inf-diss}
There is a randomized distributed algorithm, that, given any $k$-edge connected $n$-vertex graph $G$ of diameter $D$ with a source vertex $s$ and a destination vertex $t$, sends an input sequence of $N$ bits from $s$ to $t$. The number of rounds is bounded by $\widetilde{O}(N^{1-1/(D+1)}+N/k)$ with high probability.  In addition, for all integers $n,N,D$ and $k\leq n$, there exists a $k$-edge connected $n$-vertex graph $G=(V,E)$ of diameter $2D$, and a pair $s,t$ of its vertices, such that sending $N$ bits from $s$ to $t$ in a store-and-forward manner requires at least $\Omega(\min\{(N/(D \log n))^{1-1/(D+1)},n/k\}+N/k+D)$ rounds.
\end{theorem}
\paragraph{Applications to Secure Distributed Computation.} Recently, Parter and Yogev \cite{parter2019low} presented a general simulation result that converts any non-secure distributed algorithm to an equivalent secure algorithm, while paying a small overhead in the number of rounds. This transformation is based on the combinatorial graph structure of low-congestion cycle cover, namely, a collection of nearly edge-disjoint short cycles that cover all edges in the graph. The security provided by \cite{parter2019low} was limited to adversaries who can manipulate at most one edge of the graph in a given round; in fact if the graph is only $2$-edge connected, no stronger security guarantees, in terms of the number of edges that an adversary is allowed to corrupt is possible. 
In this paper we provide technical tools for handling stronger adversaries, who collude with $f(k)$ edges in a $k$-edge connected graph in each {given} round. In order to do so, we define a stronger variant of cycle cover that is adapted to the highly connected setting. This generalization is formalized by the notion of $k$-connected cycle cover, in which each edge in the graph is covered by $k$ almost-disjoint cycles.
Our key contribution is an algorithm that transforms any tree packing with $k$ trees of diameter $D$ into a $(k-1)$-connected cycle cover with cycle length $O(D\log n)$ and congestion $\widetilde{O}(k\log n)$. This yields a simple secure simulation of distributed algorithms in the presence of an adversary who colludes with $O(k/\log n)$ edges of the graph in each round\footnote{We note that an adversary may choose a different set of $O(k/\log n)$ edges to listen to or to corrupt in each round.}. 
Finally, we also use low-diameter tree packing to provide a simple store-and-forward algorithm for the problem of secure broadcast.

\subsection{Open Problems}\label{sec: open}
For brevity, let us say that a collection $\tset$ of spanning trees of a $(k,D)$-connected graph $G$ is an $(\alpha,D')$-packing iff $|\tset|\geq k/\alpha$ and the diameter of every tree in $\tset$ is at most $D'$. A major remaining open question is: for which values of $\alpha$ and $D'$ can we guarantee the existence of an $(\alpha,D')$-packing $\tset$ of edge-disjoint spanning tree in every $(k,D)$-connected graph. In particular, is the following statement true: every $(k,D)$-connected graph $G$ contains a collection of $\Omega(k/\poly\log n)$ edge-disjoint trees of diameter $O(D\cdot \poly\log n)$ each. The only upper bounds that we have are the ones guaranteed by Theorem \ref{thm:Karger_diameter-main}, and we do not have any lower bounds. We also do not have any upper bounds, except for those guaranteed by Theorem \ref{thm:small-depth-congestion-two}, if we allow a constant, or more generally any sub-logarithmic congestion. Additionally, obtaining an analogue of the algorithm from Theorem \ref{thm:packing_spanning_trees_in_(k,D)_with_congestion} in the distributed setting remains a very interesting open question. 

Finally, most of our results are mainly meaningful for the setting where $k=\Omega(\log n)$. It will be very interesting to consider the case of small connectivity $k=O(1)$. One can show that any $k$-edge connected graph with $k=O(1)$ of diameter $D$ is a $(k,\poly(D))$-connected graph. Is it possible to show that any $k$-edge-connected graph of diameter $D$, for some constant $k\geq 3$, has at least two edge-disjoint trees of depth at most $\poly(D)$?

\paragraph{Organization.} 
We start with preliminaries in Section~\ref{sec:prelim}.
We provide the proof of Theorem \ref{thm:small-depth-congestion-two}  in Section \ref{sec: overview-packing-cong2}, the proof of Theorem \ref{thm:Karger_diameter-main} in Section \ref{sec: proof of second main thm}, the proof of Theorem \ref{thm:diameterlowerbound} in Section \ref{sec: lower bound proof sketch}, and the proof of Theorem \ref{thm:packing_spanning_trees_in_(k,D)_with_congestion} in Section \ref{sec: kD connected packing sketch}. We discuss applications of our graph theoretic results to distributed computation in Section \ref{sec:dist}.


\input{prelim}

\section{Low-Diameter Tree Packing with Small Edge-Congestion: Proof of Theorem \ref{thm:small-depth-congestion-two}}\label{sec: overview-packing-cong2}

In this section we provide the proof of Theorem \ref{thm:small-depth-congestion-two}.

We start by showing that, if we are given a graph $G$, and a collection $\set{T_1,\ldots,T_k}$ of edge-disjoint spanning trees of $G$, such that the diameter of the tree $T_k$ is at most $2D$ (but other trees may have arbitrary diameters), then we can efficiently compute another collection $\set{T'_1,\ldots,T'_{k-1}}$ of edge-disjoint spanning trees of $G$, such that the diameter of each resulting tree $T'_i$ is bounded by $O((101k\ln n)^{D})$ with high probability.
\begin{theorem}
	\label{thm:diameter_fixing}
	There is an efficient randomized algorithm, that, given any positive integers $D,k,n$, an $n$-vertex graph $G$, and a collection $\set{T_1,\ldots,T_k}$ of $k$ spanning trees  of $G$, such that the trees $T_1,\ldots,T_{k-1}$ are edge-disjoint, and the diameter of $T_k$ is at most $2D$, computes a collection $\{T'_1,\ldots,T'_{k-1}\}$ of edge-disjoint spanning trees of $G$, such that, with probability at least $1-1/\poly(n)$, for each $1\le i\le k-1$, the diameter of tree $T'_i$ is bounded by  $O((101k\ln n)^{D})$. 
\end{theorem}

Theorem \ref{thm:small-depth-congestion-two} easily follows by combining Theorem \ref{thm:diameter_fixing} with the results of
Kaiser~\cite{kaiser2012short}, who gave a short elementary proof of the tree-packing theorem of Tutte~\cite{tutte1961problem} and Nash-Williams~\cite{nash1961edge}. His proof directly translates into an efficient
algorithm, that, given a $k$-edge connected graph $G$, computes a collection of $\floor{k/2}$ edge-disjoint spanning trees of $G$.
In order to complete the proof of Theorem \ref{thm:small-depth-congestion-two}, we use the algorithm of Kaiser~\cite{kaiser2012short} to compute an arbitrary collection $\tset=\set{T_1,\ldots,T_{\floor{k/2}}}$ of edge-disjoint spanning trees of $G$, and compute another arbitrary BFS tree $T^*$ of $G$. Since the diameter of $G$ is at most $D$, the diameter of $T^*$ is at most $2D$. We then apply Theorem~\ref{thm:diameter_fixing} to the collection $\set{T_1,\ldots,T_{\floor{k/2}},T^*}$ of spanning trees, to obtain another collection $\tset'=\set{T'_1,\ldots,T'_{\floor{k/2}}}$ of spanning trees, such that each edge of $G$ belongs to at most $2$ trees of $\tset'$, and with high probability, the diameter of each tree in $\tset'$ is at most $O((101k\ln n)^{D})$.
We note that, since we allow parallel edges, the trees in the set $\set{T_1,\ldots,T_{\floor{k/2}},T^*}$ are edge-disjoint in graph $G\cup E(T^*)$.

The main technical tool that we use in order to prove of Theorem \ref{thm:diameter_fixing} is the following theorem, that allows one to ``fix'' a diameter of a connected graph using a low-diameter tree.

\begin{theorem}
	\label{thm:random_tree_planting}
	Let $H$ be a connected graph with $|V(H)|\leq n$, and let $T$ be a rooted tree of depth $D$, such that $V(T)=V(H)$. For a real number $0<p<1$, let $R$ be a random subset of the edges of $T$, where each edge $e\in E(T)$ is added to $R$ independently with probability $p$. Then with probability at least $1-\frac{D}{n^{48}}$, the diameter of the graph $H\cup R$ is at most $(\frac{101\ln n}{p})^{D}$.
\end{theorem}

Theorem \ref{thm:diameter_fixing} easily follows from Theorem \ref{thm:random_tree_planting}: For each $1\leq i< k$, we construct a graph $G_i$ as follows. Start with $G_i=T_i$ for all $1\le i\le k$. Compute a random partition $E_1,\ldots,E_{k-1}$ of the edges of $E(T_k)$, by adding each edge $e\in E(T_k)$ to a set $E_i$ chosen uniformly at random from $\set{E_1,\ldots,E_{k-1}}$ independently from other edges. Using Theorem \ref{thm:random_tree_planting} with $p=1/(k-1)$, it is immediate to see that with high probability, the diameter of each resulting graph $G_i$ is bounded by $O((101k\ln n)^{D})$. We then let $T'_i$ be a BFS tree of graph $G_i$, rooted at an arbitrary vertex.
In order to complete the proof of Theorem \ref{thm:small-depth-congestion-two}, it is now enough to prove Theorem \ref{thm:random_tree_planting}. 

\begin{proofof}{Theorem \ref{thm:random_tree_planting}}
	Recall that we are given a connected graph $H$ with $|V(H)|\leq n$, and a rooted tree $T$ of depth $D$, such that $V(T)=V(H)$, together with a parameter $0<p<1$. We let $R$ be a random subset of $E(T)$, where each edge $e\in E(T)$ is added to $R$ independently with probability $p$. Our goal is to show that the diameter of the graph $H\cup R$ is at most $\left(\frac{101\ln n}{p}\right)^D$ with probability at least $1-\frac{D}{n^{48}}$.
	Denote $V=V(H)=V(T)$.
	For each $0\le i\le D$, let $V_i$ be the set of nodes lying at level $i$ of the tree $T$ (that is, at distance $i$ from the tree root), and denote $V_{\le i}=\bigcup_{t=0}^{i}V_t$. Let $H'=H\cup R$. 
	
	We say that a node $x\in V$ is \emph{good} if either (i) $x\in V_{\le D-1}$; or (ii) $x\in V_{D}$, and there is an edge in $R$ connecting $x$ to a node in $V_{D-1}$. 
	We assume that $V=\set{v_1,\ldots,v_{n'}}$, where the vertices are indexed in an arbitrary order. Given an ordered pair $(x,x')$ of 
	vertices in $H$, and a path $P$ connecting $x$ to $x'$, let $\sigma(P)$ be a sequence of vertices that lists all the vertices appearing on $P$ in their natural order, starting from vertex $x$ (so in a sense, we think of $P$ as a directed path). For an ordered pair $(x,x')\in V$ of vertices, let $P_{x,x'}$ be shortest path connecting $x$ to $x'$ in $H$, and among all such paths $P$, choose the one whose sequence $\sigma(P)$ is smallest lexicographically. 
	Observe that $P_{x,x'}$ is unique, and, moreover, if some pair $u,u'$ of vertices lie on $P_{x,x'}$, with $u$ lying closer to $x$ than $u'$ on $P_{x,x'}$, then the sub-path of $P_{x,x'}$ from $u$ to $u'$ is precisely $P_{u,u'}$. 
	
	Let $M=\frac{50\ln n}{p}$.
	For a pair $x,x'$ of vertices of $V$, we let $B(x,x')$ be the bad event that length of $P_{x,x'}$ is greater than $M$ and there is no good internal node on $P_{x,x'}$. Notice that event $B(x,x')$ may only happen if every inner vertex on $P_{x,x'}$ lies in $V_D$, and for each such vertex, the unique edge of $T$ that is incident to it was not added to $R$. Therefore, the probability that event $B(x,x')$ happens for a fixed pair $x,x'$ of vertices is at most $(1-p)^M=(1-p)^{(50\ln n)/p}\leq n^{-50}$.
	Let $B$ be the bad event that $B(x,x')$ happens for some pair $x,x'\in V$ of nodes. From the union bound over all pairs of nodes in $V$, the probability of $B$ is bounded by $n^{-48}$.
	
	Recall that $H$ is a subgraph of $H'$ and $\dist_H(\cdot,\cdot)$ is the shortest-path distance metric on $H$. We use the following immediate observation.
	
	\begin{observation}
		\label{obs:good_node_is_dense}
		If the event $B$ does not happen, then for every node $x\in V$, there is a good node $x'\in V$ such that $\dist_{H}(x,x')\le M$.
	\end{observation} 

	We prove Theorem~\ref{thm:random_tree_planting} by induction on $D$.
	The base of the induction is when $D=1$. In this case, $T$ is a star graph. Let $c$ denote the vertex that serves as the center of the star.
	For any pair $x_1,x_2\in V$ of vertices, we denote by $x'_1$ the good node that is closest to $x_1$  in $H$, and we define $x'_2$ similarly for $x_2$. 
	Notice that, from the definition of good vertices, either $x'_1=c$, or it is connected to $c$ by an edge of $R$, and the same holds for $x'_2$.
	Therefore, $\dist_{H'}(x'_1,x'_2)\le 2$ must hold.  If the event $B$ does not happen, then, since $H$ is a subgraph of $H'$, $\dist_{H'}(x_1,x_2)\le \dist_{H'}(x_1,x'_1)+\dist_{H'}(x'_1,x'_2)+\dist_{H'}(x_2,x'_2)\le \dist_{H}(x_1,x'_1)+\dist_{H'}(x'_1,x'_2)+\dist_{H}(x_2,x'_2)\le 2M+2\le \frac{101\ln n}{p}$. Therefore, with probability at least $1-n^{-48}$, $\dist_{H'}(x_1,x_2)\le \frac{101\ln n}{p}$. 
	
	Assume now that Theorem~\ref{thm:random_tree_planting} holds for every connected graph $H$ and every tree $T$ of depth at most $D-1$, with $V(T)=V(H)$. Consider now some connected graph $H$, and a rooted tree $T$ of depth $D$, with $V(T)=V(H)$. We partition the edges of $E(T)$ into two subsets: set $E_1$ contains all edges incident to the vertices of $V_D$, and set $E_2$ contains all remaining edges. Let $E'_1=E_1\cap R$, and let $E'_2=E_2\cap R$. Notice that the definition of good vertices only depends on the edges of $E'_1$, and so the event $B$ only depends on the random choices made in selecting the edges of $E'_1$, and is independent from the random choices made in selecting the edges of $E'_2$.

	Let $L$ be a subgraph of $H'$, obtained by starting with $L=H$, and then adding all edges of $E'_1$ to the graph. Finally, 
	we define a new graph $\hat{H}$, whose vertex set is $V_{\le D-1}$, and there is an edge between a pair of nodes $w,w'$ in $\hat{H}$ iff the distance between $w$ and $w'$ in $L$ is at most $M+2$.
	We also let $\hat T$ be the tree obtained from $T$, by discarding from it all vertices of $V_D$ and all edges incident to vertices of $V_D$. Observe that $V(\hat H)=V(\hat T)=V_{\leq D-1}$. The idea is to use the induction hypothesis on the graph $\hat H$, together with the tree $\hat T$. In order to do so, we need to prove that $\hat H$ is a connected graph, which we do next.

	\begin{observation}
		If the event $B$ does not happen, then graph $\hat{H}$ is connected.
	\end{observation}
	\begin{proof}
		Assume that the event $B$ does not happen, and assume for contradiction that graph $\hat H$ is not connected. Let $\cset=\set{C_1,\ldots,C_r}$ be the set of all connected components of graph $\hat H$. For every pair $C_i,C_j$ of distinct components of $\cset$, consider the set $\pset_{i,j}=\set{P_{x,x'}\mid x\in V(C_i),x'\in V(C_j)}$ of paths (recall that $P_{x,x'}$ is the shortest path connecting $x$ to $x'$ in $H$ with $\sigma(P_{x,x'})$ lexicographically smallest among all such paths). We let $P_{i,j}$ be a shortest path in $\pset_{i,j}$. Choose two distinct components $C_i,C_j\in \cset$, whose path $P_{i,j}$ has the shortest length, breaking ties arbitrarily. Assume that $P_{i,j}$ connects a vertex $v\in C_i$ to a vertex $u\in C_j$, so $P_{i,j}=P_{v,u}$. Recall that $H\subseteq L$, and so the path $P_{i,j}$ is contained in graph $L$. Since we did not add edge $(u,v)$ to $\hat H$, the length of $P_{i,j}$ is greater than $M+2$. Since we have assumed that event $B$ does not happen, there is at least one good inner vertex on path $P_{i,j}$. Let $X$ be the set of all good vertices that serve as inner vertices of $P_{i,j}$.
		
		We first show that for each $x\in X$, $x\not \in V(\hat H)$ must hold. Indeed, assume for contradiction that $x\in V(\hat H)$, so $x$ belongs to some connected component of $V(\hat H)$. Assume first that $x\in V(C_i)$. Recall that the sub-path of $P_{i,j}$ from $x$ to $u$ is precisely $P_{x,u}$, so this path lies in $\pset_{i,j}$. But its length is less than the length of $P_{i,j}$, contradicting the choice of $P_{i,j}$. Otherwise, $x$ belongs to some connected component $C_{\ell}$ of $\cset$ with $\ell\neq i$. The sub-path of $P_{i,j}$ from $v$ to $x$ is precisely $P_{v,x}$, so this path must lie in $\pset_{i,\ell}$. Since its length is less than the length of $P_{i,j}$, this contradicts the choice of the components $C_i,C_j$. We conclude that $x\not\in V(\hat H)$.
		
		
		Since $V(\hat H)$ contains all vertices of $V_{\leq D-1}$, and every vertex in $X$ is a good vertex, it must be the case that $X\subseteq V_D$. Consider again some vertex $x\in X$. Since $x$ is a good vertex and $x\in V_D$, there must be an edge $e_x=(x,x')\in E'_1$, connecting $x$ to some vertex $x'\in V_{\leq D-1}$. In particular, $x'$ must belong to some connected component of $\cset$, and the edge $e_x$ lies in graph $L$. Assume that $X=\set{x_1,x_2,\ldots,x_q}$, where the vertices are indexed in the order of their appearance on $P_{i,j}$, from $v$ to $u$. Consider the sequence $\tilde \sigma=(v,x'_1,x'_2,\ldots,x'_q,u)$ of vertices. All these vertices belong to $V(\hat H)$, and $v\in C_i$, while $u\in C_j$. For convenience, denote $v=x'_0=x_0$ and $u=x'_{q+1}=x_{q+1}$. Then there must be an index $1\leq a\leq q$, such that $x'_a$ and $x'_{a+1}$ belong to distinct connected components of $\cset$. Note that the sub-path of $P_{i,j}$ between $x_a$ and $x_{a+1}$ is precisely $P_{x_a,x_{a+1}}$ -- the shortest path connecting $x_a$ to $x_{a+1}$ in $H$. Since no good vertices lie between $x_a$ and $x_{a+1}$ on this path, and since we have assumed that event $B$ does not happen, the length of this path is at most $M$. Therefore, there is a path in graph $L$, connecting $x'_a$ to $x'_{a+1}$, whose length is at most $M+2$. This path connects a pair of vertices that belong to different connected components of $\hat H$, contradicting the construction of $\hat H$.
	\end{proof}

Consider now the tree $\hat T$ and the graph $\hat H$. Recall that $\hat T$ is a rooted tree of depth $D-1$, $V(\hat T)=V(\hat H)$, $|V(\hat{H})|\le |V(H)|\le n$, and, assuming the event $B$ did not happen, $\hat H$ is a connected graph. Moreover, set $E'_2$ of edges is a subset of $E(\hat T)=E_2$, obtained by adding every edge of $E(\hat T)$ to $E'_2$ with probability $p$, independently from other edges. Therefore, assuming that event $B$ did not happen, we can use the induction hypothesis on the graph $\hat H$, the tree $\hat T$, and the set $E'_2$ of edges as $R$. Let $B'$ be the bad event that the diameter of $\hat H\cup E'_2$ is greater than $(\frac{101\ln n}{p})^{D-1}$. Note that the event $B'$ only depends on the random choices made in selecting the edges of $E'_2$.
From the induction hypothesis, the probability that $B'$ happens is at most $\frac{D-1}{n^{48}}$.


Lastly, we show that, if neither of the events $B,B'$ happens, then  $\diam(H')\le (\frac{101\ln n}{p})^{D}$.
\begin{observation}
	If neither of the events $B,B'$ happens, then $\diam(H')\le (\frac{101\ln n}{p})^{D}$.	
\end{observation}
\begin{proof}
	Consider any pair $x_1,x_2\in V$ of vertices. It is sufficient to show that, if events $B,B'$ do not happen, then $\dist_{H'}(x_1,x_2)\le (\frac{101\ln n}{p})^{D}$.
	
	Let $x'_1$ be a good node in $V(H)$ that is closest to $x_1$, and define $x'_2$ similarly for $x_2$. From Observation~\ref{obs:good_node_is_dense}, $\dist_{H}(x_1,x'_1)\le M$. If $x'_1\in V_{\le D-1}$, then we define $x''_1=x'_1$, otherwise we let $x''_1$ be the node of $ V_{D-1}$ that is connected to $x'_1$ by an edge of $E'_1$, and we define $x''_2$ similarly for $x_2$.
	Therefore, $x''_1,x''_2\in V_{\le D-1}=V(\hat{H})$, and, assuming event $B$ does not happen,
	$\dist_{H'}(x_1,x''_1)\le M+1$, and $\dist_{H'}(x_2,x''_2)\le M+1$.
	Since we have assumed that the bad event $B'$ does not happen, $\dist_{\hat H\cup E_2'}(x''_1,x''_2)\le (\frac{101\ln n}{p})^{D-1}$. 
	Recall that for every edge $e=(u,v)\in \hat H\cup E_2'$, if $e\in E_2'$ then $e\in E(H')$; otherwise, $e\in E(\hat H)$, and there is a path in graph $H\cup E_1'$ of length at most $M+2$ connecting $u$ to $v$ in $H$. Therefore, $\dist_{H'}(x_1'',x_2'')\leq (M+2)\cdot \dist_{\hat H}(x_1'',x_2'')\leq  (\frac{101\ln n}{p})^{D-1}\cdot (M+2)$. 
	
	Altogether, since $M=(50\ln n)/p$,
	\[\begin{split}
	\dist_{H'}(x_1,x_2)&\le  \dist_{H'}(x_1,x''_1)+\dist_{H'}(x''_1,x''_2)+\dist_{H'}(x_2,x''_2)\\
	&\le \left(\frac{101\ln n}{p}\right )^{D-1}\cdot (M+2)+(2M+2)\\
	&\le \left (\frac{101\ln n}{p}\right )^{D}.
	\end{split}\]
\end{proof}

The probability that either $B$ or $B'$ happen is bounded by $\frac{D}{n^{48}}$. Therefore, with probability at least $1-\frac{D}{n^{48}}$, neither of the events happens, and $\diam(H')\le (\frac{101\ln n}{p})^{D}$.
This concludes the proof of Theorem~\ref{thm:random_tree_planting}.
\end{proofof}

\section{Low-Diameter Packing of Edge-Disjoint Trees: Proof of Theorem \ref{thm:Karger_diameter-main}}
\label{sec: proof of second main thm}

In this section we provide the proof of Theorem \ref{thm:Karger_diameter-main}.
The main tool in the proof of Theorem \ref{thm:Karger_diameter-main} is the following theorem.

\begin{theorem}
\label{thm:Karger_diameter}
Let $k,D,n$ be any positive integers with $k>1000\ln n$, let $\frac{707\ln n}{k}\le p \le 1$ be a real number, and let $G$ be an $n$-vertex $k$-edge-connected graph of diameter $D$. Let $G'$ be a sub-graph of $G$ with $V(G')=V(G)$, where every edge $e\in E(G)$ is added to $G'$ with probability $p$ independently from other edges. Then, with probability at least $1-1/\poly(n)$, $G'$ is a connected graph, and its diameter is bounded by $k^{D(D+1)/2}$.
\end{theorem}

Karger \cite{karger1999random} has shown that, if $G$ is a $k$-connected graph, and $G'$ is obtained by sub-sampling the edges of $G$ with probability $\Omega(\log n/k)$, then $G'$ is a connected graph with high probability. Theorem~\ref{thm:Karger_diameter} further shows that the diameter of $G'$ is with high probability bounded by $k^{D(D+1)/2}$, where $D$ is the diameter of $G$. 

Theorem \ref{thm:Karger_diameter-main} easily follows from Theorem~\ref{thm:Karger_diameter}: 
Let $r=\lfloor k/(707\ln n)\rfloor$. We partition $E(G)$ into subsets $E_1,\ldots,E_{r}$ by choosing, for each edge $e\in E(G)$, an index $i$ independently and uniformly at random from $\{1,2,\ldots,r\}$ and then adding $e$ to $E_i$.
For each $1\le i\le r$, we define a graph $G_i$ by setting $V(G_i)=V(G)$ and $E(G_i)=E_i$. 
Finally, for each graph $G_i$, we compute an arbitrary BFS tree $T_i$, and return the resulting collection $\tset=\set{T_1,\ldots,T_r}$ of trees. It is immediate to verify that the graphs $G_1,\ldots,G_r$ are edge-disjoint, and so are the trees of $\tset$. Moreover, applying Theorem \ref{thm:Karger_diameter} to each graph $G_i$ with $p=1/r$, we get that with probability $1-1/\poly(n)$, $\diam(T_i)\leq 2\diam(G_i)\leq O(k^{D(D+1)/2})$. Using the union bound over all $1\leq i\leq r$ completes the proof of Theorem  \ref{thm:Karger_diameter-main}. It now remains to prove Theorem~\ref{thm:Karger_diameter}. 

\input{packing_trees_using_diameter_fixing.tex}

\section{Lower Bound: Proof of  Theorem~\ref{thm:diameterlowerbound}}
\label{sec: lower bound proof sketch}

In this section we provide the proof of Theorem \ref{thm:diameterlowerbound}.
We start by proving the following slightly weaker theorem; we then extend it to obtain the proof of Theorem \ref{thm:diameterlowerbound}.
%

\begin{theorem}
\label{thm:diameterlowerbound-weaker}
For all positive integers $k,D,\eta,\alpha$ such that $k/(4D\alpha \eta)$ is an integer, there exists a $k$-edge connected graph $G$ with $|V(G)|=O\left(\left (\frac{k}{2D\alpha\eta}\right)^D\right)$ and diameter at most $2D$, such that, for any collection $\tset=\{T_1,\ldots,T_{k/\alpha}\}$ of $k/\alpha$ spanning trees of $G$ that causes edge-congestion at most $\eta$, some tree $T_i\in \tset$ has diameter at least $\frac 1 4 \cdot \left(\frac{k}{2D\alpha\eta} \right)^D$.
\end{theorem}

Notice that the main difference from Theorem  \ref{thm:diameterlowerbound} is that the graph $G$ is no longer required to be simple; the number of vertices of $V(G)$ is no longer fixed to be a prescribed value; and the diameter of $G$ is $2D$ instead of $2D+2$. 

\begin{proof}
For a pair of integers $w>1,D\geq 1$, we let $T_{w,D}$ be a tree of depth $D$, such that every vertex lying at levels $0,\ldots,D-1$ of $T_{w,D}$ has exactly $w$ children. In other words, $T_{w,D}$ is the full   $w$-ary tree of depth $D$. 
We denote $N_{w,D}=|V(T_{w,D})|=1+w+w^2+\cdots+w^D\leq w^{D+1}/(w-1)$. We assume that for every inner vertex $v\in V(T_{w,D})$, we have  fixed an arbitrary ordering of the children of $v$, denoted by $a_1(v),\ldots,a_w(v)$.

A \emph{traversal} of a tree $T$ is an ordering of the vertices of $T$. 
%
A \emph{post-order traversal} on a tree $T$, $\pi(T)$, is defined as follows. If the tree consists of a single node $v$, then $\pi(T)=(v)$.
Otherwise, let $r$ be the root of the tree and consider the sequence $(a_1(r),\ldots,a_w(r))$ of its children.
For each $1\leq i\leq w$, let $T_i$ be the sub-tree of $T$ rooted at the vertex $a_i(r)$. We then let $\pi(T)$ be the concatenation of $\pi(T_1),\pi(T_2),\ldots,\pi(T_w)$, with the vertex $r$ appearing at the end of the sequence; see Figure~\ref{fig:inorder} for an illustration.
For simplicity, we assume that $V(T_{w,D})=\set{v_1,v_2,\ldots,v_{N_{w,D}}}$, where the vertices are indexed in the order of their appearance in $\pi(T_{w,D})$, so the traversal visits these vertices in this order.

Next, we define a graph $G_{w,D}$, as follows. The vertex set of $G_{w,D}$ is the same as the vertex set of $T_{w,D}$, namely $V(G_{w,D})=V(T_{w,D})$.  The edge set of $G_{w,D}$ consists of two subsets: $E_1=E(T_{w,D})$, and another set $E_2$ of edges that contains, for each $1\leq i<N_{w,D}$, $k$ parallel copies of the edge $(v_i,v_{i+1})$. We then set $E(G_{w,D})=E_1\cup E_2$. For convenience, we call the edges of $E_1$ \emph{blue edges}, and the edges of $E_2$ \emph{red edges}; see Figures~\ref{fig:inorder} and \ref{fig:additionaledge}.


\begin{figure}[!htb]
	\centering
	\begin{minipage}{0.46\textwidth}
		\centering
		\includegraphics[width=0.9\linewidth, height=0.17\textheight]{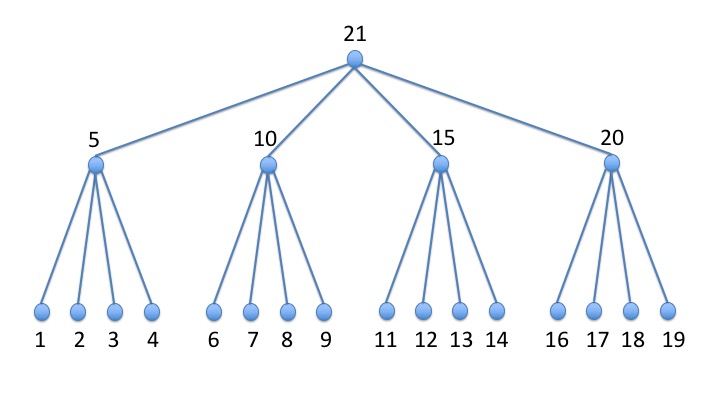}
		\caption{Tree $T_{4,2}$ with vertices indexed according to post-order traversal.}
		\label{fig:inorder}
	\end{minipage}%
	\hspace{0.2pt}
	\begin{minipage}{0.46\textwidth}
		\centering
		\includegraphics[width=0.9\linewidth, height=0.17\textheight]{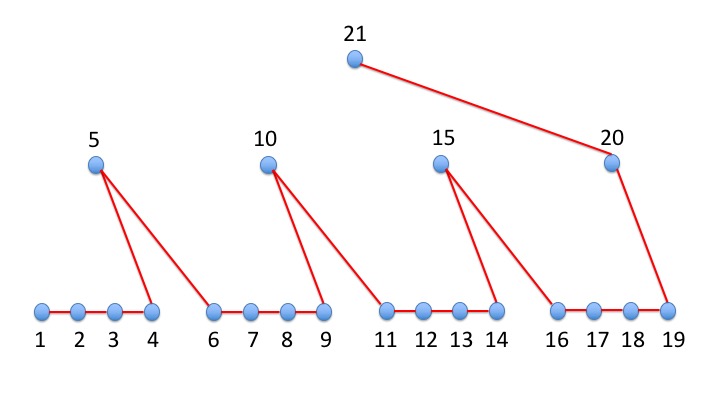}
		\caption{The edge set $E_2$ in $G_{4,2}$ (only a single copy of each edge is shown).}
		\label{fig:additionaledge}
	\end{minipage}
\end{figure}

It is easy to verify that graph $G_{w,D}$ must be $k$-edge connected, since for any partition of $V(G_{w,D})$, there is some index $1\leq i<N_{w,D}$ with $v_i,v_{i+1}$ separated by the partition, and so $k$ parallel edges connecting $v_i$ to $v_{i+1}$ must cross the partition.

We now fix an integer $w=k/(2D\alpha\eta)$ (note that $w\ge 2$), and we let $T=T_{w,D}$ be the corresponding tree and $G=G_{w,D}$ the corresponding graph. For convenience, we denote $N_{w,D}$ by $N$. Recall that $N\leq w^{D+1}/(w-1)=O\left(\left (\frac{k}{2D\alpha\eta}\right)^D\right)$.
As observed before, $G$ is $k$-edge connected. Since the depth of $T$ is $D$, and $T\subseteq G$, it is easy to see that the diameter of $G$ is at most $2D$. 

We now consider any collection $\tset=\{T_1,\ldots,T_{k/\alpha}\}$ of $k/\alpha$ spanning trees of $G$ that causes edge-congestion at most $\eta$. Our goal is to show that some tree $T_i\in \tset$ has diameter at least $\frac1 4 \cdot \left(\frac{k}{2D\alpha\eta} \right)^D$.



For convenience, we denote $V(G)=V(T)=V$. 
We say that a vertex $x\in V$ is an \emph{ancestor} of a vertex $y\in V$ if $x$ is an ancestor of $y$ in the tree $T$, that is, $x\neq y$, and $x$ lies on the unique path connecting $y$ to the root of $T$.

Let $L\subseteq V$ be the set of vertices that serve as leaves of the tree $T$. We denote by $u=v_1$ a vertex of $L$ that has the lowest index, and by $u'$ the vertex of $L$ with the largest index. It is easy to see that $u'=v_{N-D}$, as every vertex whose index is greater than that of $u'$ is an ancestor of $u'$. For each $1\le j\le k/\alpha$, we denote by $P_j$ the unique path that connects $u$ to $u'$  in tree $T_j$. Let $\pset=\set{P_j\mid 1\leq j\leq k/\alpha}$. It is enough to show that at least one of the paths $P_j$ has length at least $\frac 1 4 \cdot \left(\frac{k}{2D\alpha\eta} \right)^D$. In order to do so, we show that $\sum_{j=1}^{k/\alpha}|E(P_j)|$ is sufficiently large. At a high level, we consider the red edges $(v_i,v_{i+1})$ (the edges of $E_2$), and show that many of the paths in $\pset$ must contain copies of each such edge. This in turn will imply that $\sum_{P_j\in \pset}|E(P_j)|$ is large, and that some path in $\pset$ is long enough.

For each vertex $v_i\in L$ such that $v_i\ne u'$, we let $S_i=\set{v_1,\ldots,v_i}$, and we let $\notS_i=\set{v_{i+1},\ldots,v_N}$.
Notice that, since $u\in S_i$ and $u'\in\notS_i$,  every path in $\pset$ must contain an edge of $E_G(S_i,\notS_i)$. Note that the only red edges in $E_G(S_i,\notS_i)$ are the $k$ parallel copies of the edge $(v_i,v_{i+1})$. In the next observation, 
we  show that the number of blue edges in $E_G(S_i,\notS_i)$ is bounded by $Dw$. 



\begin{observation}
	\label{obs:cutedge}
	For each vertex $v_i\in L$ such that $v_i\ne u'$, for every blue edge $e\in E_G(S_i,\notS_i)$, at least one endpoint of $e$ must be an ancestor of $v_i$. 
\end{observation}

\begin{proof}
	We consider a natural layout of the tree $T$, where for every inner vertex $x$ of the tree, its children $a_1(x),\ldots,a_w(x)$ are drawn in this left-to-right order (see Figure \ref{fig:cutting}). Consider the path $Q$ connecting the root of $T$ to $v_i$, so every vertex on $Q$ (except for $v_i$) is an ancestor of $v_i$. All vertices lying to the left of $Q$ in the layout are visited before $v_i$ by $\pi(T)$. All vertices lying to the right of $Q$, and on $Q$ itself (excluding $v_i$) are visited after $v_i$. It is easy to see that the vertices of $Q$ separate the two sets in $T$, and so the only blue edges connecting $S_i$ to $\notS_i$ are edges incident to the vertices of $V(Q)\setminus \set{v_i}$.
\end{proof}

\begin{figure}[h]
	\centering
	\scalebox{0.3}{\includegraphics{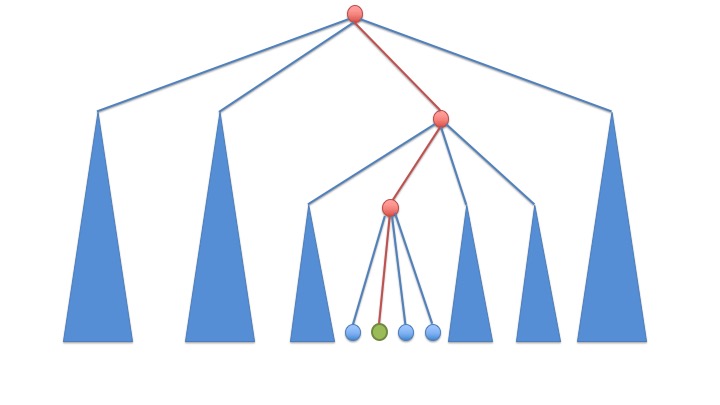}}
	\caption{A layout of the tree $T$. Vertex $v_i$ is shown in green and path $Q$ in red. All vertices lying to the left of $Q$ in this layout appear before $v_i$ in $\pi(T)$, and all vertices lying to the right of $Q$ or on $Q$ (except for $v_i$) appear after $v_i$ in $\pi(T)$.\label{fig:cutting}}
\end{figure}

Since every vertex of the tree $T$ has at most $w$ children, and since the depth of the tree is $D$, we obtain the following corollary of Observation \ref{obs:cutedge}.
\begin{corollary}\label{cor: number of cut edges}
	For each vertex $v_i\in L$ such that $v_i\ne u'$, at most $Dw$ blue edges lie in $E_G(S_i,\notS_i)$. 
\end{corollary}

Since the trees in $\tset$ cause edge-congestion $\eta$, at most $Dw\eta$ trees of $\tset$ may contain blue edges in $E_G(S_i,\notS_i)$. Each of the remaining $\frac{k}{\alpha}-Dw\eta\geq \frac{k}{2\alpha}$ trees contains a copy of the red edge $(e_i,e_{i+1})$ (recall that $w=k/(2D\alpha\eta)$.) Therefore, $\sum_{P_j\in \pset}|E(P_j)|\geq |L|\cdot \frac{k}{2\alpha}\geq \frac{Nk}{4\alpha}$,
since $|L|\geq |N|/2$.
We conclude that at least one path $P_j\in \pset$ must have length at least $\frac{Nk}{4\alpha}/\frac{k}{\alpha}\geq \frac{N}{4}$, and so the diameter of $T_j$ is at least $\frac{N}{4}$. Since $N\geq  w^D\geq\left(\frac{k}{2D\alpha\eta} \right)^D$,  the diameter of $T_j$ is at least $\frac 1 4 \cdot \left(\frac{k}{2D\alpha\eta} \right)^D$.
\end{proof}

We are now ready to complete the proof of Theorem \ref{thm:diameterlowerbound}. First, we show that we can turn the graph $G$ into a simple graph, and ensure that $|V(G)|=n$, if $n\ge 3k\cdot \left (\frac{k}{2D\alpha\eta}\right)^D$.
Let $G'_{w,D}$ be the graph obtained from $G_{w,D}$ as follows. For each $1\le i\le N$, we replace the vertex $v_i$ with a set $X_i=\{x_i^1,x_i^2\ldots,x_i^k\}$ of $k$ vertices that form a clique. For each $1\le i<N$, the $k$ red edges connecting $v_i$ to $v_{i+1}$ are replaced by the perfect matching $\{(x^t_i,x^t_{i+1})\}_{1\le t\le k}$ between vertices of $X_i$ and vertices of $X_{i+1}$. Each blue edge $(v_i,v_j)$ is replaced by a new edge $(x^1_i,x^1_j)$.
Since $n\ge 3k\cdot \left (\frac{k}{2D\alpha\eta}\right)^D>k|V(G)|+k$, we add $n-k|V(G)|>k$ new vertices that form a clique, and for each newly-added vertex, we add an edge connecting it to $x^1_{N}$ (recall that the vertex $v_{N}$ is the root of $T$). We denote $G'=G'_{w,D}$ for simplicity. It is not hard to see that $G'$ has $n$ vertices and it is $k$-edge connected. Moreover, $G'$ has diameter at most $2D+2$, since its subgraph induced by vertices of $\{x^1_i\}_{1\le i\le N}$ has diameter $2D$, and every other vertex of $G'$ is a neighbor of one of the vertices in $\{x^1_i\}_{1\le i\le N}$. The tree $T'$ is defined exactly as before, except that every original vertex $v_j$ is now replaced with its copy $x^1_j$. Let $L$ denote the set of all leaf vertices in $T'$.

Assume that we are given a collection $\tset=\{T_1,\ldots,T_{k/\alpha}\}$ of $k/\alpha$ spanning trees of $G'$ that causes edge-congestion at most $\eta$. For each $1\le i\le k/\alpha$, we denote by $Q_i$ the unique path that connects $x^1_1$ to $x^1_{N-D}$ in $T_i$ and denote $\qset=\{Q_i\mid 1\le i\le k/\alpha\}$. 
For each every leaf vertex $x_j^1\in L$, we define a cut $(W_j,\overline W_j)$ as follows: $W_j=\bigcup_{1\le s\le j}X_s$ and  $\overline{W_j}=V(G')\setminus W_j$.  Using reasoning similar to that in Corollary~\ref{cor: number of cut edges}, it is easy to see that for every leaf vertex $x_j^1\in L$,  the set $E_{G'}(W_j,\overline{W_j})$ of edges contains at most $Dw$ blue edges -- the edges of the tree $T'$. Since the trees in $\tset$ cause edge-congestion at most $\eta$, at most $Dw\eta$ trees of $\tset$ may contain blue edges in $E_{G'}(W_j,\overline{W_j})$. Therefore, for each of the remaining $\frac{k}{\alpha}-Dw\eta\ge\frac{k}{2\alpha}$ trees $T_i$, path $Q_i$ must  contain a red edge from $\{(x^t_j,x_{j+1}^{t})\}_{1\le t\le k}$.
Therefore, the sum of lengths of all paths of $\qset$ is at least $\frac{Nk}{4\alpha}$, and so at least one path $Q_i\in \qset$ must have length at least $\frac{N}{4}$. We conclude that some tree $T_i\in \tset$ has diameter at least $\frac 1 4 \cdot \left(\frac{k}{2D\alpha\eta} \right)^D$.

Lastly, we extend our results to edge-independent trees. We use the same simple graph $G'$ and the same tree $T'$ as before, setting the congestion parameter $\eta=2$.
Assume that we are given a collection $\tset'=\{T'_1,\ldots,T'_{k/\alpha}\}$ of $k/\alpha$ edge-independent spanning trees of $G'$ and let $x\in V(G')$ be their common root vertex. For each $1\le i\le k/\alpha$, we denote by $Q'_i$ the unique path that connects vertex $x^1_1$ to vertex $x^1_{N-D}$ in  tree $T'_i$, and we denote $\qset'=\{Q'_i\mid 1\le i\le k/\alpha\}$. Note that, for each $1\le i\le k/\alpha$, the path $Q'_i$ is a sub-path of the path obtained by concatenating the path $Q''_i$, connecting $x^1_1$ to $x$ in $T'_i$, with the path $Q'''_i$, connecting $x^1_{N-D}$ to $x$ in $T'_i$. Since the trees in $\tset'$ are edge-independent, the paths in $\{Q''_i\}_{1\le i\le k/\alpha}$ are edge-disjoint and so are the paths in $\{Q'''_i\}_{1\le i\le k/\alpha}$. Therefore, the paths of $\qset'$ cause edge-congestion at most $2$. 
The remainder of the proof is the same as before and is omitted here.

\section{Tree Packing for $(k,D)$-Connected Graphs: Proof of Theorem \ref{thm:packing_spanning_trees_in_(k,D)_with_congestion}}
\label{sec: kD connected packing sketch}

\input{packing_trees_in_kD.tex}

\input{applications.tex}

\newpage









\bibliographystyle{alpha}

\bibliography{REF}

\end{document}

%% file: abstract.tex
\begin{abstract}
Edge connectivity of a graph is one of the most fundamental graph-theoretic concepts.
The celebrated \emph{tree packing} theorem of Tutte and Nash-Williams from 1961 states that every $k$-edge connected graph $G$ contains a collection $\tset$ of $\lfloor k/2 \rfloor$ edge-disjoint \emph{spanning} trees, that we refer to as a \emph{tree packing}; the \emph{diameter} of the tree packing $\tset$ is the largest diameter of any tree in $\tset$. A desirable property of a tree packing, that is both sufficient and necessary for leveraging the high connectivity of a graph in distributed communication networks, is that its diameter is low. Yet, despite extensive research in this area, it is still unclear how to compute a tree packing, whose diameter is sublinear in $|V(G)|$,  in a low-diameter graph $G$, or alternatively how to show that such  a packing does not exist. 
%
%
%
%
%
In this paper, we provide first non-trivial upper and lower bounds on the diameter of tree packing.
We start by showing that, for every $k$-edge connected $n$-vertex graph $G$ of diameter $D$, there is a tree packing $\tset$ containing $\Omega(k)$ trees, of diameter $O((101k\log n)^D)$, with edge-congestion  at most $2$.
 
 Karger's edge sampling technique demonstrates that, if $G$ is a $k$-edge connected graph, and $G[p]$ is a subgraph of $G$ obtained by sampling each edge of $G$ independently with probability $p=\Theta(\log n/k)$, then with high probability $G[p]$ is connected. We extend this result to show that the diameter of $G[p]$ is bounded by $O(k^{D(D+1)/2})$ with high probability. This immediately implies that for every $k$-edge connected $n$-vertex graph $G$ of diameter $D$, there is a tree packing $\tset$ containing $\Omega(k/\log n)$ edge-disjoint trees of diameter at most $O(k^{D(D+1)/2})$ each. 

 We complement the above two results by showing that they are nearly tight: namely, that there is a $k$-edge connected graph of diameter $2D$, such that any packing of $k/\alpha$ trees with edge-congestion $\eta$ contains at least one tree of diameter $\Omega\left ((k/(2\alpha \eta D))^{D}\right )$, for any $k,\alpha$ and $\eta$.

 Lastly, we show that if, for every pair $u,v$ of vertices in a given graph $G$, there is a collection of $k$ edge-disjoint paths  connecting $u$ to $v$, of length at most $D$ each,  then we can efficiently compute a tree packing of size $k$, diameter $O(D\log n)$, and edge-congestion $O(\log n)$. 

We provide several applications of low-diameter tree packing in the settings of distributed network optimization. In particular, we show $o(\sqrt{n})$-round algorithms for problems such as MST and approximate minimum cut for graphs with $n^{\epsilon}$ edge connectivity and constant diameter. Finally, we illustrate several applications to the setting of secure distributed algorithms in which the adversary is allowed to collide with $\Omega(k/\log n)$ edges in a $k$-edge connected graph. 
\end{abstract}
%

%% file: intro.tex
 Edge connectivity of a graph is one of the most basic graph theoretic parameters, with various applications to network reliability and information dissemination. 
A key tool for leveraging high edge connectivity of a given graph is \emph{tree packing}: a large collection of spanning trees that are (nearly) edge-disjoint. 
A celebrated result of Tutte \cite{tutte1961problem} and Nash-Williams \cite{nash1961edge} shows that for every $k$-edge connected graph, there is a tree packing $\tset$ containing $\lfloor k/2 \rfloor$ edge-disjoint trees. This beautiful theorem has numerous algorithmic applications, but unfortunately it provides no guarantee on the diameter of the individual trees in $\tset$. In the worst case, trees in $\tset$ may have diameter that is as large as $\Omega(|V(T)|)$, even if the diameter of the original graph is very small. 
Given a graph $G$ and a collection $\tset$ of trees in $G$, we say that the trees in $\tset$ are \emph{edge-disjoint} iff every edge of $G$ lies in at most one tree of $\tset$, and we say that they cause \emph{edge-congestion} $\eta$ iff every edge of $G$ lies in at most $\eta$ trees of $\tset$. The \emph{diameter} of a tree-packing $\tset$ is the maximum diameter of any tree in $\tset$.

The diameter of a graph is a central graph measure that determines the round complexity of distributed algorithms for various central graph problems, including minimum spanning tree, global minimum cut, shortest $s$-$t$ path, and so on. All these problems admit a trivial lower bound of $\Omega(D)$ for the round complexity (where $D$ is the diameter of the graph), and in fact  
a stronger lower bound of $\Omega(D+\sqrt{n})$,  which is almost tight for general $n$-vertex graphs, that was shown by Das-Sarma et al.~\cite{sarma2012distributed}.
%
%
Despite attracting a significant amount of attention over the last decade (see e.g., \cite{pritchard2011fast,ghaffari2013distributed,nanongkai2014almost,Censor-HillelGK14,censor2014new,Kuhn14,ghaffari2015distributed,censor2015tight,dory2018distributed,daga2019distributed}), algorithms that exploit large edge connectivity of the input graph in the distributed setting
are quite rare. The only examples that we are aware of are  recent algorithms for minimum cut by Daga et al. \cite{daga2019distributed} and by Ghaffari et al. \cite{ghaffari2019faster}.


%
%

Censor-Hillel et al. \cite{Censor-HillelGK14} presented several distributed algorithms, that, given a $k$-edge connected $n$-vertex graph of diameter $D$, computes a fractional tree packing of $\Omega(k/\log n)$ trees that are fractionally edge-disjoint\footnote{In the fractional setting, each tree $T$ in the packing has a weight $w(T)$ and for each edge $e$, the sum of weights of all trees that contain $e$ is at most $1$.} in $\widetilde{O}(D+\sqrt{n})$ rounds. These trees have been used to parallelize the flow of information, obtaining nearly optimal \emph{throughput} for store-and-forward algorithms\footnote{In this class of algorithms, the nodes can only forward the messages they receive (e.g., network coding is not allowed).}. However, as these trees might have diameter as large as $\Omega(n)$ in the worst case, it is not clear how to use them in order to improve the \emph{round complexity} of the problem at hand, as opposed to improving the throughput. In particular, in terms of optimizing the number of communication rounds, 
it may still be preferable to send the entire information over a single BFS tree rather than spreading it over \emph{many} trees of potentially \emph{large} diameter. 
%

The problem of computing a low-diameter tree packing was studied later by Ghaffari \cite{ghaffari2015distributed} from the perspective of optimization. Specifically, he studied the multi-message broadcast problem, where a designated source vertex is required to send $k$ messages to all other nodes in the network. Denoting by $\opt(G)$ the minimum number of rounds required for the broadcast on an input graph $G$, he constructed a tree packing of size $k$, where both the diameter and the congestion are bounded by $\widetilde{O}(\opt(G))$. 
While this approach provides a nearly optimal broadcast scheme, it does not provide absolute upper bounds on the diameter of the tree packing, and moreover, the congestion caused by the tree packing can be large.

A recent work of Ghaffari and Kuhn \cite{ghaffari2013distributed}  provides the following negative result for packing low-diameter trees into a graph: they show that for any large enough $n$ and any $k\geq 1$, there is a $k$-edge-connected  $n$-vertex graph of diameter $\Theta(\log n)$, such that, in any partitioning of the graph into spanning subgraphs, all but $O(\log n)$ of the subgraphs have diameter $\Omega(n/k)$. 
In light of this result, it is natural to consider the following key question:
\begin{quote}
\emph{(1) Is it possible to compute a tree packing whose diameter is strongly \emph{sublinear} in $|V(G)|$, provided that the diameter of the input graph $G$ is \emph{sublogarithmic} in $|V(G)|$?}
\end{quote}
Our second key question aims at crystallizing the main challenge to computing low-diameter tree packing. So far, we have compared the diameter of the tree packing to the diameter of the original graph. However, as observed above, the results of \cite{ghaffari2013distributed} indicate that there may be a large gap between these two measures, even for graphs whose diameter is logarithmic in $n$. A more natural reference point is the following. We say that a graph $G$ is \emph{$(k,D)$-connected}, iff for every pair $u,v\in V(G)$ of distinct vertices, there are $k$ edge-disjoint paths connecting $u$ to $v$ in $G$, such that the length of each path is bounded by $D$. Clearly, if there is a tree packing of edge-disjoint trees of diameter at most $D$ into $G$, then $G$ must be $(k,D)$-connected. The question is whether the reverse is also true, if we allow a small congestion and a small slack in the diameter of the trees.
%
The celebrated result of Tutte and Nash-Williams shows that, if every pair of vertices in $G$ has $k$ edge-disjoint paths connecting them, then there are  $\floor{k/2}$ edge-disjoint spanning trees in $G$. However, this result is not length-preserving, in the sense that the tree paths may be much longer than the original paths connecting pairs of vertices. Our goal is then to provide such a length-preserving transformation from collections of short edge-disjoint paths connecting pairs of nodes in $G$ to a low-diameter tree packing.
%
\begin{quote}
\emph{(2) Given a $(k,D)$-connected graph $G$, can one obtain a tree packing of $\widetilde \Omega(k)$ trees of diameter $\widetilde O(D)$ into $G$, with small edge-congestion?}
\end{quote}
In this paper, we address both questions. For the first question, we show two efficient algorithms, that, given a $k$-edge connected $n$-vertex graph $G$ of diameter at most $D$, construct a low-diameter tree packing. 
We complement this result by an almost matching lower bound. We address the second question by providing an efficient algorithm, that, given a $(k,D)$-connected graph $G$, computes a collection of $k$ spanning trees of diameter at most $O(D\log n)$ each, that cause edge-congestion of $O(\log n)$.

%% file: prelim.tex
\section{Preliminaries}
\label{sec:prelim}

We use the notation $\log$ for logarithms to the base of $2$.  All graphs are finite and they do not have loops. By default, graphs are allowed to have parallel edges; graphs without parallel edges are explicitly called simple graphs.


Let $G=(V,E)$ be a graph.
For two disjoint subsets of its vertices $A,B \subseteq V$, we denote by $E_G(A,B)$ the set of edges in $G$ that have one endpoint in $A$ and the other endpoint in $B$, and denote by $\delta_G(A)$ the set of edges in $G$ that have exactly one endpoint in $A$. For a pair $u,v\in V(G)$ of vertices of $G$, we denote by $\dist_G(u,v)$ the length of the shortest path connecting $u$ to $v$ in $G$, and we denote by $\diam(G)$ the diameter of $G$, namely $\diam(G)=\max_{u,v\in V}\dist_G(u,v)$. For a path $P$ in $G$, we denote by $|P|$ its length, that is, the number of edges in $P$.
For a vertex $u \in V(G)$, let $\Gamma_G(u)$ be the set of  neighbors of $u$ in $G$. 

For two graphs $G,H$ 
we define their \emph{union graph} $G\cup H$ to be the graph whose vertex set is $V(G)\cup V(H)$ and whose edge set is $E(G)\cup E(H)$ (note that we allow $V(G)\cap V(H)$ to be non-empty). 

For a real number $p\in [0,1]$, let $\mathcal{D}(G,p)$ be the distribution of graphs, where the vertex set of the resulting graph is $V(G)$, and each edge of $G$ is included in the graph with probability $p$ independently from other edges.


We say that two paths $P$, $P'$ are \emph{edge-disjoint}, iff $E(P)\cap E(P')=\emptyset$.
We say that two paths $P$, $P'$ are \emph{internally disjoint}, iff for every vertex $v\in V(P)\cap V(P')$, $v$ is an endpoint of both paths. 
Given a set $\pset=\set{P_1,\ldots,P_r}$ of paths of $G$, we say that the paths of $\pset$ are \emph{edge-disjoint} iff every edge of $G$ belongs to at most one path of $\pset$, 
and we say that the paths of $\pset$ are \emph{internally disjoint} iff every pair of paths in $\pset$ are internally disjoint.
We say that the set $\pset$ of paths causes \emph{congestion $\eta$} iff every edge $e\in E(G)$ belongs to at most $\eta$ paths in $\pset$.


For a positive integer $k$, we say that a graph $G=(V,E)$ is \emph{$k$-edge-connected} iff, for every subset $E'\subseteq E$ of at most $k-1$ edges, $G\setminus E'$ is connected. Equivalently, $G$ is $k$-edge-connected iff for every pair $u,v\in V$ of its vertices, there is a set of $k$ edge-disjoint paths in $G$ connecting $u$ to $v$. We will also use the following stronger notion of connectedness.
\begin{definition}[$(k,D)$-connectivity]
Let $G$ be a graph, and let $k,D$ be two positive integers.
We say that $G$ is \emph{$(k,D)$-connected} iff for every pair $u,v\in V(G)$ of its nodes, there are $k$ edge-disjoint paths in $G$ connecting $u$ to $v$, such that the length of each of these paths is at most $D$.	
\end{definition}


Let $T$ be a tree rooted at $r$. 
For each integer $i\ge 0$, we say that a node $v\in V(T)$ is at the $i$th level of $T$ if the length of the unique path connecting $v$ to $r$ in $T$ is $i$. We let $V_i(T)$ be the set of all nodes that lie on the $i$th level of the tree $T$, and we denote $V_{\le i}(T)=\bigcup_{t=0}^iV_t(T)$. Therefore, the root lies at level $0$, the children of the root are at level $1$ and so on.
For a collection $\tset=\set{T_1,\ldots,T_r}$ of spanning trees of $G$, we say that the trees of $\tset$ are \emph{edge-disjoint} if every edge of $G$ belongs to at most one tree of $\tset$.
We say that the trees of $\tset$ are \emph{edge-independent}, if all the trees are rooted at a same vertex $v_0\in V(G)$, and for every vertex $v\in V(G)\setminus\{v_0\}$, if we denote by $\pset(v)$ the set of paths that contains, for each tree $T\in \tset$, the unique path connecting $v$ to $v_0$ in $T$, then all paths in $\pset(v)$ are edge-disjoint.
We say that the set $\tset$ of trees causes \emph{congestion $\eta$} iff every edge $e\in E(G)$ belongs to at most $\eta$ trees in $\tset$. 


\noindent{\bf Flows and cuts.}
Let $\pset$ be the set of all paths in $G$. A \emph{flow} $f$ in $G$ is defined to be an assignment of non-negative values $\{f(P)\}_{P\in\pset}$ to all paths $P\in \pset$. A path $P\in \pset$ is called a \emph{flow-path} of $F$ iff $f(P)>0$. The \emph{value} of the flow $f$ is $\sum_{P\in \pset}f(P)$. Let $P$ be a flow-path that originates at $u\in V(G)$ and terminates at $u'\in V(G)$. We say that the node $u$ \emph{sends} $f(P)$ units of flow to $u'$ along the path $P$. For each edge $e\in E(G)$, we define the congestion of the flow $f$ on the edge $e$ to be $\sum_{P\in \pset: e\in P}f(P)$, namely the total amount of flow of $f$ through $e$.
The total congestion of flow $f$ is the maximum congestion of $f$ on any edge of $G$.
A \emph{cut} in a graph $G$ is a bipartition of its vertex set $V$ into non-empty subsets. The \emph{value} of a cut $(S,V\setminus S)$ is $|E_G(S,V\setminus S)|$.

%% file: packing_trees_using_diameter_fixing.tex
\subsection{Bounding the Diameter of a Random Subgraph: Proof of Theorem  \ref{thm:Karger_diameter}}\label{subsec: proof of diameter bound for kargers sampling}
This subsection is dedicated to proving Theorem \ref{thm:Karger_diameter}. 
We assume that we are given an $n$-vertex $k$-edge connected graph $G=(V,E)$, with $k>1000\ln n$, and a parameter $\frac{707\ln n}{k}\le p\le 1$. Our goal is to show that a random graph $G'$, obtained by independently sub-sampling every edge of $G$ with probability $p$, has diameter at most $k^{D(D+1)/2}$ with probability at least $1-1/\poly(n)$.

Let $B$ be the bad event that the graph $G'$ is not connected. We start by establishing that $B$ only happens with low probability, using a well known result of Karger~\cite{karger1999random}.
\begin{claim}\label{claim: connected whp}
	The probability that the event $B$ happens is at most $O(1/n^{10})$.
\end{claim}
\begin{proof}
We use the following result of Karger~\cite{karger1999random}.
\begin{theorem}[Adaptation of Theorem 2.1 from~\cite{karger1999random}]
\label{thm:Karger_sampling}
Let $k,n$ be any positive integers, and let $d,p$ be any positive real numbers such that $0<p<1$. Let $G$ be an $n$-vertex $k$-edge connected graph. Let $G'\sim \mathcal{D}(G,p)$ be a random subgraph of $G$ and let $\epsilon=\sqrt{\frac{3(d+2)\ln n}{kp}}$.  
If $\epsilon<1$ then, with probability $1-O(1/n^d)$,
every cut in $G'$ has value between $(1+\epsilon)$ and $(1-\epsilon)$ times its expected value.
\end{theorem}

We apply Theorem~\ref{thm:Karger_sampling} to the graph $G$, with the parameter $p$ and $d=10$. Since $G$ is $k$-edge connected and $p\geq (707\ln n)/k$, we get that:

\[\epsilon=\sqrt{\frac{3(d+2)\ln n}{kp}}\leq \sqrt{\frac{36\ln n}{k\cdot (707
\ln n)/k}}\leq \sqrt{\frac{36}{707}} <0.3<1.\]

 Therefore, with probability $1-O(1/n^{10})$, for every cut $(S,V\setminus S)$ in $G'$, $|E_{G'}(S,V\setminus S)|\ge (1-\eps)\cdot p\cdot|E_{G}(S,V\setminus S)|\ge 0.7\cdot pk>0$. Therefore, with probability $1-O(1/n^{10})$, graph $G'$ is connected, and event $B$ happens with probability $O(1/n^{10})$.
\end{proof}

We now proceed to bound the diameter of $G'$. 
Denote $G=(V,E)$, and let $T$ be a BFS tree of $G$, rooted at an arbitrary node of $G$. Since $G$ has diameter at most $D$, the depth of $T$ is at most $D$. For each integer $0\le i\le D$, we denote by $V_i$ the set of nodes that lie at the $i$th level of $T$ (recall that the root lies at level $0$), and we denote $V_{\le i}=\bigcup_{j=0}^iV_j$. For each $0\le i\le D-1$, let $E_i$  be the set of edges of $T$ connecting vertices of $V_i$ to vertices of $V_{i+1}$.  We also let $E_{\out}=E\setminus E(T)$, so $E=E_{\out}\cup\left(\bigcup_{i=0}^{D-1}E_i\right)$.

Recall that $G'\sim \mathcal{D}(G,p)$. 
We first define a different (but equivalent) sampling algorithm for generating a random graph $G'$ from the distribution $\mathcal{D}(G,p)$. We will then use this algorithm to bound the diameter of $G'$. The algorithm consists of $D+1$ phases. For each $0\leq i\leq D$, we compute a random subgraph $G'_i$ of $G$,  with $V(G'_i)=V(G)$, such that $G'_0\subseteq G'_1\subseteq\cdots\subseteq G'_D$. The final graph $G'_D$ is denoted by $G'$. For all $0\leq i\leq D$, we denote by $\cset_i$ the set of all connected components of the graph $G'_i$. Throughout the algorithm, we maintain a set $\hat E$ of edges, that is initialized to $\emptyset$.

In order to execute the $0$th phase, we consider the edges of $E_{\out}$. Each such edge is added to the set $\hat{E}$ with probability $p$ independently from other edges. Let $E'_{\out}\subseteq E_{\out}$ be the set of edges that are added to $\hat{E}$ in this phase. We then set $G'_0=(V,E'_{\out})$. Observe that $G'_0$ may not be a connected graph. We denote by $\cset_0$ the set of all connected components of $G'_0$.
We refer to the connected components of $\cset_0$ as \emph{phase-$0$ clusters}. 

For each $1\le i\le D$, in order to execute the $i$th phase, we consider the set $E_{D-i}$ of edges. Each such edge is added to $\hat{E}$ with probability $p$ independently from other edges. We denote by $E'_{D-i}\subseteq E_{D-i}$ the set of edges that are added to $\hat{E}$ at phase $i$.
Graph $G'_i$ is obtained from the graph $G'_{i-1}$ by adding all edges of $E'_{D-i}$ to it.  As before, we denote by $\cset_i$ the set of all connected components of $G'_i$, and we call them \emph{phase-$i$ clusters}. 

Let $E'$ be the set $\hat{E}$ at the end of this algorithm. We denote by $G'=(V,E')$ the final graph that we obtain. Clearly, $G'=G'_D$, and it is generated from the distribution $\mathcal{D}(G,p)$, since $E=E_{\out}\cup\left(\bigcup_{i=0}^{D-1}E_i\right)$, and the edge sets $E_{\out},E_0,\ldots,E_{D-1}$ are mutually disjoint. We denote by $T'$ the subgraph of $T$ with $V(T')=V(T)$ and $E(T')=\bigcup_{i=0}^{D-1}E'_i$. Observe that $T'\sim\mathcal{D}(T,p)$.

Consider a pair $u,u'\in V$ of distinct vertices. We say that $u$ and $u'$ are \emph{joined at phase $0$}, if they belong to the same connected component of $G'_0$. We say that they are \emph{joined at phase $i$} for $1\leq i\leq D$, if $u$ and $u'$ belong to the same connected component of $G'_i$ but they lie in different connected components of $G'_{i-1}$. For all $0\leq i\leq D$, let $\Pi_i$ denote the set of all pairs of vertices that joined at phase $i$. Note that, if the event $B$ does not happen, then every pair $(u,u')$ of distinct vertices of $V$ lies in a unique set $\Pi_i$, for some $0\leq i\leq D$.

In order to bound the distances between pairs of nodes in $G'$, we need the following theorem, that slightly generalizes Theorem~\ref{thm:random_tree_planting}. The proof is similar to that of Theorem~\ref{thm:random_tree_planting} and is deferred to Section~\ref{sec:proof_of_random_tree_planting_generalized}.
\begin{theorem}
\label{thm:random_tree_planting_generalized}
Let $T$ be a rooted tree of depth $D$ with $|V(T)|\leq n$, and let $H$ be a connected graph with $V(H)\subseteq V(T)$. For a real number $0<p<1$, let $R\sim \mathcal{D}(T,p)$ be a random subgraph of $T$, so $V(R)=V(T)$, and every edge of $E(T)$ is added to $E(R)$ independently with probability $p$. Then with probability at least $1-\frac{D}{n^{48}}$, for every pair $u,v$ of vertices of $H$, $\dist_{R\cup H}(u,v)\le (\frac{101\ln n}{p})^{D}$.
\end{theorem}
We use a parameter $N=(101\ln n)/p$. Since $p\geq (707\ln n)/k$, we get that $7N \leq k$. 
For each $0\le i\le D$, we define a distance threshold $M_i$, as follows. We let $M_0=N^D$, and for all $1\leq i\leq D$, we let $M_i=7N^{D-i}\cdot M_{i-1}$. It is easy to verify that, for all $0\leq i\leq D$:
\[M_i\leq 7^iN^{D+(D-1)+\cdots+D-i}\leq (7N)^{D(D+1)/2}\leq k^{D(D+1)/2}.\]
For each $0\leq i\leq D$, we say that a bad event $B_i$ happens, if for some pair $(u,u')\in \Pi_0\cup\cdots\cup\Pi_i$ of distinct vertices, the distance between $u$ and $u'$ in $G'$ is greater than $M_i$. The following lemma is central to the proof of Theorem~\ref{thm:Karger_diameter}.

\begin{lemma}
\label{lem:diameter_recursion}
For each $0\le i\le D$, the probability of event $B_i$ is at most $i/n^{43}$.
\end{lemma}

Observe that, if none of the events $B,B_0,\ldots,B_D$ happen, then $G'$ is a connected graph, and in particular, every pair $(u,u')$ of distinct vertices of $G$ belongs to some set $\Pi_i$, for some $0\leq i\leq D$, so $\dist_{G'}(u,u')\leq k^{D(D+1)/2}$. Using the union bound, the probability that at least one of the events $B,B_0,\ldots,B_D$ happens is bounded by $O(1/n^{10})$. Therefore, with probability at least $1-O(1/n^{10})$,  graph $G'$ is connected, and $\diam(G')\leq  k^{D(D+1)/2}$. 
In order to complete the proof of Theorem~\ref{thm:Karger_diameter}, it is now enough to prove Lemma \ref{lem:diameter_recursion}.

\emph{Proof of Lemma \ref{lem:diameter_recursion}:}
The proof is by induction on $i$. The base case is when $i=0$. Let $(u,u')\in \Pi_0$ be any pair of vertices of $G'$ that are joined at phase $0$. Let $B_0(u,u')$ be the bad event that the distance from $u$ to $u'$ in $G'$ is greater than $M_0=N^D$. Clearly, event $B_0$ may only happen if event $B_0(u,u')$ happens for some pair $(u,u')\in \Pi_0$ of vertices. We now bound the probability of each such event  separately. 
	
Let $(u,u')\in \Pi_0$ be any pair of vertices joined at phase $0$. Recall that $u,u'$ lie in the same connected component of $G'_0$, and so there is some path $Q$ connecting $u$ to $u'$ in $G'_0$. Consider now the graph $Q$, and the tree $T$ that we have defined before, whose depth is bounded by $D$. Recall that $T'\subseteq T$ is obtained from $T$ by sub-sampling each of its edges independently with probability $p$. Using Theorem~\ref{thm:random_tree_planting_generalized} with graph $H=Q$, the  tree $T$, and the sampling probability $p$, we conclude that the probability that the distance from $u$ to $u'$ in $Q\cup T'$ is greater than $\left(\frac{101\ln n}p\right )^D=N^D$ is bounded by $D/n^{46}$. Recall that $Q\subseteq G'_0$ and so $Q\cup T'\subseteq G'$. Therefore, $\dist_{G'}(u,u')\leq \dist_{Q\cup T'}(u,u')$, and so the probability that event $B_0(u,u')$ happens is bounded by $D/n^{46}$. Using the union bound over all pairs $(u,u')\in \Pi_0$ and the fact that $D\leq n$, we conclude that $\prob{B_0}\leq 1/n^{43}$.
	
We now assume that the claim is true for all indices $0,\ldots,(i-1)$, and prove it for index $i$. 
As before, let $(u,u')\in \Pi_i$ be any pair of vertices of $G'$ that are joined at phase $i$. Let $B_i(u,u')$ be the bad event that the distance from $u$ to $u'$ in $G'$ is greater than $M_i$. Clearly, event $B_i$ may only happen if event $B_i(u,u')$ happens for some pair $(u,u')\in \Pi_i$ of vertices, or one of the events $B_0,\ldots,B_{i-1}$ happens. We now bound the probability of each such event $B_i(u,u')$  separately.

Recall that $G'_i$ is the graph that we have obtained at the end of phase $i$ of the sampling algorithm. Note that $G'_i$ is determined completely by the random choices made in phases $0,1,\ldots,i$.
Let $(u,u')\in \Pi_i$ be a pair of vertices that are joined at phase $i$. 
By the definition, $u$ and $u'$ belong to different phase-$(i-1)$ clusters but the same phase-$i$ cluster. Therefore, there is some simple path $Q$ in graph $G'_i$ that connects $u$ to $u'$.
Recall that graph $G'_i$ is obtained from the graph $G'_{i-1}$ by adding the edges of $E'_{D-i}$ to it -- the edges that we have sampled in phase $i$. The edges of $E'_{D-i}$ are sampled from the set $E_{D-i}$ of edges, connecting vertices of $V_{D-i}$ to vertices of $V_{D-i+1}$. For convenience, we denote the edges of $E'_{D-i}$ by $\tilde E$. Let $Q_1,Q_2,\ldots,Q_t$ be the set of segments of $Q$, obtained by deleting all edges of $\tilde E$ from $Q$. Note that each such segment $Q_j$ is contained in some phase-$(i-1)$ cluster, and $t\geq 2$, since $u$ and $u'$ lie in different phase-$(i-1)$ clusters.
We assume that the segments are indexed by their natural order on path $Q$, and that $u\in Q_1$, while $u'\in Q_t$.
 For each $1\leq j<t$, we let $L_j$ be the sub-path of $Q$, connecting the last vertex of $Q_j$ to the first vertex of $Q_{j+1}$. Notice that all edges in $L_j$ belong to the set $\tilde E$, and so each such segment $L_j$ is either a single edge of $\tilde E$, or it consists of two such edges, that share a common vertex in $V_{D-i}$ (see Figure~\ref{fig:joined_at_level}). In either case, each such segment $L_j$ must contain a single vertex that belongs to $V_{D-i}$, which we denote by $w_j$.  


\begin{figure}[h]
\centering
\includegraphics[scale=0.5]{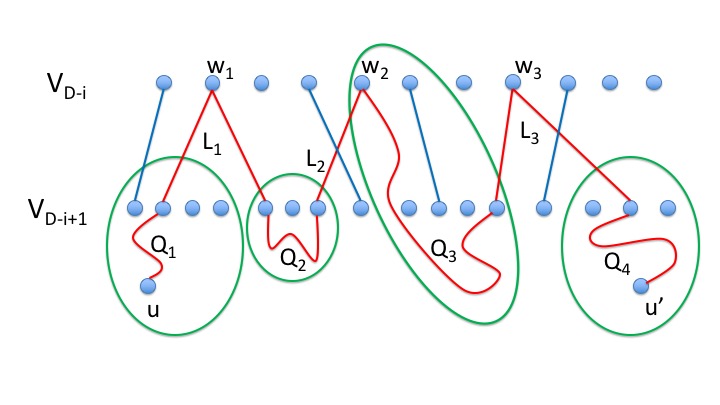}
\caption{Vertices $u$ and $u'$ are joined at level $i$; the path $Q$ is shown in red; the edges of $\tilde E\setminus E(Q)$ are shown in blue; the phase-$(i-1)$ clusters that share vertices with $Q$ are shown in green.}
\label{fig:joined_at_level}
\end{figure}


 We denote $W=\{w_1,\ldots,w_{t-1}\}$, so $W\subseteq V_{D-i}$, and we define a new graph $H$, whose vertex set is $W$, and, for each $1\le j\le t-2$, there is an edge between vertex $w_j$ and vertex $w_{j+1}$. Observe that $H$ is a path, connecting the vertices of $W$ in their natural order. Note that $H$ is guaranteed to be a connected graph, and that it only depends on the random choices made in phases $0,\ldots,i$.

 Let $\hat T$ be the sub-tree of $T$ that is induced by the vertices of $V_{\leq D-i}$, and let $\hat T'$ be the sub-tree of $\hat T$ with $V(\hat T')=V(\hat T)$, and $E(\hat T')$ containing all edges of $E'_{D-i-1}\cup\cdots\cup E'_0$. In other words, the edges of $\hat T'$ are all edges that were sampled in phases $(i+1),\ldots,D$ of the sampling algorithm.  Observe that $\hat T'\sim\mathcal{D}(\hat T,p)$.
 Finally, let $H'=H\cup \hat T'$. We let $B'_i(u,u')$ be the bad event that the distance from $w_1$ to $w_{t-1}$ in the graph $H'$ is greater than $N^{D-i}$. Observe that the event $B'_i(u,u')$ only depends on random choices made in phases $(i+1),\ldots,D$.
Using Theorem~\ref{thm:random_tree_planting_generalized} with the graph $H$, the tree $\hat T$, and the sampling probability $p$, together with the fact that $N=(101\ln n)/p$, we conclude that, the probability that the event $B'_i(u,u')$ happens is bounded by $D/n^{46}$. 
Lastly, we need the following claim.

\begin{claim}\label{claim: short distance for level i pair}
	If neither of the events $ B_{i-1},B_i'(u,u')$ happens, then neither does event $B_i(u,u')$.
\end{claim}

\begin{proof}
	Assume that neither of the events $B_{i-1},B_i'(u,u')$ happens. We show that the distance between $u$ and $u'$ in $G'$ is bounded by $M_i$, that is, event $B_i(u,u')$ does not happen.
	
	Let $P$ be the shortest path  connecting $w_1$ to $w_{t-1}$ in graph $H'$. Since we have assumed that event $B'_i(u,u')$ does not happen, $|P|\leq N^{D-i}$. We would like to turn the path $P$ into a path $P'$ connecting $u$ to $u'$ in graph $G'$, without increasing its length by too much. Observe first that an edge $e=(v,v')\in E(P)$ must be of one of two types: either it is an edge of $\hat T'$, and hence it is also an edge of $G'$; or it is an edge of the form $(w_j,w_{j+1})$, in which case it may not be an edge of $G'$. In order to complete the proof, we show that each such edge can be replaced by a short path in $G'$, and we show that $u$ and $u'$ can be connected by short paths to $w_1$ and $w_{t-1}$, respectively, in graph $G'$.

	\begin{observation}\label{claim: short paths}
		Assume that event $B_{i-1}$ does not happen. Then for each $1\leq j<t-1$, there is a path $P_j$ of length at most $2+M_{i-1}$ in graph $G'$, connecting vertex $w_j$ to vertex $w_{j+1}$.
		Moreover, there is a path $P_0$ of length at most $1+M_{i-1}$ in graph $G'$ connecting $u$ to $w_1$, and there is a path $P_{t-1}$ of length at most $1+M_{i-1}$ in graph $G'$ connecting $w_{t-1}$ to $u'$.
	\end{observation}
	
	\begin{proof}
		From the way we have partitioned the path $Q$ into segments, either $u$ and $w_1$ lie in the same phase-$(i-1)$ cluster, or there is an edge $(v,w_1)\in \tilde E$, such that $v$ lies in the same phase-$(i-1)$ cluster as $u$. In the former case, we also denote $w_1$ by $v$ for convenience. Therefore, $u$ and $v$ where joined before phase $i$, and so $\dist_{G'}(u,v)\leq M_{i-1}$, by our assumption that event $B_{i-1}$ does not happen. Therefore, there is a path in $G'$ of length at most $M_{i-1}+1$ that connects $u$ to $w_1$. Similarly, there is a path of length at most $M_{i-1}+1$ in graph $G'$ connecting $w_{t-1}$ to $u'$.
		
		Consider now some index $1\leq j<t-1$. From the definition of segments of $Q$, there is some phase-$(i-1)$ cluster $C$, and vertices $v,v'\in C$, such that: (i) either $w_j=v$, or edge $(w_j,v)\in \tilde E$; and (ii) either $w_{j+1}=v'$, or edge $(w_{j+1},v')\in \tilde E$. In either case, $v,v'\in \Pi_{i'}$ for some $i'<i$, and, since we have assumed that event $B_{i-1}$ does not happen, $\dist_{G'}(v,v')\leq M_{i-1}$. Since $\tilde E\subseteq E(G')$, $\dist_{G'}(w_j,w_{j+1})\leq 2+\dist_{G'}(v,v')\leq 2+M_{i-1}$.
	\end{proof}

In order to obtain the desired path $P'$, we replace each edge of the form $(w_j,w_{j+1})$ on path $P$ with the corresponding path $P_j$, and we append $P_1$ and $P_{t-1}$ to the beginning and to the end of the resulting path. It is easy to verify that $|P'|\leq |P|\cdot (M_{i-1}+2)+2M_{i-1}+2\leq |P|\cdot 7M_{i-1}\leq 7N^{D-i}M_{i-1}=M_i$.
\end{proof}

So far we have shown that, if the events $B_{i-1}$, $B'_i(u,u')$ do not happen, then neither does event $B_i(u,u')$. Recall that event $B_i$ may only happen if some event in $\set{B_{i-1}}\cup \set{B_i(u,u')\mid (u,u')\in \Pi_i}$ happens. 
Therefore, event $B_i$ may only happen if some event in $\set{B_{i-1}}\cup \set{B'_i(u,u')\mid (u,u')\in \Pi_i}$ happens. 

 From the induction hypothesis, the probability of 
event $B_{i-1}$ happening is bounded by $ (i-1)/n^{43}$, and, from the previous discussion, for each $(u,u')\in \Pi_i$, the probability of the event $B'_i(u,u')$ is bounded by $D/n^{46}$. Taking the union bound over all these events, and using the facts that $|\Pi_i|\leq n^2$ and $D\leq n$, we conclude that the probability that any event in $\set{B_{i-1}}\cup \set{B'_i(u,u')\mid (u,u')\in \Pi_i}$ happens is bounded by $i/n^{43}$, and this also bounds the probability of the event $B_i$.
\qed

\input{appendix_diameter_fixing.tex}

%% file: appendix_diameter_fixing.tex
\subsection{Proof of Theorem~\ref{thm:random_tree_planting_generalized}}
\label{sec:proof_of_random_tree_planting_generalized}


Recall that we are given a connected graph $H$ and a rooted  tree $T$ of depth $D$ with $|V(T)|\leq n$ and $V(H)\subseteq V(T)$, together with a parameter $p$. We let $R$ be a random subgraph of $T$ with $V(R)=V(T)$, where every edge of $E(T)$ is added to $E(R)$ with probability $p$ independently from other edges; in other words, $R\sim \mathcal{D}(T,p)$. Our goal is to show with  probability at least $1-\frac{D}{n^{48}}$, for every pair $u,v$ of vertices of $H$, $\dist_{R\cup H}(u,v)\le (\frac{101\ln n}{p})^{D}$.
The proof is a slight modification of the proof of Theorem~\ref{thm:random_tree_planting}. Note that the main difference between Theorem \ref{thm:random_tree_planting_generalized} and Theorem \ref{thm:random_tree_planting} is that now the tree $T$ may contain vertices in addition to $V(H)$.

We denote $V=V(T)$.
As before, for each $0\le i\le D$, we let $V_i$ be the set of nodes lying at level $i$ of the tree $T$, and denote $V_{\le i}=\bigcup_{t=0}^{i}V_t$. We also denote $H'=H\cup R$. 

We say that a node $x\in V(H)$ is \emph{good} if either (i) $x\in V_{\le D-1}\cap V(H)$; or (ii) $x\in V_{D}\cap V(H)$, and there is an edge in $R$ connecting $x$ to a node in $V_{D-1}$. Let $M=\frac{50\ln n}{p}$.
As before, we assume that $V(H)=\set{v_1,\ldots,v_{n'}}$, where the vertices are indexed in an arbitrary order. Given an ordered pair $(x,x')$ of vertices in $H$, and a path $P$ of $H$ connecting $x$ to $x'$, let $\sigma(P)$ be a sequence of vertices that lists all the vertices appearing on $P$ in their natural order, starting from vertex $x$. For an ordered pair $(x,x')\in V(H)$ of vertices, let $P_{x,x'}$ be shortest path connecting $x$ to $x'$ in $H$, and among all such paths $P$, choose the one whose sequence $\sigma(P)$ is smallest lexicographically. 
Observe that $P_{x,x'}$ is unique, and, moreover, if some pair $u,u'\in V(H)$ of vertices lie on $P_{x,x'}$, with $u$ lying closer to $x$ than $u'$ on $P_{x,x'}$, then the sub-path of $P_{x,x'}$ from $u$ to $u'$ is precisely $P_{u,u'}$. 

For a pair $x,x'\in V(H)$ of vertices of $H$, we let $B(x,x')$ be the bad event that length of $P_{x,x'}$ is greater than $M$ and there is no good internal node on $P_{x,x'}$.  Exactly as before, the probability that event $B(x,x')$ happens for a fixed pair $x,x'$ of vertices is at most $(1-p)^M=(1-p)^{(50\ln n)/p}<n^{-50}$.

Let $B$ be the bad event that $B(x,x')$ happens for some pair $x,x'\in V(H)$ of nodes. From the union bound over all pairs of distinct nodes in $V(H)$, the probability of $B$ is bounded by $n^{-48}$.
The following observation is an analogue of Observation \ref{obs:good_node_is_dense}, and its proof is identical.

\begin{observation}
	\label{obs:good_node_is_dense2}
	If the event $B$ does not happen, then for every node $x\in V(H)$, there is a good node $x'\in V$ such that $\dist_{H}(x,x')\le M$.
\end{observation} 

As before, we prove Theorem~\ref{thm:random_tree_planting_generalized} by induction on $D$.
The base of the induction is when $D=1$. In this case, $T$ is a star graph. Let $c$ denote the vertex that serves as the center of the star.
For any pair $x_1,x_2\in V(H)$ of vertices, we denote by $x'_1$ the good node that is closest to $x_1$ in $H$, and we define $x'_2$ similarly for $x_2$. 
Notice that, from the definition of good vertices, either $x'_1=c$, or it is connected to $c$ by an edge of $R$, and the same holds for $x'_2$.
Therefore, $\dist_{H'}(x'_1,x'_2)\le 2$ must hold.  If the event $B$ does not happen, then, since $H$ is a subgraph of $H'$, $\dist_{H'}(x_1,x_2)\le \dist_{H'}(x_1,x'_1)+\dist_{H'}(x'_1,x'_2)+\dist_{H'}(x_2,x'_2)\le \dist_{H}(x_1,x'_1)+\dist_{H'}(x'_1,x'_2)+\dist_{H}(x_2,x'_2)\le 2M+2\le \frac{101\ln n}{p}$. Therefore, with probability at least $1-n^{-48}$, $\dist_{H'}(x_1,x_2)\le \frac{101\ln n}{p}$. 

Assume now that Theorem~\ref{thm:random_tree_planting_generalized} holds for every connected graph $H$ and every tree $T$ of depth at most $D-1$, with $V(H)\subseteq V(T)$. Consider now some connected graph $H$, and a rooted tree $T$ of depth $D$, with $V(H)\subseteq V(T)$ and $|V(T)|\leq n$. 
We can assume without loss of generality that every vertex of $V_D$ lies in $V(H)$, since all other vertices of $V_D$ can be discarded from $T$.
We partition the edges of $E(T)$ into two subsets: set $E_1$ contains all edges incident to the vertices of $V_D$, and set $E_2$ contains all remaining edges. Let $R_1\subseteq R$ be the subgraph of $R$ containing only the edges of $E_1\cap E(R)$ and their endpoints, and let $R_2\subseteq R$ be obtained from $R$ by discarding all vertices of $V_D$ and their incident edges. Notice that the definition of good vertices only depends on the edges of $R_1$, and so the event $B$ only depends on the random choices made in selecting the edges of $R_1$, and is independent of the random choices made in selecting the edges of $R_2$.

Let $L$ be a subgraph of $H'$, obtained by starting with $L=H$, and then adding every edge of $R_1$ together with their endpoints to the graph. Equivalently, $L=H\cup R_1$.
 
Finally, 
we define a new graph $\hat{H}$, whose vertex set consists of two subsets: set $U_1=V_{\le D-1}\cap V(H)$, and set $U_2$, containing all vertices $v\in V_{D-1}$, such that $v$ is connected with an edge of $R_1$ to some vertex of $V_D\cap V(H)=V_D$.
We set $V(\hat H)=U_1\cup U_2$. Observe that $V(\hat H)\subseteq V(L)$. In order to define the edge set $E(\hat H)$, 
 we add an edge between a pair of nodes $w,w'$ in $\hat{H}$ iff the distance between $w$ and $w'$ in $L$ is at most $M+4$.
We also let $\hat T$ be the tree obtained from $T$, by discarding all vertices of $V_D$ from it. Observe that $V(\hat H)\subseteq V(\hat T)=V_{\leq D-1}$. As before, the idea is to use the induction hypothesis on the graph $\hat H$, together with the tree $\hat T$. In order to do so, we need to prove that $\hat H$ is a connected graph, which we do next.

\begin{observation}
	\label{obs:hatH_connected}
	If the event $B$ does not happen, then graph $\hat{H}$ is connected.
\end{observation}

\begin{proof}
	Assume that the event $B$ does not happen, and assume for contradiction that graph $\hat H$ is not connected. Let $\cset=\set{C_1,\ldots,C_r}$ be the set of all connected components of graph $\hat H$. 
	
	For every vertex $v\in V(\hat H)$, we define a set $\Gamma(v)\subseteq V(H)$ of vertices, as follows. If $v\in V(H)$, then $\Gamma(v)$ contains a single vertex -- the vertex $v$. Otherwise, $v\in V_{D-1}\setminus V(H)$ must hold, and it must be connected by at least one edge of $R_1$ to some vertex in $V_D\cap V(H)=V_D$. We then let $\Gamma(v)$ contain every vertex of $V_D$ that is connected to $v$ by an edge of $R_1$.
	
	For an ordered pair $(u,v)$ of vertices of $V(\hat H)$, we define a set $\pset(u,v)$ of paths as follows: $\pset(u,v)=\set{P_{x,y}\mid x\in \Gamma(u),y\in \Gamma(v)}$ (recall that $P_{x,y}$ is the shortest path that starts at $x$ and ends at $y$ in $H$, with the lexicographically smallest sequence $\sigma(P_{x,y})$.) Observe that every path $P_{x,y}\in \pset(u,v)$ can be augmented to a path connecting $u$ to $v$ in graph $L$, by  appending the edge $(u,x)$ to the beginning of the path (if $u\neq x$), and appending the edge $(y,v)$ to the end of the path (if $y\neq v$).

	
	For every ordered pair $(C_i,C_j)$ of distinct components of $\cset$, consider the set $\pset_{i,j}=\bigcup_{u\in C_i,v\in C_j}\pset(u,v)$ of paths. We let $P_{i,j}$ be a shortest path in $\pset_{i,j}$. We choose two distinct components $C_i,C_j\in \cset$ with $P_{i,j}$ having the shortest length, breaking ties arbitrarily. Assume that $P_{i,j}\in \pset(u,v)$, for $u\in C_i$ and $v\in C_j$. Let $x\in \Gamma(u)$ and $y\in \Gamma(v)$ be the endpoints of $P_{i,j}$, so $P_{i,j}=P_{x,y}$.
	Let $P'$ be the augmented path obtained from $P_{x,y}$, by appending the edge  $(u,x)$ to the beginning of the path (if $u\neq x$), and appending the edge $(y,v)$ to the end of the path (if $y\neq v$), so $P'$ now connects $u$ to $v$.
	 Recall that $L=H\cup R_1$, and so the path $P'$ is contained in graph $L$. Since we did not add edge $(u,v)$ to $\hat H$, the length of $P'$ is greater than $M+4$. Therefore, the length of the path $P_{x,y}$ in graph $H$ is at least $M+2$. Since we have assumed that event $B$ does not happen, there is at least one good inner vertex on path $P_{x,y}$. Let $X$ be the set of all good vertices that serve as inner vertices of $P_{x,y}$.
	
	We first show that for each $z\in X$, $z\not \in V(\hat H)$ must hold. Indeed, assume otherwise, that is, $z\in V(\hat H)$ for some $z\in X$. Then $z$ must belong to some connected component $C_{\ell}\in \cset$. Since $z$ is a good vertex, $z\in V(H)$, and so $\Gamma(z)=\set{z}$. Therefore, the sub-path of $P_{x,y}$ from $x$ to $z$ lies in $\pset(u,z)$, and the sub-path of $P_{x,y}$ from $z$ to $y$ lies in $\pset(z,v)$. We denote the former path by $P_1$ and the latter path by $P_2$. The length of each of these paths is less than the length of $P_{x,y}$.
	
	Assume first that $\ell=i$, that is, $z\in V(C_i)$. Then $P_2\in \pset_{i,j}$, and its length is less than the length of $P_{x,y}$, a contradiction. Otherwise, $\ell\neq i$. But then $P_1\in \pset_{i,\ell}$, and its length is less than the length of $P_{i,j}$, a contradiction. We conclude that for each $z\in X$, $z\not\in V(\hat H)$.

	Since $V(\hat H)$ contains all vertices of $V_{\leq D-1}\cap V(H)$, and every vertex in $X$ is a good vertex, it must be the case that $X\subseteq V_D$. Consider again some vertex $z\in X$. Since $z$ is a good vertex and $z\in V_D$, there must be an edge $e_z=(z,z')\in E(R_1)$, connecting $z$ to some vertex $z'\in V_{\leq D-1}$. From the definition of graph $\hat H$, $z'\in V(\hat H)$, and in particular, $z'$ must belong to some connected component of $\cset$, while the edge $e_z$ lies in graph $L$. Assume that $X=\set{z_1,z_2,\ldots,z_q}$, where the vertices are indexed in the order of their appearance on $P_{i,j}$, from $x$ to $y$. Consider the sequence $\sigma'=(u,z'_1,z'_2,\ldots,z'_q,v)$ of vertices. All these vertices belong to $V(\hat H)$, and $u\in C_i$, while $v\in C_j$. For convenience, denote $x=z'_0$ and $y=z'_{q+1}$. Then there must be an index $1\leq a\leq q$, such that $z'_a$ and $z'_{a+1}$ belong to distinct connected components of $\cset$. Note that the sub-path of $P_{i,j}$ between $z_a$ and $z_{a+1}$ is precisely $P_{z_a,z_{a+1}}$ -- the shortest path connecting $z_a$ to $z_{a+1}$ in $H$. Since no good vertices lie between $z_a$ and $z_{a+1}$ on this path, and since we have assumed that event $B$ does not happen, the length of this path is at most $M$. Therefore, there is a path in graph $L$, connecting $z'_a$ to $z'_{a+1}$, whose length is at most $M+2$. This path connects a pair of vertices that belong to different connected components of $\hat H$, contradicting the definition of $\hat H$.
\end{proof}

Consider now the tree $\hat T$ and the graph $\hat H$. Recall that $\hat T$ is a rooted tree of depth $D-1$, $V(\hat T)=V(\hat H)$, and, assuming the event $B$ did not happen, $\hat H$ is a connected graph. Moreover, $R_2\sim \mathcal{D}(\hat T,p)$. Therefore, assuming that event $B$ did not happen, we can use the induction hypothesis on the graph $\hat H$, the tree $\hat T$, and the random sub-graph  $R_2$ of $\hat T$. Let $B'$ be the bad event that for some pair $x_1,x_2\in V(\hat H)$ of vertices, $\dist_{\hat H\cup  R_2}(x_1,x_2)>(\frac{101\ln n}{p})^{D-1}$. From the induction hypothesis, the probability that $B'$ happens is at most $\frac{D-1}{n^{48}}$.


Lastly, we show that, if neither of the events $B,B'$ happen, then for every pair $x_1,x_2\in V(H)$ of vertices, $\dist_{H'}(x_1,x_2)\le (\frac{101\ln n}{p})^{D}$.

\begin{observation}
	If neither of the events $B,B'$ happen, then for every pair $x_1,x_2\in V(H)$ of vertices of $H$, $\dist_{H'}(x_1,x_2)\le (\frac{101\ln n}{p})^{D}$.
\end{observation}
\begin{proof}
	Consider any pair $x_1,x_2\in V(H)$ of vertices. 
	Let $x'_1$ be a good node in $V(H)$ that is closest to $x_1$ in $H$, and define $x'_2$ similarly for $x_2$. From Observation~\ref{obs:good_node_is_dense2}, $\dist_{H}(x_1,x'_1)\le M$. If $x'_1\in V_{\le D-1}$, then we define $x''_1=x'_1$, otherwise we let $x''_1$ be the node of $ V_{D-1}$ that is connected to $x'_1$ by an edge of $E'_1$, and we define $x''_2$ similarly for $x_2$.
	Therefore, $x''_1,x''_2\in V(\hat{H})$, and, assuming event $B$ does not happen,
	$\dist_{H'}(x_1,x''_1)\le M+1$, and $\dist_{H'}(x_2,x''_2)\le M+1$.
	Since we have assumed that the bad event $B'$ does not happen, $\dist_{\hat H\cup R_2}(x''_1,x''_2)\le (\frac{101\ln n}{p})^{D-1}$. 
	Recall that for every edge $e=(u,v)\in \hat H\cup R_2$, if $e\in E(R_2)$, then $e\in E(H')$; otherwise, $e\in E(\hat H)$, and there is a path in graph $H\cup R_1$ of length at most $M+4$ connecting $u$ to $v$ in $H$. Therefore, $\dist_H(x_1'',x_2'')\leq (M+4)\cdot \dist_{\hat H}(x_1'',x_2'')\leq  (\frac{101\ln n}{p})^{D-1}\cdot (M+4)$. 
	
	Altogether, $\dist_{H'}(x_1,x_2)\le  \dist_{H'}(x_1,x''_1)+\dist_{H'}(x''_1,x''_2)+\dist_{H'}(x_2,x''_2)\le (\frac{101\ln n}{p})^{D-1}\cdot (M+4)+(2M+2)\le (\frac{101\ln n}{p})^{D}$, since $M=(50\ln n)/p$. 
\end{proof}

The probability that either $B$ or $B'$ happen is bounded by $\frac{D}{n^{48}}$. Therefore, with probability at least $1-\frac{D}{n^{48}}$, neither of the events happens,  for every pair $x_1,x_2\in V(H)$ of vertices of $H$, $\dist_{H'}(x_1,x_2)\le (\frac{101\ln n}{p})^{D}$.

%% file: packing_trees_in_kD.tex
In this section we provide the proof of Theorem~\ref{thm:packing_spanning_trees_in_(k,D)_with_congestion}.
Recall that we are given a $(k,D)$-connected $n$-vertex graph $G$. Our goal is to design an efficient randomized algorithm that computes a collection $\tset=\{T_1,\ldots,T_{k}\}$ of $k$ spanning trees of $G$, such that, for each $1\le \ell\le k$, the tree $T_{\ell}$ has diameter at most $O(D\log n)$, and with high probability each edge of $G$ appears in $O(\log n)$ trees of $\tset$. Note that we allow the graph $G$ to have parallel edges. However, we can assume w.l.o.g. that for every pair $(u,v)$ of vertices of $G$, there are at most $k$ parallel edges $(u,v)$; all remaining edges can be deleted without violating the $(k,D)$-connectivity property of $G$. 

The main tool that we use in our proof is the following theorem and its corollary.
\begin{theorem}\label{thm: finding_flow_body}
	There is an efficient algorithm, that, given a $(k,D)$-connected graph $G$, a subset $U\subsetneq V(G)$ of its vertices, and an additional vertex $s\in V(G)\setminus U$, computes a flow $f$ in $G$ with the following properties:
	
	\begin{itemize}
		\item the endpoints of every flow-path lie in $U\cup\set{s}$;
		\item for each vertex $u\in U$, the total flow on all paths that originate or terminate at $u$ is at least $k$;
		\item the total amount of flow through any edge is at most $2$; and
		\item each flow-path has length at most $2D$.
	\end{itemize}
\end{theorem}
Notice that a flow-path is allowed to contain vertices of $U\cup \set{s}$ as inner vertices. 
We defer the proof of Theorem \ref{thm: finding_flow_body} to Section \ref{subsec: finding the flow}, after we complete the proof of Theorem \ref{thm:packing_spanning_trees_in_(k,D)_with_congestion} using it. We obtain the following useful corollary of the theorem.
\begin{corollary}\label{cor: flow and bipartition}
	There is an efficient algorithm, that, given a $(k,D)$-connected graph $G$ and a subset $S\subseteq V(G)$ of its vertices, computes a bi-partition $(S',S'')$ of $S$, and a flow $f$ from vertices of $S''$ to vertices of $S'$, such that the following hold:
	
	\begin{itemize}
		\item every vertex of $S''$ sends at least $k/2$ flow units;
		\item every flow-path has length at most $2D$;
		\item the total amount of flow through any edge is at most $3$; and
		\item $|S'|\leq \frac{|S|}{2}+1$.
	\end{itemize}
\end{corollary}

\begin{proof}
	Let $s\in S$ be an arbitrary vertex, and set $U=S\setminus \set{s}$. We apply Theorem \ref{thm: finding_flow_body} to graph $G$, vertex set $U$ and the vertex $s$, to obtain a flow $f$. Recall that every vertex of $U$ sends or receives at least $k$ flow units, and all flow-paths have length at most $2D$. Let $\pset'$ be the set of all paths in $G$ on which a non-zero amount of flow is sent. Since the algorithm in Theorem \ref{thm: finding_flow_body} is efficient, we are guaranteed that $|\pset'|\leq n^c$ for some constant $c$, where $n=|V(G)|$. It will be convenient for us to ensure that for every path $P\in \pset'$, $f(P)$ is an integral multiple of $1/n^c$. In order to achieve this, for every flow-path $P\in \pset'$, we round $f(P)$ up to the next integral multiple of $1/n^c$. Note that this increases the total amount of flow by at most $1$, so the total amount of flow through any edge is at most $3$.
	
	
	We now compute a bi-partition $(S',S'')$ of $S$, as follows. We start from an arbitrary partition $(S',S'')$. Consider any vertex $v\in S$. For convenience, we direct all flow-paths of $\pset'$ for which $v$ serves as an endpoint away from $v$. Let $q'(v)$ be the total amount of flow that originates at $v$ and terminates at vertices of $S'$, and define $q''(v)$ similarly for the total amount of flow between $v$ and $S''$.
	
	If $v\in S'$, but $q'(v)>q''(v)$, then we move $v$ from $S'$ to $S''$.  Similarly, if $v\in S''$, but $q''(v)>q'(v)$, then we move $v$ from $S''$ to $S'$. Notice that in either case, the total amount of flow between vertices of $S'$ and vertices of $S''$ increases by at least $1/n^c$. We continue performing these modifications, until for every vertex $v\in S'$, $q'(v)\leq q''(v)$, and for every vertex $v\in S''$, $q''(v)\leq q'(v)$. Since the total amount of flow between $S'$ and $S''$ grows by at least $1/n^c$ in every iteration, the number of such iterations is bounded by $O(|E(G)|\cdot n^c)=O(\poly(n))$.
	
	Consider the partition $(S',S'')$ of $S$ obtained at the end of this algorithm. Assume w.l.o.g. that $|S'|\leq |S''|$; otherwise we switch $S'$ and $S''$. If the vertex $s$ lies in $S''$, then we move it to $S'$. Notice that we are now guaranteed that for every vertex $u\in S''$, $q'(u)\geq q''(u)$, and so at least $k/2$ flow units are sent between $u$ and the vertices of $S'$. In order to obtain the final flow $f'$, we discard from $f$ all flow-paths except those connecting the vertices of $S''$ to the vertices of $S'$, and we direct these flow paths towards the vertices of $S'$. It is easy to verify that $|S'|\leq |S|/2+1$. 
\end{proof}

Our algorithm consists of two phases. In the first phase, we define a partition of the vertices of $G$ into layers $L_1,\ldots,L_h$, where $h=O(\log n)$. Additionally, for each $1\leq i\leq h$, we define a flow $f_i$ in graph $G$ between vertices of $L_i$ and vertices of $L_1\cup\cdots\cup L_{i-1}$. In the second phase, we use the layers and the flows in order to construct the desired set of spanning trees.

\paragraph{Phase 1: partitioning into layers.}
We use a parameter $h=\Theta(\log n)$, whose exact value will be set later. We now define the layers $L_h,\ldots,L_1$ in this order, and the corresponding flows $f_h,\ldots,f_1$. In order to define the layer $L_h$, we let $S=V(G)$, and we apply Corollary \ref{cor: flow and bipartition} to the graph $G$ and the set $S$ of its vertices, to obtain a partition $(S',S'')$ of $S$, with $|S'|\leq |S|/2+1$, and the flow $f$ between the vertices of $S''$ and the vertices of $S'$, where every vertex of $S''$ sends at least $k/2$ units of flow, each flow-path has length at most $2D$, and the edge-congestion caused by $f$ is at most $3$. We then set $L_h=S''$ and $f_h=f$, and continue to the next iteration.

Assume now that we have constructed $L_h,\ldots,L_i$, we now show how to construct $L_{i-1}$. Let $S=V(G)\setminus (L_h\cup\cdots\cup L_i)$. We apply 
Corollary~\ref{cor: flow and bipartition} to the graph $G$ and the set $S$ of its vertices, to obtain a partition $(S',S'')$ of $S$, with $|S'|\leq |S''|/2+1$, and the corresponding flow $f$. We then set $L_{i-1}=S''$, $f_{i-1}=f$, and continue to the next iteration. If we reach an iteration where $|S|\leq 2$, we arbitrarily designate one of the two vertices as $s$, and we let $U$ be a set of vertices containing the other vertex. We then use Theorem \ref{thm: finding_flow_body} in order to find a flow of value at least $k$ between the two vertices, such that the edge-congestion of the flow is at most $2$, and every flow-path has length at most $2D$. We then add the vertex that lies in $U$ to the current layer, and the vertex $s$ to the final layer $L_1$. If we reach an iteration where $|S|=1$, then we add the vertex of $S$ to the final layer $L_1$ and terminate the algorithm. The number $h$ of layers is chosen to be exactly the number of iterations in this algorithm. Notice that $h\leq 2\log n$ must hold. Observe also that, for all $1<i\leq h$, flow $f_i$ originates at vertices of $L_i$, terminates at vertices of $L_1\cup\cdots\cup L_{i-1}$, uses flow-paths of length at most $2D$, and causes edge-congestion at most $3$.

\paragraph{Phase 2: constructing the trees.}
In order to construct the spanning trees $T_1,\ldots,T_k$, we start with letting each tree contain all vertices of $G$ and no edges. We then process every vertex $v\in V(G)$ one-by-one. Assume that $v\in L_i$, for some $1\leq i\leq h$. Consider the following experiment. Let $\qset(v)$ be the set of all flow-paths that carry non-zero flow in $f_i$, and connect $v$ to vertices of $L_1\cup\cdots\cup L_{i-1}$. Let $F(v)$ be the total amount of flow $f_i$ on all paths $P\in \qset(v)$; recall that $F(v)\geq k/2$ must hold. We choose a path $P\in \qset(v)$ at random, where the probability to choose a path $P$ is precisely $f_i(P)/F(v)$. We repeat this experiment $k$ times, obtaining paths $P_1(v),\ldots,P_k(v)$. For each $1\leq j\leq k$, we add all edges of $P_j(v)$ to $T_j$. 
Consider the graphs $T_1,\ldots,T_k$ at the end of this process. Notice that each such graph $T_j$ may not be a tree. We show first that the diameter of each such graph is bounded by $O(D\log n)$.

\begin{claim}\label{claim: low diam}
	For all $1\leq j\leq k$, $\diam(T_j)\leq O(D\log n)$.
\end{claim}
\begin{proof}
	Fix an index $1\leq j\leq k$.
	Let $r$ be the unique vertex lying in $L_1$. We prove that for all $1\leq i\leq h$, for every vertex $v\in L_i$, there is a path connecting $v$ to $r$ in $T_j$, of length at most $2D(i-1)$ by induction on $i$.
	
	The base of the induction is when $i=1$ and the claim is trivially true. Assume now that the claim holds for layers $L_1,\ldots,L_{i-1}$. Let $v$ be any vertex in layer $L_i$. Consider the path $P_j(v)$ that we have selected. Recall that this path has length at most $2D$, and it connect $v$ to some vertex $u\in L_1\cup\cdots\cup L_{i-1}$. By the induction hypothesis, there is a path $P$ in $T_j$ of length at most $2D(i-2)$, that connects $u$ to $r$. Since all edges of $P_j(v)$ are added to $T_j$, the path $P_j(v)$ is contained in $T_j$. By concatenating path $P_j(v)$ with path $P$, we obtain a path connecting $v$ to $r$, of length at most $2D(i-1)$.
\end{proof}

Lastly, we prove that with high probability, every edge of $G$ belongs to $O(\log n)$ graphs $T_1,\ldots,T_k$.
\begin{claim}\label{claim: low congestion}
	With probability at least $(1-1/\poly(n))$, every edge of $G$ lies in at most $O(\log n)$ graphs $T_1,\ldots,T_k$.
\end{claim}
The proof follows the standard analysis of the Randomized Rounding technique and is delayed to Section \ref{subsec: proof of randomized rounding}.

For each $1\leq j\leq k$, we can now let $T'_j$ be a BFS tree of the graph $T_j$, rooted at the vertex $r$. From Claim \ref{claim: low diam}, each tree $T'_j$ has diameter at most $O(D\log n)$, and from Claim \ref{claim: low congestion}, the resulting set of trees cause edge-congestion $O(\log n)$. 

\subsection{Proof of Theorem \ref{thm: finding_flow_body}} \label{subsec: finding the flow}


For every vertex $u\in U$, let $\pset(u)$ be the set of all paths in graph $G$ of length at most $2D$, that connect $u$ to vertices of $(U\cup\set{s})\setminus\set{u}$. Notice that for a pair $u,u'\in U$ of distinct vertices, each path connecting $u$ to $u'$ belongs to both $\pset(u)$ and $\pset(u')$. Let $\pset^*=\bigcup_{u\in U}\pset(u)$. We use the following linear program, that has no objective function; our goal will be to find a feasible solution satisfying all constraints.
\begin{eqnarray*}
\mbox{(LP-1)}	&\sum_{P\in \pset(u)}f(P)\geq k&\forall u\in U\\
	&\sum_{\stackrel{P\in \pset^*:}{e\in P}}f(P)\leq 2&\forall e\in E(G)\\
	&f(P)\geq 0&\forall P\in \pset^*
	\end{eqnarray*}

Note that, if $f$ is a feasible solution to (LP-1), then it satisfies all requirements of Theorem \ref{thm: finding_flow_body}.
The following claim provides an efficient algorithm for solving (LP-1); its proof  uses standard techniques and is deferred to Section~\ref{sec:solve_LP}.

\begin{claim}\label{claim: solve LP}
	There is an efficient algorithm that computes a feasible solution to (LP-1), if such a solution exists.
\end{claim}

It now remains to prove that there is a feasible solution to (LP-1). We do so using the following lemma, that proves a stronger claim, namely that there is an integral solution to (LP-1).

 \begin{lemma}
 	\label{lem:spider_connection}
 	Let $G$ be a $(k,D)$-connected graph, let $U\subsetneq V(G)$ be any subset of its vertices, and let $s\not\in U$ be any additional vertex. Then there exists a set $\cal{P}$ of paths in $G$, such that:
 	\begin{itemize}
 		\item each path $P\in \cal{P}$ connects a pair of distinct vertices in $U\cup \set{s}$;
 		\item each node in $U$ is the endpoint of at least $k$ paths in $\cal{P}$ (but $s$ may serve as an endpoint on fewer paths); 
 		\item each path $P\in \cal{P}$ has length at most $2D$; and
 		\item each edge of $G$ appears on at most two paths in $\cal{P}$.
 	\end{itemize}
 \end{lemma}

Notice that the lemma immediately implies that there is a feasible solution to (LP-1), as we can simply send one unit of flow on each path of $\pset$. We now turn to prove Lemma \ref{lem:spider_connection}.

\emph{Proof of Lemma~\ref{lem:spider_connection}.}
The proof relies on a theorem from~\cite{chuzhoy2008algorithms}, that needs the following definitions.

\begin{definition}\emph{\textbf{(Canonical Spider)}} 
	Let $\cal{M}$ be any collection of simple paths, such that each path $P\in \cal{M}$ has
	a distinguished endpoint $t(P)$, and the other endpoint is denoted by $v(P)$. We say that the paths in $\cal{M}$ form a
	\emph{canonical spider} iff $|\mathcal{M}|>1$ and there is a vertex $v$, such that for every path $P\in \cal{M}$, $v(P) = v$. Moreover, the only vertex that appears on more than one path of $\cal{M}$ is $v$ (see Figure \ref{fig:spidercycle}). We refer to v as the \emph{head} of the spider, and the paths of $\cal{M}$ are called the \emph{legs} of the spider. 
\end{definition}

\begin{definition}\emph{\textbf{(Canonical Cycle)}}
	\label{def:cycle}
	Let $\mathcal{M}= \{Q_1,\ldots,Q_h\}$ be any collection of simple paths, where each path
	$Q_i$ has a distinguished endpoint $t(Q_i)$ that does not appear on any other path of $\cal{M}$, and the other endpoint
	is denoted by $v(Q_i)$. We say that paths of $\cal{M}$ form a \emph{canonical cycle}, iff:
	
	\begin{itemize}
		\item  $h$ is an odd integer;
		\item for each $1\leq i\leq h$, there is a vertex $v'(Q_i)\ne v(Q_i)$ on path $Q_i$, such that $v'(Q_i)=v(Q_{i-1})$ (here we use the convention that $Q_0= Q_h$); and
		\item for each $1\leq i\leq h$, no vertex of $Q_i$ appears on any other path of $\cal{M}$, except for $v'(Q_i)$ that belongs to $Q_{i-1}$
		only and $v(Q_i)$ that belongs to $Q_{i+1}$ only (see Figure~\ref{fig:spidercycle}).
	\end{itemize}  
\end{definition}

Note that the definition of a canonical cycle here is slightly stronger than definition of a canonical cycle in~\cite{chuzhoy2008algorithms}, since we additionally require that, for each $1\leq i\leq h$, the vertex $v'(Q_i)\ne v(Q_i)$.

\begin{figure}[h]
	\centering
	\scalebox{1.0}{\includegraphics[scale=0.52]{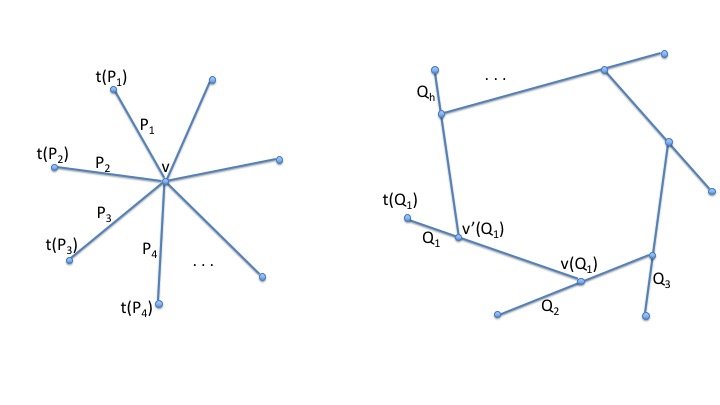}}
	\caption{A canonical spider (left) and a canonical cycle (right).\label{fig:spidercycle}}
\end{figure}


We use the following result of Chuzhoy and Khanna (Theorem $4$ in~\cite{chuzhoy2008algorithms}). We note that the theorem appearing in~\cite{chuzhoy2008algorithms} is slightly weaker since they used a weaker definition of a canonical cycle, but their proof immediately implies the stronger result that we state below.

\begin{theorem}
	\label{thm:prefix_decomposition}
	There is an efficient algorithm, that, 
	given any collection $\cal{Q}$ of paths, where every path $P \in \cal{Q}$ has a distinguished
	endpoint $t(P)$ that does not appear on any other path of $\cal{Q}$, computes, for
	each path $P \in \cal{Q}$, a prefix (i.e. a sub-path of $P$ that contains $t(P)$) $q(P)$, such that, in the graph induced by $\{q(P)\mid P \in \cal{Q}\}$, the prefixes appearing in each connected component either form a canonical spider, a canonical cycle, or the connected component contains exactly one prefix $q(P)$, where $q(P) = P$ for some $P \in \cal{Q}$.
\end{theorem}

Recall that we are given a $(k,D)$-connected graph $G$, together with a subset $U\subsetneq V(G)$ of its vertices, that we call \emph{terminals}, and an additional vertex $s\not \in U$. From the definition of $(k,D)$-connectivity, we are guaranteed that every vertex $u\in U$, there is a set $\rset(u)$ of $k$ edge-disjoint simple paths in $G$, of length at most $D$ each, connecting $u$ to $s$. Let $\rset=\bigcup_{u\in U}\rset(u)$. Intuitively, we would like to apply Theorem \ref{thm:prefix_decomposition} to the set $\rset$ of paths, where for each vertex $u\in U$, and for each path $R\in \rset(u)$, the distinguished endpoint $t(R)$ is $u$. The idea is then to use the resulting canonical cycle and canonical spider structures in order to connect the vertices of $U$ to each other and to $s$ via short paths that are disjoint in their edges, thus constructing the collection $\pset$ of paths. For example, if a set $\mset$ of prefixes of the paths in $\rset$ form a canonical spider, we can partition the legs of the spider into pairs, and each pair then defines a path connecting two vertices of $U$ to each other, which is then added to $\pset$. There are two problems with this approach. The first problem is that Theorem \ref{thm:prefix_decomposition} requires that the distinguished endpoints $t(P)$ of the paths $P\in \rset$ are distinct from each other, and moreover that $t(P)$ does not lie on any other path of $\rset$. This difficulty is easy to overcome by making $k$ copies of every terminal $u\in U$ and then modifying the paths in $\rset(u)$ so that each of them starts from a different copy. The second difficulty is that it is possible that some resulting set $\mset$ of prefixes that forms a canonical spider consists entirely of paths that belong to a single set $\rset(u)$, and so the spider cannot be used to connect distinct vertices of $U$ to each other. The reason that this may happen is that the paths in $\rset(u)$ are only guaranteed to be edge-disjoint, and so they may share vertices. If, in contrast, they were internally vertex-disjoint, then such a problem would not arise. In order to overcome these difficulties, we slightly modify the graph $G$, first by replacing it with its line graph, so that any set of edge-disjoint paths in $G$ corresponds to a set of internally node-disjoint paths in the line graph, and then creating $k$ copies of each terminal $u\in U$. We now describe the construction of the modified graph $H$, in two steps.


In the first step, we construct the line graph $L$ of $G$, as follows: the vertex set $V(L)$ contains a vertex $v_e$ for each edge $e\in E(G)$. Given a pair $v_e,v_{e'}$ of vertices of $L$, we connect them with an edge iff $e$ and $e'$ share an endpoint in $G$. 
%
%

Let $H$ be the graph obtained from graph $L$ by adding, for each terminal $u\in U$, a collection $\set{u_1,\ldots,u_k}$ of $k$   vertices, that we call the \emph{copies of $u$}. For each such new vertex $u_i$, and for every edge $e$ that is incident to $u$ in $G$, we add the edge $(u_i,v_e)$ to the graph. Additionally, we add the vertex $s$ to the graph, and connect it to every vertex $v_e$ where $e$ is an edge incident to $s$ in $G$. 


Recall that we have defined, for every vertex $u\in U$, a collection $\rset(u)$ of $k$ edge-disjoint simple paths in $G$ of length at most $D$ each, connecting $u$ to $s$.
Denote $\rset(u)=\set{R_1(u),\ldots,R_k(u)}$. We transform the set $\rset(u)$ of paths into a set $\rset'(u)$ of $k$ paths in graph $H$, that are internally vertex-disjoint, and each path connects a distinct copy of $u$ to $s$. In order to do so, fix some $1\leq i\leq k$, and consider the path $R_i(u)$. Let $e_1^i,e_2^i,\ldots,e_r^i$ be the sequence of edges on the path $R_i(u)$, with $e_1^i$ incident to $u$ and $e_r^i$ incident to $s$. Consider the following sequence of vertices in graph $H$: $(u_i,v_{e_1^i},v_{e_2^i},\ldots,v_{e_r^i},s)$. It is easy to verify that this vertex sequence defines a path in graph $H$, that we denote by $R'_i(u)$. Let $\rset'(u)=\set{R'_i(u)\mid 1\leq i\leq k}$ be the resulting set of paths. Since the paths in $\rset(u)$ are edge-disjoint, it is immediate to verify that the paths in $\rset'(u)$ are internally node-disjoint; in fact the only vertex that these paths share is the vertex $s$. The number of inner vertices on each such path is at most $D$. For each path $R'_i(u)$, we let its distinguished endpoint $t(R'_i(u))$ be the vertex $u_i$. 
%
Lastly, we let $\qset=\bigcup_{u\in U}\rset'(u)$. Observe that for every path $R\in \qset$, the distinguished endpoint $t(R)$
does not lie on any other paths of $\qset$.

We apply Theorem~\ref{thm:prefix_decomposition} to the resulting set $\qset$ of paths and obtain, for each path $P\in \qset$, a prefix $q(P)$. Let $\hat{H}$ be the subgraph of $H$ that is induced by all edges and vertices that appear on the paths in $\{q(P)\mid P \in \qset\}$. 
Let $\cset$ be the set of all connected components of $\hat H$. 
For every component $C\in \cset$, we denote by $\qset(C)\subseteq \qset$ the set of paths whose prefixes are contained in $C$, and we denote by $\tilde{\qset}(C)=\set{q(P)\mid P\in \qset(C)}$ the corresponding set of prefixes, so $C=\bigcup_{P'\in \tilde \qset(C)}P'$.

Next, for every component $C\in \cset$, we define a collection $\pset(C)$ of paths in the original graph $G$, with the following properties:
\begin{properties}{P}
	\item an edge of $G$ may lie on at most two paths in $\pset(C)$; \label{prop: paths in one set edge-disjoint}
	\item the paths in $\pset(C)$ only contain edges $e\in E(G)$ with $v_e\in V(C)$; \label{prop: paths only use allowed edges}
	\item for every terminal $u\in U$, the number of paths of $\pset(C)$ for which $u$ serves as an endpoint is at least as large as the number of  paths of $\rset'(u)$ that lie in $\qset(C)$; and \label{prop: few paths per terminal} 
	\item every path in $\pset(C)$ has length at most $2D$; \label{prop: short paths}
\end{properties}

Assume first that we have computed, for every component $C\in \cset$, a set $\pset(C)$ of paths in graph $G$ with the above properties. We then set $\pset=\bigcup_{C\in \cset}\pset(C)$. It is easy to verify that set $\pset$ has all required properties. Indeed, since the components of $\cset$ are disjoint in their vertices, Properties \ref{prop: paths in one set edge-disjoint} and \ref{prop: paths only use allowed edges} ensure that every edge of $G$ belongs to at most two paths of $\pset$. Since, for every terminal $u\in U$, $|\rset'(u)|=k$, Property \ref{prop: few paths per terminal} ensures that $u$ serves as an endpoint of at least $k$ paths in $\pset$. Lastly, Property \ref{prop: short paths} ensures that the length of every path in $\pset$ is at most $2D$.

From now on we fix a component $C\in \cset$. It is now sufficient to show an efficient algorithm for constructing the set $\pset(C)$ of paths in graph $G$ with Properties \ref{prop: paths in one set edge-disjoint}---\ref{prop: short paths}. 
Recall that Theorem~\ref{thm:prefix_decomposition} guarantees that the prefixes in $\tqset(C)$ either form a canonical spider, or they form a canonical cycle, or $\tqset(C)$ consists of a single path $q(P)=P$ for some path $P\in \qset$. We consider each of these different cases in turn; for the case of canonical spider we need to consider two sub-cases, depending on whether the head of the spider is $s$ or not.

\paragraph{Case 1:} This case happens if $\tqset(C)$ contains a single path, or if the paths of $\tqset(C)$ form a canonical spider, whose head is $s$. In either case, from the construction of the paths in $\qset$, it is easy to verify that for every path $P\in \qset(C)$, the prefix $q(P)$ is the path $P$ itself. For each path $P\in \tqset(C)$, we define a path $P'$ in graph $G$, as follows. Assume that $P=(u_i,v_{e_1},v_{e_2},\ldots,v_{e_r},s)$. We then let $P'$ be a path in graph $G$, that starts at the terminal $u$, traverses the edges $e_1,\ldots,e_r$ in this order, and terminates at $s$. Let $\pset(C)=\set{P'\mid P\in \tqset(C)}$. Since the paths in $\tqset(C)$ are vertex-disjoint except for sharing the vertex $s$, the paths in $\pset(C)$ are all edge-disjoint. It is easy to verify that Properties \ref{prop: paths in one set edge-disjoint}---\ref{prop: short paths} hold for $\pset(C)$.

\paragraph{Case 2:} This case happens if the paths in $\tqset(C)$ form a canonical spider, whose head is not $s$. Note that, from the definition of the paths in $\qset$, the head of the spider must be some vertex $v_{e^*}$ with $e^*\in E(G)$. We denote $e^*=(x,y)$. 
Note that every path $P\in \qset(C)$ contains the vertex $v_{e^*}$. Therefore, each such path must belong to a different set $\rset'(u)$, and no two paths in $\qset(C)$ may originate from two copies of the same terminal.
For every path $P\in \tqset(C)$, we define a new path $P'$ in graph $G$, as follows. Assume that the sequence of vertices on $P$ is $(u_i,v_{e_1},v_{e_2},\ldots,v_{e_r},v_{e^*})$, then we let path $P'$ start at the terminal $u$, and then traverse the edges $e_1,e_2,\ldots,e_r$ in this order. Note that path $P'$ has to terminate at a vertex that serves as an endpoint of $e^*$. We define two sets of paths: set $S_x$ contains all paths $P'$ for $P\in \tqset(C)$ that terminate at $x$, and set $S_y$ is defined similarly for $y$. Therefore, $|S_x|+|S_y|=|\tqset(C)|$. From the above discussion, every path in $S_x\cup S_y$ originates at a distinct terminal.


Assume first that $|S_x|>1$ and $|S_y|>1$.
Consider the set $S_x$ of paths. We construct a set $\Pi_x$ of pairs of paths from $S_x$ as follows. If $|S_x|$ is even, then we simply partition all paths in $S_x$ into $|S_x|/2$ disjoint pairs. Otherwise, if $|S_x|$ is odd, then we construct $(|S_x|+1)/2$ pairs, such that every path of $S_x$ belongs to exactly one pair in $\Pi_x$, except for one arbitrary path that belongs to two pairs. 
Consider now any pair $(P_1',P_2')$ of paths in $\Pi_x$. As observed before, the two paths must originate at distinct terminals. We construct a new path by concatenating $P_1'$ with $P_2'$, and add this path to $\pset(C)$. We process the paths of $S_y$ similarly. 
Notice that every prefix in $\tqset(C)$ is now a sub-path of either one or two paths in $\pset(C)$. Since the paths in $\tqset(C)$ are internally vertex disjoint, and since the edge $e^*$ is not included in any of the paths in $S_x\cup S_y$, every edge of $G$ may belong to at most two paths of $\pset(C)$. It is immediate to verify that Properties \ref{prop: paths in one set edge-disjoint}---\ref{prop: short paths} hold in $\pset(C)$.

Assume now that $|S_x|=1$ or $|S_y|=1$ (or both). We assume w.l.o.g. that $|S_y|=1$. We construct the set $\Pi_x$ of pairs of paths in $S_x$ exactly as before (if $|S_x|=1$ then $\Pi_x=\emptyset$). For every pair $(P_1',P_2')$ of paths in $\Pi_x$, we construct a new path that is added to $\pset(C)$ exactly as before. Additionally, we choose an arbitrary path $P'_i\in S_x$ that participates in at most one pair in $\Pi_x$ (notice that such a path has to exist). Let $P'$ be the unique path in $S_y$. As observed before, the two paths must originate from distinct terminals. We construct a new path in graph $G$, by concatenating the path $P'_i$, the edge $e^*$, and the path $P'$. We add the resulting path to $\pset(C)$. It is easy to verify that the resulting set $\pset(C)$ of paths satisfy Properties \ref{prop: paths in one set edge-disjoint}---\ref{prop: short paths}.

\paragraph{Case 3:} This case happens if the paths in $\tqset(C)$ form a canonical cycle. We denote the paths of $\tqset(C)$ by $Q_1,\ldots,Q_h$ in the order of their appearance on the cycle. 
We define the following set of pairs of these paths: $\Pi=\set{(Q_1,Q_2),(Q_3,Q_4),\ldots,(Q_{h-2},Q_{h-1}), (Q_{h-1},Q_h)}$ (recall that $h$ is an odd integer). Notice that every path appears in exactly one pair of $\Pi$, except for the path $Q_{h-1}$, that appears in two pairs.

Consider now some pair $(Q_i,Q_{i+1})\in \Pi$. 
We construct a two-legged spider $S_i$, that consists of the path $Q_i$, and the sub-path of $Q_{i+1}$, from $t(Q_{i+1})$ to $v'(Q_{i+1})=v(Q_i)$.
 In the resulting collection $S_1,S_3,\ldots,S_{h-2},S_{h-1}$ of spiders, every pair of spiders are mutually vertex-disjoint, except for the vertices of $Q_{h-1}$ that may appear in two spiders. We process each one of these spiders as in Case 2, to obtain a collection $\pset(C)$ of $(h+1)/2$ paths in graph $G$ that cause edge-congestion at most $2$, and that satisfy Properties \ref{prop: paths in one set edge-disjoint}---\ref{prop: short paths}.
\qed

\input{appendix_LP.tex}

%% file: appendix_LP.tex
\subsection{Proof of Claim \ref{claim: low congestion}} \label{subsec: proof of randomized rounding}
Let $f$ be the flow obtained by taking the union of the flows $f_1,\ldots,f_h$. It is easy to verify that flow $f$ causes edge-congestion at most $4h\leq 8\log n$. For every edge $e\in E(G)$, we say that a bad event $B(e)$ happens if $e$ lies in more than $120\log n$ graphs $T_1,\ldots,T_k$. It is enough to show that for each edge $e\in E(G)$, the probability of the event $B(e)$ is bounded by $1/n^6$; from the union bound over all edges $e$, it then follows that with probability at least $(1-1/n^3)$, the graphs in $\set{T_1,\ldots,T_k}$ cause edge-congestion at most $120\log n$ (we have used the fact that for every pair $(u,v)$ of vertices of $G$, there are at most $k$ parallel edges $(u,v)$ in $G$, and that $k\leq n$).

For the remainder of the proof, we fix an edge $e\in E(G)$, and we prove that the probability of event $B(e)$ is at most $1/n^6$.
	
For every vertex $v\in V(G)$, and index $1\leq j\leq k$, we let $X(v,j)$ be a random variable whose value is $1$ if the path $P_j(v)$ contains the edge $e$, and it is $0$ otherwise. Notice that, if we denote $S=\sum_{v\in V(G)}\sum_{j=1}^kX(v,j)$, then  the number of graphs $T_1,\ldots,T_k$ to which edge $e$ belongs is exactly $S$. Moreover, the random variables in $\set{X(v,j)\mid v\in V(G),1\leq j\leq k}$ are independent from each other. Consider some vertex $v\in V(G)$, and let $F(v,e)$ be the total amount of flow that $f$ sends on all flow-paths that originate from $v$ and contain the edge $e$. Notice that for each $1\leq j\leq k$, the probability that $X(v,j)=1$ is $F(v,e)/F(v)$. Therefore, the expectation of $\sum_{1\leq j\leq k}X(v,j)=k\cdot F(v,e)/F(v)\leq 2F(v,e)$, since $F(v)\geq k/2$. Altogether, the expectation of $S=\sum_{v\in V(G)}\sum_{1\leq j\leq k}X(v,j)$ is at most $2\sum_{v\in V(G)}F(v,e)$, which is precisely the total amount of flow traversing $e$ in $f$ times $2$, and is bounded by $8h\leq 16\log n$. To summarize, we are given a collection $\set{X(v,j)\mid v\in V(G),1\leq j\leq k}$ of independent $0/1$ random variables. The expectation of their sum is at most $16\log n$. We need to bound the probability that $S>120\log n$.
	
We use the following standard Chernoff bound (see e.g. \cite{measure-concentration}).
\begin{theorem}\label{thm: Chernoff}
Let $\set{Y_1,\ldots,Y_r}$ be a collection of independent random variables taking values in $[0,1]$, and let $Y=\sum_iY_i$. Assume that $\expect{Y}\leq \mu$ for some value $\mu$. Then for all  $0<\eps<1$:
\[\prob{Y>(1+\eps)\mu }\leq e^{-\eps^2\mu /3}.\]
\end{theorem}
	
Using the above bound with $\eps=1/2$ and $\mu=80\log n$, we get that the probability that $S>120\log n$ is bounded by $e^{-80\log n/12}<1/n^6$.

\subsection{Proof of Claim~\ref{claim: solve LP}}
\label{sec:solve_LP}


We rename (LP-1) by (LP-Primal-1).
Consider the following LP.
\begin{eqnarray*}
	\mbox{(LP-Primal-2)}	&\text{maximize } 0\\
	\mbox{s.t.}&&\\
	&\sum_{P\in \pset(u)}f(P)\geq k&\forall u\in U\\
	&\sum_{\stackrel{P\in \pset^*:}{e\in P}}f(P)\leq 2&\forall e\in E(G)\\
	&f(P)\geq 0&\forall P\in \pset^*
\end{eqnarray*}
It is clear that any feasible solution to (LP-Primal-1) is also a feasible solution to (LP-Primal-2), and vice versa. It is therefore sufficient to show that (LP-Primal-2) can be solved efficiently, if it has a feasible solution.  
Below is the Dual LP for (LP-Primal-2).
\begin{eqnarray*}
	\mbox{(LP-Dual-1)}	&\text{minimize } 2\cdot\sum_{e\in E(G)}\ell_e-k\cdot\sum_{u\in U}z_u\\
		\mbox{s.t.}&&\\
	&\sum_{e\in P}\ell_e\geq z_u+z_{u'}&\forall u,u'\in U: u\neq u', \forall P\in \pset(u)\cap\pset(u')\\
	&\sum_{e\in P}\ell_e\geq z_u&\forall u\in U, \forall P\in \pset(u)\cap\pset(s)\\
	&z_u\geq 0&\forall u\in U\\
	&\ell_e\geq 0&\forall e\in E(G)
\end{eqnarray*}

Recall that the number of vertices in $G$ is $n$.
Note that for (LP-Primal-2), the number of variables is exponential in $n$ and the number of constraints is polynomial in $n$, while for (LP-Dual-1), the number of variables is polynomial in $n$ and the number of constraints can be exponential in $n$.
From the strong duality, the optimal objective value of (LP-Dual-1) is $0$ if (LP-Primal-2) is feasible. We make a change to (LP-Dual-1) by replacing the objective function with a constraint that $2\cdot\sum_{e\in E(G)}\ell_e-k\cdot\sum_{u\in U}z_u=0$ to get the following LP.
\begin{eqnarray*}
	\mbox{(LP-Dual-2)}
	&&\\
	&2\cdot\sum_{e\in E(G)}\ell_e-k\cdot\sum_{u\in U}z_u=0\\
	&\sum_{e\in P}\ell_e\geq z_u+z_{u'}&\forall u,u'\in U: u\neq u', \forall P\in \pset(u)\cap\pset(u')\\
&\sum_{e\in P}\ell_e\geq z_u&\forall u\in U, \forall P\in \pset(u)\cap\pset(s)\\
	&z_u\geq 0&\forall u\in U\\
	&\ell_e\geq 0&\forall e\in E(G)
\end{eqnarray*}

\begin{claim}
\label{claim:sep_oracle}
There exists an efficient separation oracle to (LP-Dual-2).
\end{claim}

We provide the proof of Claim~\ref{claim:sep_oracle} below, after we show that there is an efficient algorithm that solves (LP-Primal-2) using it. We run the Ellipsoid Algorithm on (LP-Dual-2) using the separation oracle, and let $\cset$ be the set of all violated constraints that the oracle returns. Note that, since the running time of the Ellipsoid Algorithm is polynomial in the number of variables, when we run the Ellipsoid Algorithm on (LP-Dual-2), the size of $\cset$, which is the number of violated constraints returned by the separation oracle, is at most polynomial in $n$. Let (LP-Dual-3) be a linear program 
whose set of constraints is precisely $\cset$.  Note that the linear program (LP-Dual-3) is feasible iff the linear program (LP-Dual-2) is feasible. 
This is because, if we run the Ellipsoid Algorithm on (LP-Dual-3), then the separation oracle will return the same set of constraints and the algorithm will return the same solution or report infeasible (if it reports infeasible on (LP-Dual-2)). 
We now compute the dual of (LP-Dual-3) and obtain a linear program that we denote by (LP-Primal-3). It is not hard to see that
(LP-Primal-3) contains a subset (whose size is polynomial in $n$) of variables of (LP-Primal-2), and that for every constraint of (LP-Primal-2), there is a constraint in (LP-Primal-3), with the variables which are not in that subset omitted.
 From the strong duality, (LP-Primal-3) is feasible if (LP-Primal-2) is feasible. We can now solve (LP-Primal-3) efficiently, and the resulting solution is a feasible solution to (LP-Primal-2), as this is the same as setting all variables that do not correspond to the constraints in $\cset$ to $0$.
This finishes the proof of Claim~\ref{claim: solve LP}.

\emph{Proof of Claim~\ref{claim:sep_oracle}:}
We now show that there exists a separation oracle to (LP-Dual-2). Given a suggested solution to (LP-Dual-2), the separation oracle needs to check if it satisfies all the constraints of (LP-Dual-2), and if not, return a violated constraint.

Let $\{z_{u}\}_{u\in U}, \{\ell_e\}_{e\in E(G)}$ be the suggested solution in an iteration.
It is immediate to check whether the constraints $z_u\geq 0\text{ }\forall u\in U$, the constraints $\ell_e\geq 0\text{ }\forall e\in E(G)$ and the constraint $2\cdot\sum_{e\in E(G)}\ell_e-k\cdot\sum_{u\in U}z_u=0$
are satisfied. We will now show an efficient algorithm that checks whether the suggested solution satisfies the constraints $\sum_{e\in P}\ell_e\geq z_u+z_{u'}\text{ }\forall u,u'\in U: u\neq u', \forall P\in \pset(u)\cap\pset(u')$ and the constraints $\sum_{e\in P}\ell_e\geq z_u \text{ }\forall u\in U, \forall P\in \pset(u)\cap\pset(s)$ efficiently.

We assign each edge $e\in E(G)$ length $\ell_{e}$. For any path $P$ of $G$, we denote $\ell(P)=\sum_{e\in P}\ell_e$. Note that $U\subseteq V(G)$. We show an algorithm, that, given the suggested solution $\{z_{u}\}_{u\in U}, \{\ell_e\}_{e\in E(G)}$, either claims (correctly) that all constraints 
$\sum_{e\in P}\ell_e\geq z_u+z_{u'}\text{ }\forall u,u'\in U: u\neq u', \forall P\in \pset(u)\cap\pset(u')$
and all constraints $\sum_{e\in P}\ell_e\geq z_u\text{ }\forall u\in U, \forall P\in \pset(u)\cap\pset(s)$
are satisfied, or returns a pair $u,u'$ of distinct vertices of $U$ and a path $\hat{P}_{u,u'}\in \pset(u)\cap\pset(u')$, such that $\ell(\hat{P}_{u,u'})<z_u+z_{u'}$ (which means that the constraint 
$\sum_{e\in \hat{P}_{u,u'}}\ell_e\geq z_u+z_{u'}$ 
is not satisfied by the suggested solution),
or returns a vertex $u\in U$ and a path $\hat{P}_{u,s}\in \pset(u)\cap\pset(s)$, such that $\ell(\hat{P}_{u,s})<z_u$ (which means that the constraint 
$\sum_{e\in \hat{P}_{u,s}}\ell_e\geq z_u$ 
is not satisfied by the suggested solution).

\begin{claim}
\label{claim:dp}
There is an efficient algorithm, that, given any pair $v,v'$ of vertices of $G$, computes the shortest path (with respect to edge lengths $\{\ell_e\}_{e\in E(G)}$) connecting $v$ to $v'$ that contains at most $2D$ edges.
\end{claim}

We will prove Claim~\ref{claim:dp} below, after we complete the proof of Claim~\ref{claim:sep_oracle} using it.
For every pair $u,u'\in U$ of distinct vertices of $U$, let $\hat{P}_{u,u'}$ be the path returned by the algorithm in Claim~\ref{claim:dp}, we check if $\ell(\hat{P}_{u,u'})<z_u+z_{u'}$.
For every vertex $u\in U$, let $\hat{P}_{u,s}$ be the path returned by the algorithm in Claim~\ref{claim:dp}, we check if $\ell(\hat{P}_{u,s})<z_u$.
If there exists a pair $u,u'\in U$ of distinct vertices of $U$ such that $\ell(\hat{P}_{u,u'})<z_u+z_{u'}$, by definition, $\sum_{e\in \hat{P}_{u,u'}}\ell_e=\ell(\hat{P}_{u,u'})<z_u+z_{u'}$. In this case, we claim that the constraint $\sum_{e\in\hat{P}_{u,u'}}\ell_e\ge z_u+z_{u'}$ is violated, and return this constraint as a violated constraint.
If there does not exist a pair $u,u'\in U$ of distinct vertices of $U$ such that $\ell(\hat{P}_{u,u'})<z_u+z_{u'}$, then from Claim~\ref{claim:dp}, for any pair $u,u'$ of distinct vertices of $U$, for any path $P\in \pset(u)\cap\pset(u')$, we have $\sum_{e\in P}\ell_{e}\ge z_u+z_{u'}$. 
In this case, we know that all constraints $\sum_{e\in P}\ell_e\geq z_u+z_{u'}\text{ }\forall u,u'\in U: u\neq u', \forall P\in \pset(u)\cap\pset(u')$ are satisfied, so we then proceed to check if there exists a vertex $u\in U$ such that $\ell(\hat{P}_{u,s})<z_u$. If there does exists such a vertex $u$, by definition, $\sum_{e\in \hat{P}_{u,s}}\ell_e=\ell(\hat{P}_{u,s})<z_u$. In this case, we claim that the constraint $\sum_{e\in\hat{P}_{u,s}}\ell_e\ge z_u$ is violated, and return this constraint as a violated constraint. If there does not exist a vertex $u\in U$ such that $\ell(\hat{P}_{u,s})<z_u$, then from Claim~\ref{claim:dp}, for vertex $u\in U$, for any path $P\in \pset(u)\cap\pset(s)$, we have $\sum_{e\in P}\ell_{e}\ge z_u$. We then claim that all constraints $\sum_{e\in P}\ell_e\geq z_u+z_{u'}\text{ }\forall u,u'\in U: u\neq u', \forall P\in \pset(u)\cap\pset(u')$ and all constraints $\sum_{e\in P}\ell_e\geq z_u\text{ }\forall u\in U, \forall P\in \pset(u)\cap\pset(s)$ are satisfied.

This finishes the description of the separation oracle to (LP-Dual-2). Since it is clear that the running time of the separation oracle is polynomial in $n$, this finishes the proof of Claim~\ref{claim:sep_oracle}.

\emph{Proof of Claim~\ref{claim:dp}:}
The algorithm employs dynamic programming. It is convenient to view the algorithm as constructing $2D+1$ dynamic programming tables $\{\Pi_{i}\}_{0\le i\le 2D}$. For each $0\le i\le 2D$ and each pair $v,v'$ of vertices of $G$, the table $\Pi_i$ contains an entry $\Pi_i(v,v')$, that stores the shortest path $P^i_{v,v'}$ (with respect to edge lengths $\{\ell_e\}_{e\in E(G)}$) among all paths in $G$ that connects $v$ to $v'$ and contains at most $i$ edges, together with its length $\ell(P^i_{v,v'})$. So each entry $\Pi_i(v,v')$ has the form $\Pi_i(v,v')=(P^i_{v,v'},L^i_{v,v'})$ where $L^i_{v,v'}=\ell(P^i_{v,v'})$. When such a path does not exist, we set $P^i_{v,v'}$ to be a default value $\perp$ and set $L^i_{v,v'}=+\infty$.
 
We now describe how to compute the entries of dynamic programming tables. First we initialize the entries in $\Pi_0$.
For each vertex $v$, we set $P^0_{v,v}$ to be the path that contains a single node $v$, and we set $L^0_{v,v}=0$.
For each pair $v,v'$ of distinct vertices of $G$, we set $P^0_{v,v}=\perp$ and $L^0_{v,v'}=+\infty$. 
For each $1\le i\le 2D$, the table $\Pi_i$ is computed based on $G$ and the table $\Pi_{i-1}$ as follows. For each vertex $v\in V(G)$, we denote $N(v)\subseteq V(G)$ to be the set of neighbors of $v$ in $G$. For each pair $v,v'\in V(G)$, we set
\[L^i_{v,v'}=\min\{L^{i-1}_{v,v'}, \min_{w\in N(v)}\{\ell_{(v,w)}+L^{i-1}_{w,v'}\}\}.\]
For $P^i_{v,v'}$, we set it to be $\perp$ if $L^i_{v,v'}=+\infty$;
we set it to be the same path as $P^{i-1}_{v,v'}$ if 
$L^{i-1}_{v,v'}\le \min_{w\in N(v)}\{\ell_{(v,w)}+L^{i-1}_{w,v'}\}$; and  if $w'=\arg\min\{L^{i-1}_{v,v'}, \min_{w\in N(v)}\{\ell_{(v,w)}+L^{i-1}_{w,v'}\}\}$ and $\ell_{(v,w')}+L^{i-1}_{w',v'}<L^{i-1}_{v,v'}$, we set it to be the concatenation of the edge $(v,w')$ and the path $P^{i-1}(w',v')$.

Finally, given a pair $v,v'$ of vertices of $G$, we return the path $P^{2D}_{v,v'}$ if $P^{2D}_{v,v'}\ne \perp$, and we claim that such a path does not exist if $P^{2D}_{v,v'}=\perp$. \qed
\qed

%% file: applications.tex
\section{Applications to Distributed Computation}
\label{sec:dist}
In this section, we provide applications of our graph theoretic results to distributed and secure computation, proving
Theorems \ref{lem:lambda}, \ref{thm:mstdist} and \ref{lem:dist-inf-diss}.
Throughout, we use the standard \congest\ model \cite{Peleg:2000}, where the algorithm's execution proceeds in synchronous 
rounds, and in every round, each node can send a message of size $O(\log n)$ to 
each of its neighbors. Each node holds a processor with a unique and arbitrary ID of $O(\log n)$ bits.
As common in this model, we restrict attention to simple graphs with no parallel edges.
Our algorithms make extensive use of the random delay approach of \cite{leighton1994packet,ghaffari2015near}.
\begin{theorem}[{\cite[Theorem 1.3]{ghaffari2015near}}]\label{thm:delay}
Let $G$ be a graph and let $A_1,\ldots,A_m$ be $m$ distributed algorithms in 
the \congest model, 
where each algorithm takes at most $\dilation$ rounds, and where for each 
edge $e\in E(G)$, the total number of messages sent over $e$ by all these algorithms is at most $\congestion$. Then, there is a randomized distributed
algorithm (that uses private randomness), that, with high probability, 
produces 
a schedule that runs all the algorithms in $O(\congestion +\dilation \cdot 
\log 
n)$ rounds, after $O(\dilation \log^2 n)$ rounds of pre-computation.
\end{theorem}

Throughout, we assume that $k=\Omega(\log n)$ and consider a $k$-edge connected $n$-vertex graph $G=(V,E)$. All presented algorithms are randomized, and their correctness hold with probability at least $1-1/n^c$ for some constant $c$ (that we refer to as \emph{high probability}). 
The starting point for all the applications considered in this section is the computation of $\Omega(k)$ subgraphs $G_1,\ldots, G_k$ of $G$ with bounded congestion, such that each subgraph has a small diameter. The subgraphs $G_i$ are given in a distributed manner where each edge $(u,v)$ knows the indices of the subgraphs $G_i$ to which it belongs.

\begin{claim}[Basic Distributed Tool]\label{cl:basic-dist-tool}
There is a randomized algorithm that, given a $k$-edge connected $n$-vertex graph $G$ and a congestion bound $\eta \in [1,k]$, computes,  in $\widetilde{O}((101k\ln n/\eta)^{D})$ rounds, a collection of $k$ spanning trees that cause total edge-congestion at most $O(\eta\cdot\log n)$, and have diameter at most $O((101k\ln n/\eta)^{D})$ each. 
Moreover, the algorithm can compute $k$ spanning subgraphs with similar congestion and diameter bounds in $O(D+\eta\log n)$ rounds. The round complexity, the diameter, and the congestion bounds hold with high probability.
\end{claim}
\begin{proof}
Let $T$ be a BFS tree of the graph $G$ of depth at most $D$, computed from an arbitrary source vertex $s \in G$.
The algorithm computes a collection of $k$ subgraphs $G_1,\ldots, G_k$ which will be shown to cause bounded congestion and have bounded diameter. For every $i \in \{1,\ldots,k\}$ in parallel, each subgraph $G_i$ is computed by sampling each edge $e\in G$ into $G_i$ with probability $p=707\log n/k$, and additionally sampling each edge $e'\in T$ into $G_i$ with probability $\eta/k$, independently from all other edges. In other words, $G_i=G[p]\cup T[\eta/k]$. 
In the distributed setting, the edge sampling is made by the edge endpoint of larger ID. Each node $u$ sends its lower-ID neighbor $v$ the indices $i$ such that edge $e=(u,v)$ is in $G_i$. 
Next, the algorithm computes a truncated BFS tree $T_i$, up to depth $d=O((101 k\ln n/\eta)^D)$, in every sampled subgraph $G_i$ in parallel, using the random delay approach from Theorem \ref{thm:delay}. 

In order to analyze this algorithm, we start by showing that w.h.p., the diameter of each subgraph $G_i$ is bounded by $d$. Indeed, by Theorem \ref{thm:Karger_diameter}, each subgraph $G[p]$ for $p=707\log n/k$ is connected with high probability. By Theorem \ref{thm:random_tree_planting}, the diameter of each subgraph $G_i=G[p] \cup T[\eta/k]$ is at most $O((101 k\ln n/\eta)^D)$ with high probability. Next, we bound the congestion. A simple application of the Chernoff bound shows that, with high probability, each edge of $G$ appears in at most $O(\eta\log n)$ subgraphs. 
A single BFS computation up to depth $d$ takes $O(d)$ rounds, while sending $O(1)$ messages on each of the graph edges.
Thus, by applying the random delay approach, one can compute all $k$ BFS trees in $G_1,\ldots, G_k$ in $\widetilde{O}(d+\eta)$ rounds. Since the diameter of each graph $G_i$ is at most $d$, all resulting trees are indeed spanning with high probability. Note that, if we only need to compute spanning subgraphs of $G$, then the collection $\set{G_1,\ldots, G_k}$  of such subgraphs can indeed be computed in $O(D+\eta\cdot \log n)$ rounds. To see this observe that the BFS computation can be done in $O(D)$ rounds. The larger-ID endpoint $u$ of each edge $(u,v)$ has $\eta\log n$ messages to send to its endpoint $v$ containing the indices of the subgraphs to which $(u,v)$ belongs.
\end{proof}

\input{application-mincutmst.tex}

\input{application-shortcuts.tex}

\input{application-information.tex}

\input{application-secure.tex}

\subsection{Useful Lemmas for Bounded Independence}
We first need the notion of $d$-wise independent hash functions as presented in \cite{Vadhan12}.
\begin{definition}[Definition 3.31 in \cite{Vadhan12}]
\label{def: d-wise independent}
For	$N,M,d \in \mathbb{N}$ such that $d \leq N$, a family of functions $\mathcal{H} = \set{h : [N] \rightarrow [M]}$ is $d$-wise independent if for all distinct $x_1,x_2,...,x_d \in [N],$ the
random variables $H(x_1),...,H(x_d)$ are independent and uniformly distributed
in $[M]$ when $H$ is chosen randomly from $\mathcal{H}$.
\end{definition}
Vadhan \cite{Vadhan12} presented an explicit construction of $\mathcal{H}$, with the following parameters.
\begin{lemma}[Corollary 3.34 in \cite{Vadhan12}]
\label{lem: d-wise independent}
For every $\gamma,\beta,d \in \mathbb{N},$ there is a family of $d$-wise independent functions $\mathcal{H}_{\gamma,\beta} = \set{h : \set{0,1}^\gamma \rightarrow \set{0,1}^\beta}$ such that choosing a random function from $\mathcal{H}_{\gamma,\beta}$ takes $d \cdot \max \set{\gamma,\beta}$ random bits, and evaluating
a function from $\mathcal{H}_{\gamma,\beta}$ takes time $\poly(\gamma,\beta,d)$.
\end{lemma}

We use the following Chernoff bound for $d$-wise independent random variables from \cite{SchmidtSS95}.
\begin{theorem}
\label{thm:d-wise chernoff}
Let $X_1,...,X_n$ be $d$-wise independent random variables taking values in $[0,1]$, where $X = \sum_{i=1}^n X_i$ and $\mathbb{E}[X]=\mu$. Then for all $\epsilon \leq 1$ we have that if $d \leq \lfloor \epsilon^2 \mu e^{-1/3} \rfloor$ then:
$$\Pr[|X- \mu | \geq \epsilon \mu] \leq e^{-\lfloor d/2 \rfloor}.$$
And if $d > \lfloor \epsilon^2 \mu e^{-1/3} \rfloor$ then:
$$\Pr[|X- \mu | \geq \epsilon \mu] \leq e^{-\lfloor \epsilon^2 \mu /3 \rfloor}.$$
\end{theorem}

%% file: application-mincutmst.tex
\subsection{Distributed Approximate Verification of the Edge Connectivity}
We show the following immediate application of Theorem \ref{thm:Karger_diameter} to verify if the graph is $\lambda$-edge connected,  up to approximation factor $O(\log n)$. Given a graph $G=(V,E)$ and integer $\lambda$, if the graph is $\lambda$-edge connected then all nodes must YES, and if the graph is at most $\lambda/\log n$ connected, all nodes must output NO. The algorithm succeeds with high probability in $\widetilde{O}((\lambda\log^2 n)^{D(D+1)/2})$ rounds.

\begin{theorem}[$O(\log n)$-Approximate Verification of $\lambda$-Edge Connectivity ]\label{lem:approx-cut-upper-bound}
There is a randomized distributed algorithm, that, given an unweighted $n$-vertex graph $G=(V,E)$ of diameter $D$, and an integer $\lambda$,
ensures that with high probability, after $\widetilde{O}((\lambda\log^2 n)^{D(D+1)/2})$ rounds, if $G$  is $\lambda$-edge connected, then all nodes output YES, and if  it is at most $\lambda/\log n$-edge connected, then all nodes output NO. 
\end{theorem}\vspace{-5pt}
\begin{proof}
Let $p=707\log n/\lambda$. 
By Theorem~\ref{thm:Karger_diameter}, if the graph $G$ is $\lambda$-edge connected, then the sampled graph $G[p]$ is connected with high probability. Moreover, the diameter of $G[p]$ is bounded by $O(\lambda^{D(D+1)/2})$ with high probability.  
On the other hand, if the graph connectivity is $\lambda' \leq \lambda/\log n$, then $G[p]$ is connected with probability at most $3/4$. We will then make $O(\log n)$ edge-sampling experiments to distinguish between these two scenarios.  

For every $j \in [1,\Theta(\log n)]$, let $G_j=G[p]$, i.e., sample each edge in $G_j$ independently with probability $p$. Compute a truncated BFS tree up to depth $\lambda^{D(D+1)/2}$. If this tree spans all vertices of $V(G)$, then we say that graph $G_j$ is \emph{good}. The algorithm returns YES if at least 0.9 of the experiments are good. 

The round complexity is simply $O(\lambda^{D(D+1)/2}\cdot \log n)$. We now consider correctness. If the graph is $\lambda$-edge connected, w.h.p. all experiments are good and therefore all nodes say YES.
If the graph is $\lambda'$-connected for $\lambda' \leq\lambda/\log n$, then when sampling the edges with probability at most $1/\lambda'$ the graph is connected with probability at most $3/4$. The theorem follows by a simple application of Chernoff bound.
\end{proof}

With a slight modification, the algorithm from the above lemma can also be used to obtain an $O(\log n)$-approximation on the size of the minimum cut $\lambda$ in $G$, in $\widetilde{O}((\lambda\log^2 n)^{D(D+1)/2})$ rounds w.h.p.
\begin{corollary}\label{cor:min-cut-value-approx}
There is a randomized distributed algorithm, that, given an unweighted $n$-vertex graph $G=(V,E)$ of diameter $D$, computes an estimate $\widetilde{\lambda}$ on the value $\lambda$ of the global minimum cut in $G$, such that $\widetilde{\lambda}\in [\lambda, \lambda \cdot O(\log n)]$, in $\widetilde{O}((\lambda\log^2 n)^{D(D+1)/2})$ rounds. Both the correctness and the round complexity of the algorithm hold with high probability.
\end{corollary}
\begin{proof}
The algorithm considers the values $1,2,4,\ldots$ of $\lambda'$ one by one. For each such value, it applies the algorithm for $\lambda'$-edge connectivity verification, until the first value $\lambda'$ is encountered on which the verification algorithm returns ``No''. The algorithm then terminates and returns this value of $\lambda'$.
\end{proof}

\paragraph{Separation between MST and Approximate Minimum Cut.}
In the distributed graph theory literature, the problems of approximating the global minimum cut of a graph, and of computing an MST are considered to be more or less ``equivalent" in terms of their round complexities in general graphs. In fact, the classical algorithms for minimum cut are based on repeated application of MST computation.
As we will show, this is no longer the case when we consider moderately highly connected low-diameter graphs. 
In order to illustrate this gap, we compare the round complexities of computing an MST, and of approximating the value of the global minimum cut in $k$ edge-connected graphs of diameter $4$. Lotker, Pat-Shamir and Peleg \cite{LotkerPP06} showed that computing an MST in an $n$-vertex graphs of diameter $4$ may require $\Omega((n/\log n)^{1/3})$ rounds. Although their construction is a $2$-edge connected graph, we show that a slight modification of their construction gives a $k$-edge connected for any $k=O(n^{1/4})$. We prove the following theorem.
\begin{theorem}[MST Lower Bound in $k$-Connected Graphs]
For every large enough integer $n$ and an integer $k=O(n^{1/4})$, there exists an $n$-vertex $k$-edge connected graph of diameter $4$, for which computing an MST requires $\widetilde{\Omega}((n/k)^{1/3})$ rounds.
\end{theorem}

\begin{proof}
Our starting point is a graph $F_m$, that was used in the lower bound proof of \cite{LotkerPP06}. 
The graph consists of a root vertex $c$, and a collection $U=\set{u_1,\ldots,u_m}$ of additional vertices; we denote $s=u_1$ and $r=u_m$. Every vertex $u_i$ is connected to the root vertex $c$ with an edge. Additionally, the graph contains a collection $\pset=\set{P_1,\ldots, P_{m^2}}$ of $m^2$ disjoint paths of length $m$ each. For all $1\leq j\leq m^2$, we denote by $v^j_i$ the $i$th vertex on path $P_j$. For all $1\leq i\leq m$, we connect the vertex $u_i$ to all vertices $v^1_i,\ldots,v^{m^2}_i$.

Next, we modify the graph $F_m$ to make it $k$-edge connected; the resulting graph is denoted by $F_{m,k}$. 
In order to obtain the graph $F_{m,k}$, we start from the graph $F_m$. For all $1\leq j\leq m^2$ and $1\leq i\leq m$, we replace the vertex $v^j_i$ with a $k$-clique $V_{j,i}$. We connect the vertex $u_i$ to every vertex of $V_{j,i}$.
For all $1\leq j\leq m^2$ and $1\leq i < m$, we add an arbitrary perfect matching between the vertices of $V_{j,i}$ and the vertices of $V_{j,i+1}$; formally, if we denote the vertices of $V_{j,i}$ by $v_{j,i,1},\ldots, v_{j,i,k}$, then we add the edges between
$(v_{j,i,\ell})$ and $(v_{j,i+1,\ell})$ for every $1\leq j\leq m^2$, $1\leq i\leq m-1$, and $\ell \in \{1,\ldots,k\}$ (see Figure \ref{fig:MST-LB} for an illustration).

For all $1\leq i\leq m$, we let $S_i$ be the star graph, that is a sub-graph of the resulting graph $F_{m,k}$, induced by the vertex $u_i$, and all vertices in sets $V_{j,i}$, for $1\leq j\leq m^2$.

 We set $m=\lceil (n/k)^{1/3} \rceil$, so graph $F_{m,k}$ has $\Theta(n)$ vertices and diameter $4$. 
 We claim that graph $F_{m,k}$ is $k$-edge connected. Indeed, let $(X,Y)$ be any partition of the vertices of $F_{m,k}$ into two subsets. Assume for contradiction that $|E(X,Y)|<k$. Notice first that for all $1\leq i<i'\leq m$, there are at least $k$ edge-disjoint paths in graph $F_{m,k}$ connecting $u_i$ to $u_{i'}$. Therefore, all vertices $u_1,\ldots,u_m$ must lie on the same side of the cut. Assume w.l.o.g. that it is $X$. It is then easy to verify that, since $m\geq k$ (as $k=O(n^{1/4})$), vertex $c$ must also lie in $X$. Consider now some set $V_{j,i}$ of vertices, for some $1\leq j\leq m^2$ and $1\leq i\leq m$. Every vertex of $V_{j,i}$ is connected by an edge to the vertex $u_i$. Therefore, not all vertices of $V_{j,i}$ lie in $Y$. Moreover, for any partition of the vertices of $V_{j,i}$ into two subsets, at least $k-1$ edges must connect the two subsets (as every pair of vertices in a $k$-clique has $(k-1)$ edge-disjoint paths connecting them). Therefore, if $Y$ contains any vertex of $V_{j,i}$, then $|E(X,Y)|\geq k$ must hold (at least $k-1$ edges must connect vertices of $V_{j,i}$ lying on different sides of the cut, and additionally every vertex of $V_{j,i}\cap Y$ is connected to $u_i$.) Since $Y\neq \emptyset$, it must contain a vertex from some set $V_{i,j}$, and so $|E(X,Y)|\geq k$ must hold, a contradiction. We conclude that $F_{m,k}$ is $k$-edge connected.


We next turn to prove a lower bound on the number of rounds for computing an MST in the graph $F_{m,k}$. The proof follows that of \cite{LotkerPP06} almost exactly; we provide it here  for completeness.
The key idea in the proof of \cite{LotkerPP06} is to consider the \emph{mailing problem}: given the graph $F_{m,k}$, the source vertex $s=u_1$ is required to send an input set $X$ of $m^2$ bits to the destination vertex $r=u_m$. We adapt the argument of \cite{LotkerPP06} to show that this requires $\widetilde{\Omega}((n/k)^{1/3})$ rounds in $F_{m,k}$. Then, we deduce a lower bound on computing an MST in $F_{m,k}$, using the reduction provided in \cite{LotkerPP06}.

Given a graph $G$, a sender vertex $s$ and a receiver vertex $r$, together with an input $b$-bit string $X=x_1,\ldots, x_b$, we denote by  $\mail(G,s,r,X)$ the mailing problem of sending the string $X$ from $s$ to $r$. We denote by $\Mail(G,s,r,b)$ the collection of all problems $\mail(G,s,r,X)$, where $X$ is a string of length $b$. Given an algorithm $A$, we denote by $T_A(G,s,r,X)$ the number of rounds the algorithm takes to solve problem $\mail(G,s,r,X)$, and we denote by $T^b_A(G,s,r)$ the maximum, over all $b$-bit strings $X$, of $T_A(G,s,r,X)$.

\paragraph{Lower Bound for the Mailing Problem.} 

\begin{claim}\label{cl:det}
For any deterministic algorithm $A$ for the mailing problem in the \congest\ model, and any $m\geq 2$, $T^{m^2}_A(F_{m,k},s,r) = \Omega(m/\sqrt{\log n})$.
\end{claim}

\begin{proof}
For every $1\leq i\leq m$, we define a graph $Z_i$, called \emph{the $i$th tail of the graph $F_{m,k}$}. Graph $Z_i$ is the subgraph of $F_{m,k}$ induced by the nodes of $\{V_{j,\ell} ~\mid~ 1 \leq j\leq m^2,  i+1 \leq \ell\leq m\} \cup \{u_\ell ~\mid~ i+1 \leq \ell \leq m\} \cup \{c\}$. We let the $0^{th}$ tail $Z_0$ be the subgraph of $F_{m,k}$ induced by $V(F_{m,k})\setminus \{s\}$.

We now fix a deterministic algorithm $A$. Let $\phi_X$ denote the execution of $A$ on an $m^2$-bit input $X$ in the graph $F_{m,k}$ with sender $s$ and receiver $r$. Let $C_t(X)$ denote the vector of states of the nodes in the tail graph $Z_t$ at the end of round $t$ in the execution $\phi_X$; we refer to $C_t(X)$ as \emph{the configuration of $Z_t$ on $X$ in round $t$}. 

Define $\mathcal{C}_t=\{C_t(X) ~\mid~ X \mbox{~is an $m^2$-bit string}\}$ and let $\rho_t=|\mathcal{C}_t|$ be the number of distinct reachable configurations of $Z_t$ in round $t$.

\begin{claim}
	For all $0\leq t<m$,  $\rho_t \leq 2^{t(t+1)B/2}$, where $B=O(\log n)$ is the bandwidth of each edge in the \congest\ model. 
\end{claim}
\begin{proof}
	Observe first that for $t=0$, since the input string is known only to the sender $s$, all other nodes are in their initial states and thus $\rho_0=1$. 
	
	Next, we show that for all $t>0$, $\rho_{t+1}\leq \rho_t \cdot 2^{B\cdot t}$.
Indeed, consider some $t>0$, and a configuration $\hat{C} \in \mathcal{C}_t$. Recall that the tail set $V(Z_{t+1})$ is connected to the remainder of the graph by two sets of edges: (i) edges connecting cliques $V_{t,j}$ and $V_{t+1,j}$ for every $1\leq j \leq m^2$ -- we call them type-1 edges; 
and (ii) edges $(u_\ell,c)$ for $1\leq \ell \leq t$ -- we call them type-2 edges. 
Clearly, the number of type-$2$ edges is $t$.

 We now count the number of different configurations in $\mathcal{C}_{t+1}$ that may arise from the single configuration $\hat{C}\in \mathcal{C}_t$.
 
  The key observation is that, since the state of each node in $Z_t$ is determined by $\hat{C}$, the messages sent from the nodes of $Z_t$ to the nodes of $Z_{t+1}$ are fully determined by $\hat{C}$. In particular, all messages sent via type-1 edges, and via edges internal to $Z_{t+1}$ are completely determined by $\hat C$. The only additional messages are those sent along the $t$ type-2 edges. Since each such edge may carry at most $B$ bits, the total number of distinct messages sent along such edges is bounded by $2^{Bt}$. Therefore, at most $2^{Bt}$ different configuration in $\cset_{t+1}$ may arise from a single configuration in $\cset_t$, and so $\rho_{t+1}\leq 2^{Bt}\cdot \rho_t$.
  
  We conclude that for all $t\geq 0$, $\rho_t\leq 2^B\cdot 2^{2B}\cdots2^{(t-1)B}\leq 2^{t(t+1)B/2}$.
\end{proof}
	
Let $R=	T^{m^2}_A(F_{m,k},s,r) $; our goal is to show that $R= \Omega(m/\sqrt{B})$. 
If $R\geq m$, then we are done, so assume that $R<m$. Then there must be at least $2^{m^2}$ possible different states for the receiver $r$ at round $R$, so $\rho_R\geq 2^{m^2}$ must hold. Since we have assumed that $R< m$, we get that $2^{BR(R+1)/2}\geq 2^{m^2}$ must hold, that is $R=\Omega(m/\sqrt{B})$ as required.
\end{proof}

\paragraph{Extension to randomized algorithms.} Using standard
techniques, one can show that all Las-Vegas algorithms for the mailing problem admit the
same asymptotic lower bounds as deterministic algorithms. The proof is based on fixing a deterministic distributed mailing algorithm $A$ and establishing a slightly stronger claim.
\begin{claim}
For every $m\geq 2$ and for at least half of the possible $m^2$-bit input string $X$  of the mailing problem $T_A(F_{m,k},s,r, X) \geq \alpha m/\sqrt{B}$ for some constant $\alpha>0$ and $B=O(\log n)$.
\end{claim}
To see this simply define $\tau_{mid}$ to be the minimum integer such that $T_A(F_{m,k},s,r,X)\leq \tau_{mid}$ for at least half of the possible input string. We then have that $2^{B \cdot \tau_{mid}\cdot (\tau_{mid}+1)/2}\geq \rho_{\tau_{mid}}\geq 2^{m^2-1}$ and so $\tau_{mid}=\alpha m/\sqrt{B}$. Since the algorithm takes at least $\tau_{mid}$ rounds on half of the inputs, its expected running time over the uniform distribution of all instances in $\mail(F_{m,k},s,r,X)$ is $\Omega(\tau_{mid})$. The randomized round complexity follows by applying the Yao principle.

\paragraph{Reduction to MST.} We follow the exact same scheme of \cite{LotkerPP06}. 
We start with the unweighted base graph $F_{m,k}$, and define a family of weighted graphs $\mathcal{F}_{m,k}$ such that, if the mailing problem requires $\Omega(t)$ rounds on $F_{m,k}$, then algorithm for computing MST requires at least $\Omega(t)$ rounds on some graph in $\mathcal{F}_{m,k}$.

The edge weights of the graphs in $\mathcal{F}_{m,k}$ are set as follows. 
All edge weights, except for the edges that belong to the stars $S_1,\ldots,S_m$, are set to $0$. The edges of $S_m$ have weight $1$, and the edges of the stars $S_2,\ldots,S_{m-1}$ have weight $10$. Consider now the star $S_1$. For all $1\leq j\leq m^2$, let $v^j$ be an arbitrary distinguished vertex from the clique $V_{j,1}$. Let $E_1$ be the set of all edges of $S_1$ that connect $u_1$ to the distinguished vertices $v^1,\ldots,v^{m^2}$, and let $E_2$ contain the remaining edges of $S_1$. We set the weight of every edge in $E_2$ to $10$. Lastly, every edge in $E^1$ is given a weight of either $0$ or $2$. Specifically, for every binary string $X$ of length $m^2$, we define a graph $F(X)\in \fset_{m,k}$, by appropriately setting the weights of the edges in $E_1$: for all $1\leq j\leq m^2$, if the $j$th bit of $X$ is $0$, then we set the weight of the edge $(u_1,v^j)$ to be $0$, and otherwise we set it to $2$.

Consider now a minimum spanning tree in a graph $F(X)$. For each $1\leq j\leq m^2$, we can use the $0$-weight edges to connect the set $\bigcup_{1\leq i\leq m}V_{j,i}$ of vertices to each other, obtaining a connected component $C_j$. Similarly, we can use the $0$-weight edges to connect the vertices $u_1,\ldots,u_m$ to the vertex $c$, obtaining a connected component $C_0$. For each $1\leq j\leq m^2$, we now need to connect $C_j$ to $C_0$. If the $j$th bit of $X$ is $0$, then the cheapest way to do so is to employ the edge $(u_1,v^j)$. Otherwise, the cheapest way to connect $C_j$ to $C_0$ is to use one of the edges of $S_m$ that is incident to a vertex of $C_j$, whose weight is $1$.

Therefore, by solving the MST problem, the endpoint $r$ learns the $m^2$-bit input of $s$.  
\end{proof}

By combining this lower bound with the upper bound of Theorem \ref{lem:approx-cut-upper-bound}, we have:
\begin{corollary}[Separation between Distributed Min-Cut and MST]
In every $n$-vertex graph of diameter $4$, the approximate verification of the $k$-edge connectivity of the graph can be done in $\widetilde{O}(k^{10})$ rounds. In contrast the computation of MST on $4$-diameter graphs requires $\widetilde{\Omega}((n/k)^{1/3})$ on a $k$-edge connected subgraph. Thus for $k=o(n^{1/32})$, approximate verification of the $k$-edge connectivity is strictly faster than computing an MST.  
%
\end{corollary}

\begin{figure}[h!]
\begin{center}
	\includegraphics[scale=0.5]{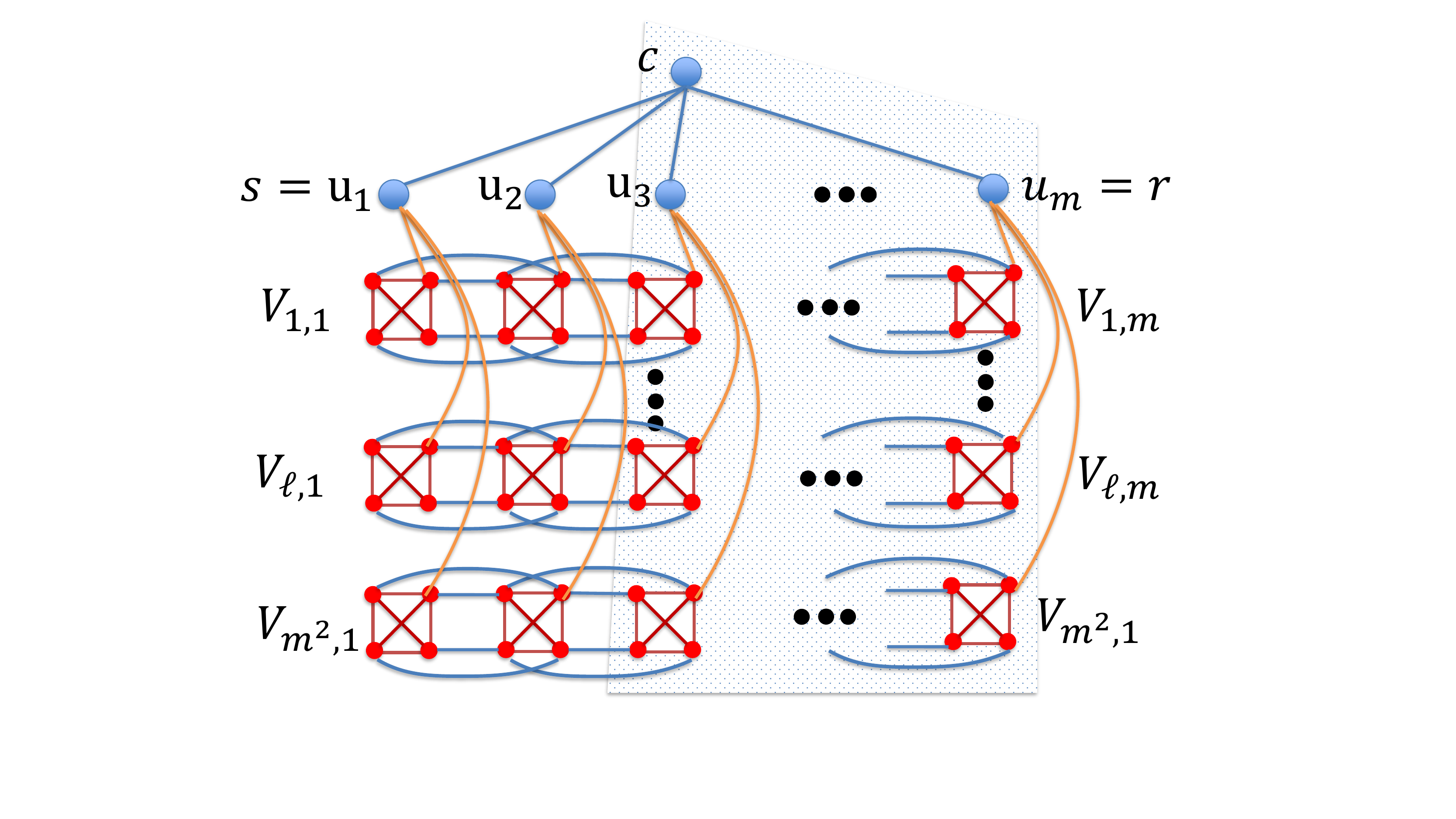}
	\caption{Illustration of the graph $F_{m,k}$, obtained by adapting of lower-bound graph $F_m$ from \cite{LotkerPP06}. Every vertex $u_i$ is connected to all clique vertices in $V_{i,j}$ for every $j \in \{1,\ldots, m^2\}$. For clarity of presentation, the figure shows only one edge to a particular clique member. The nodes in the shaded area are the tail sets $Z_2$.}
	\label{fig:MST-LB}
	\end{center}
\end{figure}

%% file: application-shortcuts.tex
\subsection{Improved Low-Congestion Shortcuts with Applications}
Low-congestion shortcuts, introduced by Ghaffari and Haeupler \cite{GhaffariH16}, is a basic communication backbone that is used in algorithms for several optimization problems in the \congest\ model.
We start by providing a formal definition, and then show a distributed computation of shortcuts of improved quality, that leverages the graph connectivity.
 
\begin{definition}[Low-Congestion Shortcuts, \cite{GhaffariH16}]
Given a graph $G=(V,E)$, and a partition  $S_1,\ldots, S_N$ of $V$ into disjoint subsets, such that for all $1\leq i\leq N$, graph $G[S_i]$ is connected, an $(\alpha,\beta)$-shortcut is a collection $\{H_1,\ldots, H_N\}$ of subgraphs of $G$,  that  satisfy the following: 
\begin{itemize}
\item{(1)} for each edge $e \in E$, there are at most $\alpha$ subgraphs $G[S_i]\cup H_i$ containing $e$; and
\item{(2)} the diameter of each subgraph $G[S_i]\cup H_i$ is at most $\beta$.
\end{itemize}
\end{definition}
Ghaffari and Haeupler \cite{GhaffariH16} showed that the quality of several algorithms depends on the sum of $\alpha$ (i.e., congestion) and $\beta$ (i.e., the dilation). The quantity of $\alpha+\beta$ is usually referred to as the \emph{quality} of the shortcuts.
As observed by \cite{GhaffariH16} for every $n$-vertex graph $G$ and any collection of vertex-disjoint subsets $S_1,\ldots, S_N$, there exist $(\alpha,\beta)$ shortcuts for with $\alpha+\beta=O(D+\sqrt{n})$. This is also tight due to Das-Sarma et al. \cite{sarma2012distributed}. Shortcuts with improved quality 
 are known to exist for planar graphs \cite{GhaffariH16}, graphs with bounded pathwidth or treewidth \cite{haeupler2016near}, graphs with excluded minor \cite{haeupler2018minor} and graphs with small mixing time \cite{ghaffari2017distributed,GhaffariL18}.

Our key result is in providing a nearly optimal construction for low-congestion shortcuts in highly connected graphs of constant diameter. This immediately leads to improvements in a number of network optimization tasks. The input to the shortcut algorithm is a partition $\{S_1,\ldots, S_N\}$ of $V(G)$,  given in a distributed manner, that is, each vertex $S_i$ knows the ID of the set $S_i$ to which it belongs, where the ID is the largest vertex ID in $S_i$. At the end of the algorithm, each vertex $v \in S_i$ knows  all its neighbors in the augmented subgraph $G[S_i] \cup H_i$. 

Throughout we assume that we are given a $k$-edge connected graph and that all nodes know $k$ (or even a logarithmic approximation $k' \in [k/\log n,k]$). Alternatively, the nodes can first compute a $O(\log n)$-approximation of the size of the minimum cut within $\widetilde{O}((k\log^2 n)^{D(D+1)/2})$ rounds w.h.p. using a Cor. \ref{cor:min-cut-value-approx}. 
\begin{theorem}\label{lem:dist-shortcuts-construction}[Improved Shortcuts in Highly Connected Graphs]
There is a randomized algorithm that, 
for a sufficiently large $n$, given any $k$-connected $n$-vertex graph $G$ of diameter $D=O(\log n/\log\log n)$, together with a partition $\{S_1,\ldots, S_N\}$ of $V(G)$, such that for all $1\leq i\leq N$, $G[V_i]$ is a connected graph, computes $(\alpha,\beta)$ shortcuts, with $$\alpha+\beta=\widetilde{O}(\min\{\sqrt{n/k}+n^{D/(2D+1)}\},n/k),$$
in $\widetilde{O}(\alpha+\beta)$ rounds. Both the round complexity and the correctness hold with high probability.
\end{theorem}
In the remainder of this subsection, we prove Theorem \ref{lem:dist-shortcuts-construction}.

\paragraph{Warmup: Shortcuts with $\alpha=2$ and $\beta=O(n/k)$.} We first consider the simpler case of obtaining nearly-edge disjoint shortcuts of diameter $O(n/k)$. 
Given the input collection $S_1,\ldots, S_N$ of subsets of $V(G)$, for each $1\leq i\leq N$, we define the graph $H_i$ as follows: graph $H_i$ contains all vertices of $S_i$, and all neighbors of vertices of $S_i$ in $G$. The set of edges of $H_i$ consists of all edges that have one endpoint in $S_i$ and another endpoint outside of $S_i$. Clearly, we can compute all such graphs $H_i$ in a single communication round. We now turn to analyze the quality of the resulting shortcuts. For each $1\leq i\leq N$, we denote by $G_i=G[S_i]\cup H_i$.

Since the endpoints of each edge $(u,v)$ appear in the most two subsets $S_i$, it is easy to see that every edge of $G$ may belong to at most two graphs $G_i$. Next, we show that for each $1\leq i\leq N$, the diameter of $G_i$ is bounded by $12n/k$. Indeed, let $u,v$ be any pair of vertices in $G_i$, and let $P$ be the shortest path connecting $u$ to $v$ in $G_i$. Assume for contradiction, that $P$ contains at least $12n/k$ vertices. Since every vertex of $G_i$ either belongs to $S_i$, or has a neighbor in $S_i$, there is a subset $S'\subseteq V(P)$ of at least $6n/k$ vertices that belong to $S_i$. Moreover, there is a subset $S''\subseteq S'$ of at least $2n/k$ vertices, where for all $x,y\in S''$, the distance from $x$ to $y$ in $G_i$ is at least $3$. Since the path $P$ is a shortest $u$-$v$ path in $G_i$, and since the neighbors of all vertices in $S_i$ lie in $G_i$, for every pair $x,y\in S''$ of vertices, the set $\Gamma_G(x)$ of neighbors of $x$ in $G$, and the set $\Gamma_G(y)$ of neighbors of $y$ in $G$ must be disjoint. Since $G$ is $k$-edge connected, for each $x\in S''$, $|\Gamma_G(x)|\geq k$. Therefore:
$$|\bigcup_{x \in S''} \Gamma_{G}(x)|=\sum_{x \in S''} |\Gamma_{G}(x)|\geq |S''|\cdot k > n~,$$
a contradiction. We conclude that the diameter of each graph $G_i$ is at most $O(n/k)$. 

\paragraph{Improved Shortcuts for Smaller Connectivity.}
We now complete the proof of Theorem \ref{lem:dist-shortcuts-construction}.
Assume first that $k\geq n^{(D+1)/(2D+1)}$. In this case, $n/k\leq n^{D/(2D+1)}$, and from the above discussion, the claimed bounds on $\alpha+\beta$, and on the number of rounds hold. Therefore, we assume from now on that $k \leq n^{(D+1)/(2D+1)}$.
The algorithm distinguishes between two cases, depending on the value of the edge connectivity $k$. 
Define 
$$
T_{small}=
\begin{cases}
\sqrt{n/k}, \mbox{~~~if~} k\leq n^{1/(2D+1)}/(101\ln n)\\
n^{D/(2D+1)}, \mbox{~~~if~}  k\in (n^{1/(2D+1)}/(101 \ln n),n^{(D+1)/(2D+1)}].
\end{cases}
$$

At a high level, for every set $S_i$ of cardinality at most $T_{small}$, the algorithm defines $H_i=\emptyset$. 
For the remaining large sets $S_i$, of cardinalities at least $T_{small}$, the algorithm uses Claim \ref{cl:basic-dist-tool} with congestion bound $\eta=\max\{1,101 \ln n\cdot k/n^{1/(2D+1})\}$ in order to construct $k$ subgraphs $G_1,\ldots, G_k$ of $G$, of  diameter at most $(101k\ln n/\eta)^D$ each, that cause total congestion at most $O(\eta\log n)$, in $O(D+\eta\log n)$ rounds. 

The remaining $O(n/T_{small})$ subsets  $S_i$ of cardinality at least $T_{small}$ are handled as follows: each set $S_i$ chooses an index $j\in \set{1,\ldots,k}$ uniformly at random, and then sets $H_i=G_j$. We next describe the implementation details and then analyze the bounds obtained.

At the beginning of the algorithm, every vertex $v$ learns the identities of the sets $S_j$ of each of its neighbors. It then checks whether it is a vertex with largest ID in its subset $S_i$. If so, then $v$ initiates a construction of a BFS tree in $G[S_i]$, that continues up to depth $T_{small}$. Once the BFS tree reaches depth $T_{small}$, every vertex $x$ of $S_i$ that was reached in this last step checks  whether each of its neighbors that lies in $S_i$ has been explored by the BFS. If so, then $S_i$ is a small set; otherwise it is a large set. This information can be propagated back to all vertices that were explored by the BFS\footnote{If a vertex did not receive a message that it belongs to a small set within $O(T_{small})$ rounds, it knows that it is in a large set.}. If $S_i$ is a small set, then, since we set $H_i=\emptyset$, nothing else needs to be done.

We also run the algorithm from  Claim \ref{cl:basic-dist-tool} to compute a collection of spanning subgraphs $G_1,\ldots,G_k$,  of  diameter at most $(101k\ln n/\eta)^D$ each, that cause total congestion at most $O(\eta\log n)$, in $O(D+\eta\log n)$ rounds. 
Next, each large subset $S_i$ needs to select a subgraph $H_i \in \{G_1,\ldots, G_k\}$. 
We need to ensure that all vertices in $S_i$ make the same random decision for the selection of a graph $G_j$, and we will pick the graph $G_j$ in an almost uniform manner. In order to do so, we use  bounded-independence hash functions, see Definition \ref{def: d-wise independent}.
All nodes in the graph will share a short random seed of $\widetilde{O}(1)$ bits, that encodes a $\log n$-wise independent hash function $h:\{0,1\}^{c\log n} \to \{0,1\}^{\log k}$. 
Specifically, by Lemma \ref{lem: d-wise independent}, there is a family $\mathcal{H}$ of $O(\log n)$-wise independent hash functions $h:\{0,1\}^{c\log n} \to \{0,1\}^{\log k}\}$ such that choosing a random function from  $\mathcal{H}$ can be done with a seed length of size $O(\log^2 n)$. This $O(\log^2 n)$-length random seed is shared by all nodes within $O(D+\log^2 n)$ rounds. Let $h \in \mathcal{H}$ be the random function chosen by the shared random seed. 

Then, for each $1\leq i\leq N$, we let $H_i=G_j$, where $j=h(ID(S_i))$. By exchanging messages with its neighbors, each node in $S_i$ can learn all its neighbors in $G[S_i]\cup H_i$. 
This completes the description of the algorithm. We next analyze the correctness and the round complexity.

We first consider the case where $k \leq n^{1/(2D+1)}/(101\ln n)$, and thus $T_{small}=\sqrt{n/k}$ and the congestion bound is $\eta=1$. For every set $S_i$ of size at most $T_{small}$, the algorithm sets $H_i=G[S_i]$.
The remaining sets $S_i$, whose number is bounded by $n/T_{small}=\sqrt{k n}$, are randomly split among the $k$ subgraphs $G_1,\ldots, G_k$ of diameter $(101k\ln n)^D$ each, that cause edge-congestion $O(\log n)$. Using the Chernoff bound for bounded independence (see Theorem \ref{thm:d-wise chernoff}), we get that w.h.p. the edge congestion of the shortcut is bounded by $\alpha\leq O(\log^2 n\cdot \sqrt{n/k})$. In addition, the diameter of the shortcuts is at most $\beta\leq \sqrt{n/k}+(101k\ln n)^D=O(\sqrt{n/k}+n^{D/(2D+1)})$.

Next consider the case where $k \in [n^{1/(2D+1)}/(101\ln n), n^{(D+1)/(2D+1)}]$, and thus $T_{small}=n^{D/(2D+1)}$ and the congestion bound $\eta=101 \ln n \cdot k/n^{1/(2D+1)}$. Note for $D\leq c\cdot \ln n/\log\log n$ for some large constant $c\geq 1$, it holds that $n^{1/(2D+1)}\geq 101\ln^2 n$ and therefore $\eta \leq k/\ln n$.  

Finally, we describe a construction of BFS trees $T'_i$ in each $G[S_i]\cup H_i$ in parallel. This will make sure that the low-depth trees $T'_i$ that span $S_i$ are marked in the sense that each vertex knows its incident edges in each $T'_i$. It is sufficient to consider the case of large sets as for the small sets these trees were already computed.  Fix one such set $S_i$ and let $v$ be the vertex of large ID. We will build a BFS rooted at $v$ in $G[S_i]\cup H_i$ layer by layer, where in every step $j \geq 1$, we assume that we already have computed the first $j$ layers of the tree and that all vertices in layer $j$ know their neighbors in $G[S_i]\cup H_i$. For the base case of $j=1$, the root $v$ knows its edges in all the $G_\ell$ subgraphs and using the seed it can compute the index $\ell=h(ID(S_i))$. In the $j^{th}\geq 1$ step, the nodes of layer $j-1$ send a BFS message that contains the ID of $S_i$ (i.e., the ID of the root $v$) to all their neighbors in $G[S_i]\cup H_i$.
Each vertex $u$ that receives such BFS messages from the tree of $S_i$ can compute which of its incident edges are in $G[S_i]\cup H_i$. This is because in the output format of Claim \ref{cl:basic-dist-tool}  each vertex knows its edges in the subgraphs $G_1,\ldots, G_k$ and using the shared seed and the ID of $S_i$ it can compute the subgraph $G_\ell$ such that $\ell=h(ID(S_i))$. 

We now analyze the algorithm and first consider the quality of the shortcuts.
The algorithm computes $k$ subgraphs $G_1,\ldots, G_k$ of diameter $(101k\ln n/\eta)^D$ and congestion $O(\eta\log n)$. Using the Chernoff bound from Theorem \ref{thm:d-wise chernoff}, the edge-congestion of the shortcuts is bounded by 
$\alpha=O(\eta \cdot \log^2 n \cdot n/(T_{small}\cdot k))=\widetilde{O}(n^{D/(2D+1)})$. In addition, the diameter is bounded by $\beta\leq T_{small}+(101k\ln n/\eta)^D=\widetilde{O}(n^{D/(2D+1)})$, as desired. Finally we bound the number of rounds. Handling the small sets take $O(\sqrt{n/k})$ rounds as we build vertex-disjoint BFS trees up to depth $\sqrt{n/k}$. To handle the large sets, the algorithm applies Claim \ref{cl:basic-dist-tool} with congestion bound $\eta$, since it does not require the claim to output trees but rather subgraphs, this takes $O(D+\eta\log n)$ rounds w.h.p., where the high probability is on the quality of the output subgraphs and not on the running time. Sharing the random seed of length $O(\log^2 n)$ takes $O(D+\log^2 n)$. Finally, computing the BFS trees in each $G[S_i]\cup H_i$ takes $\widetilde{O}(\alpha+\beta)$ rounds w.h.p. This completes the proof of Theorem \ref{lem:dist-shortcuts-construction}.

\subsubsection{Applications of the Improved Shortcuts}
Using the improved shortcuts, we obtain a number of immediate improvements for various network optimization tasks. 
\begin{fact}[\cite{GhaffariThesis17}]
Let $\mathcal{G}$ be a graph family such that for each graph $G \in \mathcal{G}$ and
any partition of $G$ into vertex-disjoint connected subsets $S_1,\ldots,S_N$, one can find an 
$\alpha$ congestion $\beta$-dilation shortcuts such that $\max\{\alpha,\beta\}\leq K$ and this shortcuts can be computed in $\widetilde{O}(K)$ rounds.
Then: 
\begin{itemize}
\item{[Theorem 6.1.2]:} there is a randomized distributed MST algorithm that computes an MST in $\widetilde{O}(K)$ rounds, with high probability, in any graph from the family $G$. 
\item{[Theorem 7.6.1]:}  there is a randomized distributed algorithm that computes a $(1+\epsilon)$ approximation of the minimum cut in $\widetilde{O}(K)$ rounds, with high probability, in any graph from the family $G$. 
\end{itemize}
\end{fact}
Combining with our improved shortcuts for $k$-edge connected graphs, we get:
\begin{corollary}[Improved Distributed MST and $(1+\epsilon)$ Approx. Minimum Cut]
There is a randomized distributed algorithm, that, given a $k$-edge connected $n$-vertex graph of diameter $D$, computes an MST and $(1+\epsilon)$ approximation of the minimum cut
in $\widetilde{O}(\min\{\sqrt{n/k}+n^{D/(2D+1)},n/k\})$ rounds.
\end{corollary}
For $k$-connected graphs with constant diameter $D\geq 5$, this improves upon the state of the art of $O(\sqrt{n})$-rounds for the MST problem. For the $(1+\epsilon)$ approximate minimum cut problem, independently to our work, it is been briefly mentioned in \cite{ghaffari2019faster} (see footnote 4), that it is plausible to get an $\tilde O(\sqrt{n/k}+D)$-round solution. The formal proof of this fact is not yet provided in \cite{ghaffari2019faster}.
%
%
Furthermore, we stress that highly-connected graphs of small diameter might still have very poor expansion, and therefore they are not captured by the improved algorithms for fast mixing graphs\cite{GhaffariKS17,GhaffariL18}.

An additional immediate corollary of improved shortcuts is for computing an approximate SSSP. 
Haeupler and Li \cite{HaeuplerL18} provided improved algorithms for several shortest-path problems whose bounds depend on the quality of shortcuts. By plugging the bounds of Theorem \ref{lem:dist-shortcuts-construction} into Corollaries 2,3 in \cite{HaeuplerL18} we get:

\begin{corollary}[Improved Distributed SSSP Tree Algorithms]\label{cor:imp-sssp}
There are randomized algorithms, that, given a $k$-edge connected $n$-vertex weighted graph with polynomial edge weights of diameter $D$ perform the following tasks:
(1) compute a spanning tree that approximates distances to a given source vertex to within factor $(\log n)^{O(1/\epsilon)})$, in $\widetilde{O}(\min\{\sqrt{n/k}+n^{D/(2D+1)},n/k\})\cdot n^{\epsilon})$ rounds for any constant $\epsilon$; and  (2) compute  a spanning tree that approximates distances to a given source vertex within factor $2^{O(\sqrt{\log n})}$,
in $\widetilde{O}(\min\{\sqrt{n/k}+n^{D/(2D+1)},n/k\})\cdot 2^{O(\sqrt{\log n})})$ rounds.
\end{corollary}

Finally, Dory and Ghaffari \cite{approxMinWeight} recently studied the distributed approximation of minimum weight two-edge connected subgraphs ($2$-EECS). By plugging our shortcut bounds into Theorem 1.2 of \cite{approxMinWeight}, we get:
\begin{corollary}[Improved Approximation of $2$-EECS]\label{cor:twoeecs}
There is an algorithm, that, given a $k$-edge connected $n$-vertex weighted graph of (unweighted) diameter $D$, computes an $O(\log n)$-approximation of the weighted 2-ECSS in $\widetilde{O}(\min\{\sqrt{n/k}+n^{D/(2D+1)},n/k\})$ rounds, with high probability.
\end{corollary}

\subsubsection{Lower Bound of Shortcuts}

\begin{theorem}[Lower Bounds for Shortcuts in Highly Connected Graphs]\label{lem:lb-shortcuts}
For every integer $k$ and sufficiently large integer $n$, there exists a $k$-edge connected $n$-vertex graph $G$ with diameter $2D=O(\log_k n)$ and a partition of its vertices into subsets $S_1,\ldots, S_N$, each
inducing a connected subgraph of $G$, such that regardless
of how the shortcut subgraphs $H_i$ are chosen, if each
$G[S_i] \cup H_i$ has diameter at most $D'<\min\{\sqrt{n/k}+(n/D)^{D/(2D+1)}, n/k\}$, then there is
at least one edge that suffers a congestion of at least $D'$.
\end{theorem}
To show the lower bound argument, we first describe a modification of the lower bound graph construction 
from Sec. \ref{sec: lower bound proof sketch}. Roughly speaking, the resulting modified graph can be viewed as a combination of the lower bound graph from Das Sarma et al. \cite{sarma2012distributed} with the construction of Section \ref{sec: lower bound proof sketch}. 

\paragraph{A Useful Modification of the Lower Bound Graph from Sec. \ref{sec: lower bound proof sketch}.} We first describe the construction of a lower bound graph $G^*_{k,\alpha,\eta,D}$. For simplicity we denote $G^*=G^*_{k,\alpha,\eta,D}$. The graph $G^*$ is obtained by first taking the graph $G'_{w,D}$ (see the proof of Theorem~\ref{thm:diameterlowerbound} in Section~\ref{sec: lower bound proof sketch}) for $w=k/(2D\alpha\eta)$ and adding to it a collection of $q=\lfloor n/(2w^D)\rfloor$ paths $\mathcal{P}=\{P_1,\ldots, P_q\}$. 
 
Recall that $G'_{w,D}$ is the graph obtained from $G_{w,D}$
by first replacing each vertex $v_i$ with a set $X_i=\{x_i^1,x_i^2\ldots,x_i^k\}$ of $k$ vertices that form a clique, and then replacing, for each $1\le i<N$, the $k$ red edges connecting $v_i$ to $v_{i+1}$ by the perfect matching $\{(x^t_i,x^t_{i+1})\}_{1\le t\le k}$ between vertices of $X_i$ and vertices of $X_{i+1}$, and finally, replacing each blue edge $(v_i,v_j)$ by a edge $(x^1_i,x^1_j)$. We denote by $E'$ the set of edges that replace the blue edges of $G_{w,D}$.

Recall that $L$ is the set of leaf nodes of $G_{w,D}$. 
We define $L^*\subseteq V(G^*)=\{x_i^1\mid v_i\in L\}$.
So $|L^*|=w^D$.
Let $s=x^1_1$, and $t=x^1_{N-D}$ be vertices of $L^*$ that have lowest and largest index, respectively. Then, the vertices of paths of $\mathcal{P}$ are connected to the nodes of $L^*$ as follows: for each $1\le i\le |L^*|$ and for each $j \in \{1,\ldots, q\}$, there is an edge connecting each node of $X_{r_i}$ (where $r_i$ is the $i$th smallest index of vertices of $L$) to the $i^{th}$ node of each path $P_j$. 
See Figure~\ref{fig:shorcut-lb} for an illustration.
\begin{claim}\label{cl:hlp}
Any set of $k/\alpha$ paths connecting $s$ to $t$ in $G^*$ that causes edge-congestion at most $\eta$ must contain at least one path of length at least $|L^*|/2$.
\end{claim}
\begin{proof}
We use similar arguments in the proof of Theorem~\ref{thm:diameterlowerbound}.
By the definition of $G^*$, $V(G^*)=\left(\bigcup_{j=1}^N V(P_j)\right)\cup V(G'_{w,D})$. 
For each $i$ such that the node $v_i\in L$, we define $Q_i$ to be the set of vertices of $\bigcup_{1\le j\le q}V(P_j)$ that is connected to $x^1_i$, and
we let $W_i=\bigcup_{1\le t\le i} (X_i\cup Q_i)$ and $\overline{W_i}=V(G^*)\setminus W_i$.

It can be shown (similar to Corollary~\ref{cor: number of cut edges}) that, for each $i$ such that $v_i\in L$, $|E_{G^*}(W_i,\overline{W_i})\cap E'|\le Dw$. Since the $k/\alpha$ paths between $s$ and $t$ in $G^*$ cause edge-congestion $\eta$, at most $Dw\eta$ of them may contain edges in $E_{G^*}(W_i,\overline{W_i})\cap E'$. So each of the remaining $\frac{k}{\alpha}-Dw\eta\ge\frac{k}{2\alpha}$ paths either contains an edge from $\left(\bigcup_{1\le j\le q}E(P_j)\right)$ that crosses the cut $(W_i,\overline{W_i})$, or contains an edge of the perfect matching $\{(x^t_i,x^t_{i+1})\}_{1\le t\le k}$ between vertices of $X_i$ and vertices of $X_{i+1}$.
Therefore, the sum of lengths of the $k/\alpha$ paths of is at least $|L^*|\cdot \frac{k}{2\alpha}$, and therefore at least one path must have length at least $|L^*|/2$.
\end{proof}
\subsubsection*{Proof of Theorem \ref{lem:lb-shortcuts}}
The proof is divided into three parts.  First, we show lower bound of $\sqrt{n/k}$. Then, we show a lower bound of $(n/D)^{D/(2D+1)}$ for intermediate values of the edge connectivity $k$. Finally, we show a lower bound of $n/k$ provided that $k$ is sufficiently large.

\paragraph{Lower Bound of $\sqrt{n/k}$.}
We first show that there exists a $k$-edge connected $n$-vertex graph with diameter $D=O(\log_k n)$ as well as vertex-disjoint subsets $S_1,\ldots, S_q$, such that in every shortcut for these subsets, either the congestion is $\Omega(\sqrt{n/k})$, or that the diameter is at least $\Omega(\sqrt{n/k})$. 
The lower bound graph is given by $G^*=G^*_{k,\alpha,\eta,D}$ where $w$ and $D$ are chosen such that $w=k/(2D\alpha\cdot \eta)$ for $\alpha=2$ and $\eta=\log n$ and $w^D=\sqrt{n/k}$, thus $D=O(\log_k n)$. Recall that $L$ is the set of leaf nodes in $G_{w,D}$, and $|L|=w^D$.
Let $S_i=V(P_i)$ for every $i \in \{1,\ldots, q\}$ (recall that $q=\lfloor n/(2w^D)\rfloor=\Theta(\sqrt{k n})$). Assume that there are shortcuts for the sets $\{S_i\}_{1\le i\le q}$ with diameter $D'$ such that the total congestion is at most $K<\sqrt{n/k}$. We will show that these shortcuts give a collection of $q$ paths $\mathcal{P}=\{P'_1,\ldots, P'_q\}$ from $s$ to $t$ of length $D'+2$ each and causes edge-congestion at most $K$. To see this, for each $1\le i\le q$, let $u_i,v_i$ be the endpoints of $P_i$, we define $P'_i$ to be the concatenation of the $u_i-v_i$ shortest path in $G[S_i]\cup H_i$ and the edges $(s,u_i)$ and $(t,v_i)$. That is, $P'_i=(s, u_i) \circ \pi(u_i,v_i,G[S_i]\cup H_i) \circ (v_i,t)$ for every $i \in \{1,\ldots, q\}$. 

We next claim that there exists a subset $\mathcal{P}' \subseteq \mathcal{P}$ of at least $k/2$ paths of total congestion at most $\eta$. From Claim~\ref{cl:hlp}, this implies that one of these paths have length $\Omega(\sqrt{n/k})$ and thus the diameter of at least one of the subgraphs $G[S_i]\cup H_i$ is at least $\Omega(\sqrt{n/k})$ as well.
To see this, we sample each $P'_i$ with probability $k/q$. Thus, in expectation, $k$ paths are sampled with congestion $K \cdot k/2q\leq \sqrt{n/k}\cdot \Theta(k/(\sqrt{k n}))=O(1)\le \eta$. Therefore such a subset exists and the claim follows. 

\begin{figure}[h!]
\begin{center}
	\includegraphics[scale=0.5]{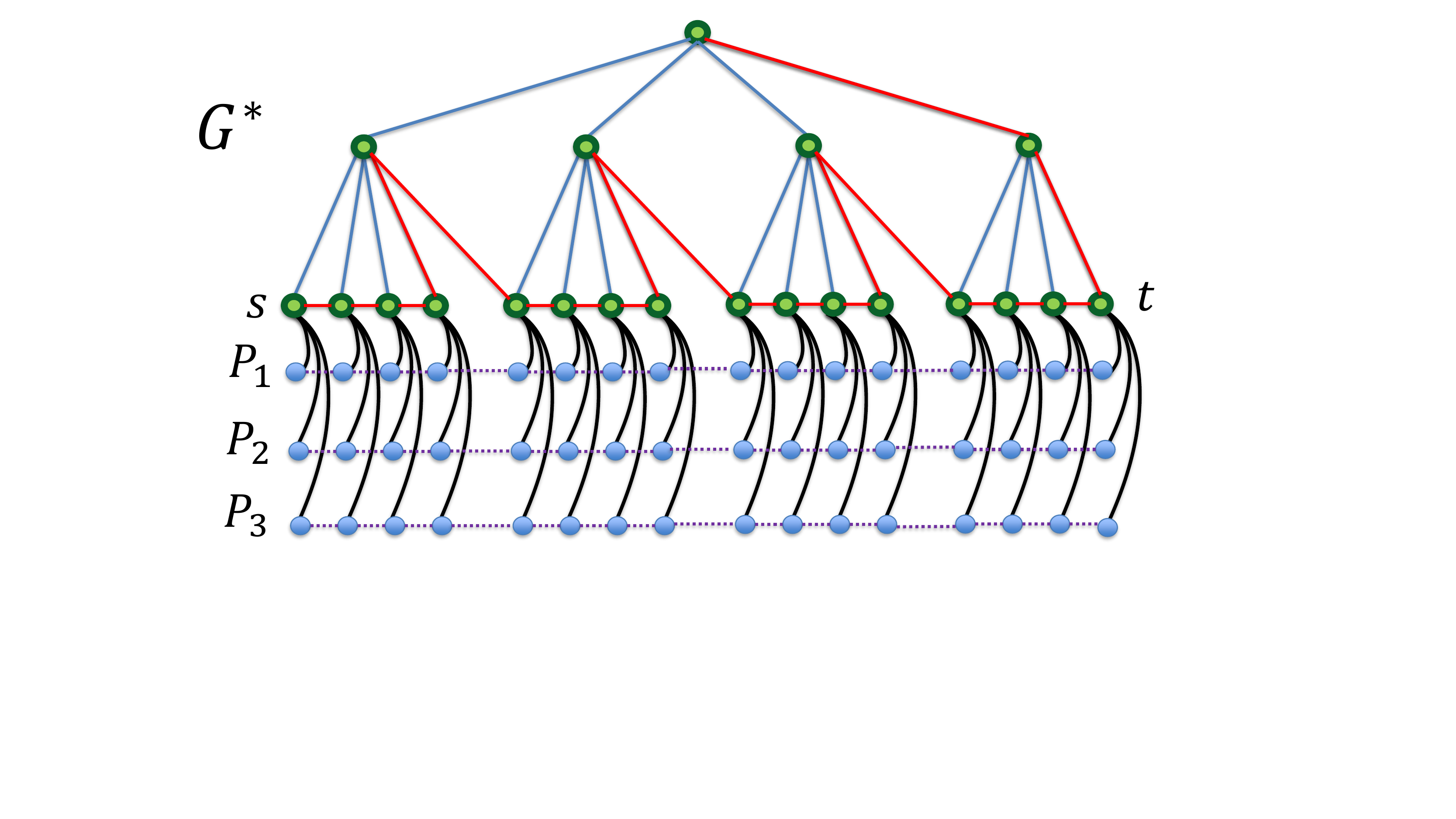}
	\caption{A lower bound graph for low-congestion shortcuts in highly connected graphs.}
	\label{fig:shorcut-lb}
	\end{center}
\end{figure}

\paragraph{Lower Bound of $\Omega{(n/D)^{D/(2D+1)}}$ for $k \in [4D \cdot n^{1/(2D+1)}, (n/D)^{(D+1)/(2D+1)}/10]$.}
Given an edge connectivity value $k$, a number of nodes $n$ and a diameter value $D$, we build the lower bound graph $G^*=G_{k,\alpha,\eta,D}$ where $\alpha=2$ and $\eta=k/(4D \cdot (n/D)^{1/(2D+1)})$. The set of marked leaf nodes $L^*$ contains $w^D$ vertices where $w=k/(2D\alpha\cdot \eta)$, thus $|L^*|=(n/D)^{D/(2D+1)}$. Note that $|L^*|\cdot 2k < n/2$ and therefore by setting the constants carefully, the graph $G^*$ contains $n$ vertices in total. Recall that $G^*$ contains a collection of $q=n/(2L^*)$ paths $P_1,\ldots, P_q$ of length $|L^*|$. 
Let $S_i=V(P_i)$ for every $i \in \{1,\ldots, q\}$. Assume that there are shortcuts for the $S_i$ sets such that the total congestion is at most $K=(n/D)^{D/(2D+1)}$ and with diameter $D'$. As in the previous paragraph, these shortcuts implies a collection of $q$ paths $\mathcal{P}=\{P'_1,\ldots, P'_q\}$ from $s$ to $t$ of total congestion $K$ and length at most $D'+2$. 

We next claim that there exists a subset $\mathcal{P}' \subseteq \mathcal{P}$ of at least $k/2$ paths of total congestion $\eta$. From Claim~\ref{cl:hlp}, this implies that one of these paths have length $\Omega(|L^*|)$ and thus the diameter of at least one of the subgraphs $G[S_i]\cup H_i$ is at least $\Omega(|L^*|)=\Omega(n^{D/(2D+1)})$.
To see this, we sample each $P'_i$ with probability $k/q$. Thus, in expectation, there are $k$ paths and the congestion is $K \cdot k/q \leq (n/D)^{D/(2D+1)} \cdot k/q=\eta/\log n$. Therefore such a subset exists and the claim follows. 

\paragraph{Lower Bound of $n/k$ for $k=\Omega((n/D)^{(D+1)/(2D+1)}\cdot \log n)$.} Set $\eta=k^{1+1/D}/(12D\cdot n^{1/D})$ and $\alpha=2$ and let $G^*=G_{k,\alpha,\eta,D}$. Recall that for $w=k/(2D\alpha\eta)$, the number of marked leaf nodes is $|L^*|=w^D=n/(10k)$. Therefore there are at least $q \in [k,7k]$ paths $P_1,\ldots, P_q$ of length $|L^*|$ in $G^*$. 

Assume towards contradiction that there are shortcuts of congestion at most $K\leq n/(10k)$ of length at most $D'\leq (n/10k)$. Hence there are $q$ paths $\mathcal{P}=\{P'_1,\ldots, P'_q\}$ from $s$ to $t$ of length at most $D'+2$ and congestion at most $K$. We next show that this implies that there is a subset of $k/4$ paths $\mathcal{P}' \subseteq \mathcal{P}$ of with congestion at most $\eta$. From Claim~\ref{cl:hlp}, at least one of these paths have length at least $|L^*|/2=\Omega(n/k)$.

Sample each path $P_i \in \mathcal{P}$ with probability $k/q$ into $\mathcal{P}'$. In expectation the number of sampled paths is $k$ and the congestion is $k/q \cdot K\leq n/(2k)\leq \eta/\log n$ by plugging the bound on the value of $k$. We therefore get that w.h.p. there exists a collection of $k/2$ paths in $\mathcal{P}$ with congestion at most $\eta$. From Claim~\ref{cl:hlp}, at least one of these paths have length at least $|L^*|/2=\Omega(n/k)$.
This completes the proof for Lemma \ref{lem:lb-shortcuts}.

%% file: application-information.tex
\subsection{Improved Bounds for the Information Dissemination Task}
In the information dissemination task, we are given a $k$-edge connected $n$-vertex graph $G$ of diameter $D$, with two special nodes $s$ and $t$. The source $s$ receives as input a sequence of $N$ bits, which it needs to send to $t$ as fast as possible. 

Ghaffari and Kuhn \cite{ghaffari2013distributed} showed a lower bound of $\Omega(\min\{N/\log^2 n,n/k\})$ rounds for the special case where the diameter $D$ of the graph is logarithmic in $n$, that is, $D=\Theta(\log n)$.

We consider the setting where the diameter $D$ is sub-logarithmic in $n$, and provide the first upper and lower bounds for this setting. We start with the following upper bound.

\begin{lemma}[Upper Bound for Information Dissemination]
There is a randomized algorithm, that, given any $k$-edge connected $n$-vertex graph $G$ of diameter $D$ with a source vertex $s$ and a destination vertex $t$, sends an input sequence of $N$ bits from $s$ to $t$. The number of rounds is bounded by $\widetilde{O}(N^{1-1/(D+1)}+N/k)$ with high probability.
\end{lemma}

\begin{proof}
	We use a parameter $\chi=N^{1/(D+1)}/(101\ln n)$, and we set $\eta=\max\{1, k/\chi\}$. By applying the algorithm from Claim \ref{cl:basic-dist-tool} to $G$, with the congestion bound $\eta$, we obtain a collection $\tset=\set{T_1, \ldots, T_k}$ of $k$ spanning trees, such that, with high probability, all trees have diameter at most $(101\ln n \cdot k/\eta)^D$, and the total edge congestion due to $\tset$ is at most $O(\eta \cdot\log n)$.  We partition the input sequence $X$ of bits into $k$ consecutive sub-sequences $X_1,\ldots,X_k$, each of which contains at most $\ceil{N/k}$ bits.
	For each $1\leq i\leq k$, consider the algorithm $A_i$, that sends the bits of $X_i$ from $s$ to $t$ along the tree $T_i$. Notice that, assuming that the algorithm for constructing the trees was successful, algorithm $A_i$ can be implemented in $O(N/k+ (101k\ln n/\eta)^D)$ rounds, and every edge is used to send $O(N/k)$ messages. Since the trees in $\tset$ cause edge-congestion $O(\eta\log n)$, the total number of messages that are sent via a single edge by all the algorithms $A_1,\ldots,A_k$ altogether is at most $O(\eta N\log n/k)$.
	
	We can now use the the random delay approach, to send all bits of $X$ from $s$ to $t$, by combining the algorithms $A_1,\ldots,A_k$. The number of rounds is bounded by   $r=\widetilde{O}((101k\ln n /\eta)^D+\eta\cdot N/k)$ with high probability.
	
	We now consider two cases. First, if $k\leq \chi$, then $\eta=1$, and $r=\widetilde{O}((101k\ln n)^D+ N/k)\leq \widetilde{O}(N^{D/(D+1)}+N/k)$. On the other hand, $N/k\geq N/\chi\geq N\cdot(101\ln n)/N^{1/(D+1)}\geq N^{D/(D+1)}\cdot (101\ln n)$. Therefore, $r\leq \tilde O(N/k)$ holds.
	
	Consider now the second case, where $k>\chi$. In this case, $\eta=k/\chi=101k\ln n/N^{1/(D+1)}$. We then get that:
	\[r=\widetilde{O}((101k\ln n /\eta)^D+\eta\cdot N/k)=\widetilde{O}(N^{D/(D+1)}+ N/\chi)=\widetilde{O}(N^{1-1/(D+1)}),
	\]
	as required.
\end{proof}


We now provide the first lower bound on the round complexity of information dissemination for graphs with sublogarithmic diameter. Unlike \cite{ghaffari2013distributed} our lower bound only holds for the weaker setting of store and forward algorithms (in which modification of messages by e.g., network coding is not allowed).

\begin{theorem}[Lower Bound for Information Dissemination]
For all integers $n,N,D$ and $k\leq n$, there 
 exists a $k$-edge connected $n$-vertex graph $G=(V,E)$ of diameter $2D$, and a pair $s,t$ of its vertices, such that sending $N$ bits from $s$ to $t$ in a store-and-forward manner requires at least
$$\Omega(\min\{(N/(D \log n))^{1-1/(D+1)},n/k\}+N/k+D)~ \mbox{~rounds.}$$
\end{theorem}

\begin{proof}
We start with the following simple observation.
\begin{observation}\label{obs:N/k}
Let $G$ be a graph of diameter $D$, such that some vertex $s$ in $G$ has degree at most $k$, and let $t$ be any vertex at distance at least $D/2$ from $s$. Then any algorithm for sending $\Omega(N)$ bits from $s$ to $t$ must send at least $N/k$ messages on at least one of the graph edges, therefore requires $N/k$ rounds. 
\end{observation}

\begin{proof}
Let $A$ be any algorithm for the problem, and consider its execution on a given sequence $X$ of $N$ bits. Assume for contradiction that the number of messages sent on each edge is strictly less than $N/k$. Then every edge that is incident to $s$ may be used to carry at most $N/k-1$ bits, and, since the degree of $s$ is $k$, fewer than $N$ bits are sent from $s$, a contradiction.
\end{proof}


We next show a lower bound of $\Omega(\min\{N/(D \log n)^{1-1/(D+1)},n/k\})$ rounds. 
To do that, we show that for every $\eta \in [1,k]$, there is a $k$-edge connected $D$-diameter graph $G_{\eta}$ and $s,t \in V(G_{\eta})$ such that if the algorithm sends the $N$ bits from $s$ to $t$ with a total edge congestion at most $\eta \cdot N/k$, then one of the $s$-$t$ paths used by the algorithm must have length at least $\min\{1/4 \cdot (k/(2\eta D\log n))^D,n/k\}$. 
The graph $G_{\eta}$ is given by taking the lower bound graph of Section \ref{sec: lower bound proof sketch} (Theorem~\ref{thm:diameterlowerbound}) with $\alpha=\log n$ and the congestion bound is $\eta$. We then choose $s$ and $t$ as the left-most and right-most leaf nodes. 

For every algorithm $A$ for this problem, let $\congestion(A)$ be the maximal number of messages sent through a given edge in $G$ (i.e., there is a graph edge on which the algorithm $A$ passes $\congestion(A)$ messages, and on all other edges at most $\congestion(A)$ messages are sent). By Observation \ref{obs:N/k}, $\congestion(A)\geq N/k$. 
Assume that there exists an algorithm $A$ for which $\congestion(A)=\lfloor \eta \cdot N/k \rfloor$ for some $\eta \in [1,k]$. Let $\mathcal{P}=\{P_1,\ldots, P_N\}$ be the $s$-$t$ paths on which these $N$ bits are sent, and 
let $\dilation(A)=\max_{j} |P_j|$ be the length of the longest path in $\mathcal{P}$. Our goal is to bound $\dilation(A)$ from below.

To do that, we claim, using a probabilistic argument, that $\mathcal{P}$ contains a subset $\mathcal{P}'$ of at least $k/\log n$ paths such that each edge $e$ appears on at most $\eta$ paths in $\mathcal{P}'$. To see this, sample each path $P_i \in \mathcal{P}$ with probability $p=k/(N\cdot \log n)$. The expected number of sampled paths is $pN/k=k/\log n$ and the expected congestion is at most $\eta/\log n$ and w.h.p. at most $\eta$. Therefore such a collection $\mathcal{P}'\subseteq \mathcal{P}$ exists. By Theorem~\ref{thm:diameterlowerbound} it then holds that one of the $s$-$t$ paths must have length $\dilation(A)=\min\{1/4 \cdot (k/(2 \eta \cdot D\log n))^D,n/k\}$.
Therefore the total running time of the algorithm $A$ is at least 
$$\dilation(A)+\congestion(A)\geq \min\{1/4 \cdot (k/(2 \eta \cdot D\log n))^D,n/k\}+\eta\cdot N/k-1.$$
 
First, observe that for every $\eta$ satisfying that $n/k=O((k/(2 \eta \cdot D\log n))^D)$, we get that $\dilation(A)=\Omega(n/k)$, as desired. So from now on assume that $n/k \geq (k/(2 \eta \cdot D\log n))^D$ and therefore 
\begin{equation}\label{eq:eta-bound}
\eta \geq k^{1+1/D}/(n^{1/D}\cdot 2D\log n)~.
\end{equation}

We distinguish between two cases depending on the value of $k$. 

\textbf{Case 1, $k \geq \frac{n}{(N/2D\log n)^{D/(D+1)}}$:} The round complexity is at least 
$$\eta N/k \geq k^{1+1/D}/(n^{1/D}\cdot 2D\log n)\cdot N/k= (k/n)^{1/D}\cdot N/(2D\log n)\geq (N/(2D\log n))^{1-1/(D+1)},$$
as desired.

\textbf{Case 2, $k \leq \frac{n}{(N/2D\log n)^{D/(D+1)}}$:} The round complexity is $\Omega((k/(2 \eta \cdot D\log n))^D+\eta\cdot N/k)$. Our goal is to find $\eta^*$ that minimizes this expression subject to the constraint of satisfying Ineq. (\ref{eq:eta-bound}). The minimum value is obtained for setting 
$$\eta^*=k/((2D\log n)^{1-1/(D+1)} \cdot N^{1/(D+1)})~.$$
One can verify that Ineq. (\ref{eq:eta-bound}) holds for $k \leq \frac{n}{(N/2D\log n)^{D/(D+1)}}$. 
Therefore the running time of $A$ is at least $\eta^* \cdot N/k=\Omega(N/(D \log n)^{1-1/(D+1)})$ as desired.

\end{proof}

%% file: application-secure.tex
\subsection{Implications to Secure Distributed Computation}
A cycle cover of a bridgeless graph $G$ is a collection $\cset$ of simple cycles in $G$, such that each edge $e\in E(G)$ belongs on at least one cycle in $\cset$. Motivated by applications to distributed computation, Parter and Yogev \cite{parter2019low} introduced the notion of \emph{low-congestion} cycle covers, in which all cycles in  $\cset$ are required to be both \emph{short} and nearly \emph{edge-disjoint}. Formally, a $(\dilation,\congestion)$-cycle cover of a graph $G$ is a collection $\cset$ of cycles in $G$, such that each cycle $C\in \cset$ has length at most $\dilation$, and each edge $e\in E(G)$ participates in at least one cycle and at most $\congestion$ cycles of $\cset$. Parter and Yogev \cite{parter2019low} showed that, using a $(\dilation, \congestion)$-cycle cover of a graph $G$, one can compile any $r$-round distributed algorithm $\mathcal{A}$ into a
\emph{resilient} one, while only incurring a multiplicative overhead of $\widetilde{O}(\dilation+\congestion)$ in the round complexity.  
Two types of adversaries were considered in \cite{parter2019low}: (i) a \emph{Byzantine} adversary, who can corrupt a \emph{single} message in each round; and (ii) an 
\emph{eavesdropper} adversary, who  can listen to \emph{one} of the graph edges of its choice  in each round.
The common to both of these types of adversaries is that they are restricted to manipulating only a \emph{single} edge of the graph in a given round. This restriction follows from the fact that the cycle
cover provides each edge $e$ with only two edge-disjoint paths connecting its endpoints: a direct one using the edge $e$, and an indirect
one using the cycle $C_e$ that covers $e$. It is noteworthy that this is the best that one can hope for if the graph is two-edge connected. Handling stronger adversaries, who collude on $k$ edges in a single round, requires that the communication graph is at least $(k+1)$-edge connected. We then need a generalization of low-congestion cycle cover that leverages this high connectivity, by covering each edge with many (nearly) edge-disjoint cycles, rather than a single one. 

To illustrate our ideas in the cleanest way, we consider an eavesdropper adversary in $k$-edge connected graphs. The adversary is allowed to eavesdrop to a fixed set of at most $k'$ edges (unknown to the graph participants) in each round during the simulation. Ideally, we would want to make $k'$ as large as possible. Our goal is to compile any given distributed algorithm $\cA$ into a $k'$-secure algorithm $\cA'$, that has the same output as $\cA$, but provides resilience against such adversaries. In other words, the resilient algorithm $\cA'$ must  guarantee that the adversary learns nothing by eavesdropping to any fixed collection of $k'$ edges of the graph.  Towards that goal, we cover each edge $e=(u,v)$ not by a single cycle (as in the standard cycle cover), but rather by a collection of $(k'\cdot\eta+1)$ cycles with overlap $\eta$. In other words, there is a collection $\cset_e\subseteq\cset$ of $(k'\cdot\eta+1)$ cycles containing $e$, such that each edge $e'\in E(G)\setminus\set{e}$ appears on at most $\eta$ cycles of $\cset_e$. This provides $u$ and $v$ with a communication backbone that is $k'$-connected and thus resilient to any adversary who takes over $k'$ edges of the graph. The efficiency of this scheme depends on several parameters that are captured by the following generalization of cycle covers. 
%
\begin{definition}[Cycle Covers of Higher Connectivity]
A $(\dilation,\congestion,\congedge,k)$-cycle cover for a graph $G=(V,E)$ is a collection $\cset$ of cycles satisfying:
\begin{itemize}
\item{\textbf{Congestion:}} Each edge $e\in E(G)$ appears on at most $\congestion$ cycles in $\mathcal{C}$.
\item{$(k,\eta)$ \textbf{Covering:}} For each edge $e\in E(G)$, there is a collection $\mathcal{C}_e=\{C_{e,1},\ldots, C_{e,k}\}$ of $k$ cycles  that contain $e$, with overlap $\eta$; that is, each edge $e'\in E(G)\setminus\set{e}$ appears on at most $\eta$ cycles of $\mathcal{C}_e$. 
\item{\textbf{Length:}} Each cycle in $\mathcal{C}$ has length at most $\dilation$. 
\end{itemize}
\end{definition}

Our key contribution is in providing an algorithm that computes a $(\dilation,\congestion,\eta, k)$-cycle cover, given a tree packing of size $k$, congestion $\eta$ and diameter $O(d/\log n)$, where $\congestion=O(k \cdot \eta \cdot \log^3 n)$.

\begin{lemma}\label{lem:cc-high-conn-teepath}[From Tree Packing to High-Connectivity Cycle Cover]
There is an efficient randomized algorithm, that, given a tree packing $\tset$ of $k$ spanning trees with congestion $\eta-1$ and diameter $D$, where $\eta \in [2,k+1]$, computes a $(\dilation,\congestion,\congedge,k')$-cycle cover with $\congestion=O(k \cdot \eta \cdot \log^3 n), \dilation=O(D\log n)$ and $k'= k-\eta+1$. The bounds on the congestion $\congestion$ and the $(k,\eta)$ covering property hold w.h.p., while the cycle length bound holds with probability $1$.
\end{lemma}

By combining Theorem \ref{thm:random_tree_planting} and Theorem \ref{thm:packing_spanning_trees_in_(k,D)_with_congestion} respectively with Lemma \ref{lem:cc-high-conn-teepath}, we obtain the following immediate corollary:
\begin{corollary}\label{lem:cc-high-conn}
(1) There is an efficient randomized algorithm that, given a $k$-edge connected $D$-diameter graph $G$,  w.h.p.  computes a $(\dilation,\congestion,\congedge,k/2)$ cycle cover with $\congedge=O(1)$, $\dilation=O((101k \ln n)^D\cdot \log n)$ and $\congestion=O(k\log^3 n)$ (all these properties hold w.h.p). \\ 
(2) There is an efficient randomized algorithm, that given a $(k,D)$-connected graph $G$,  w.h.p. computes a $(\dilation,\congestion,\eta,k)$ cycle cover with $\congestion=O(k\log^4 n), \congedge=O(\log n)$ and $\dilation=O(D\log^2 n)$.  
\end{corollary}

\emph{Proof of Lemma \ref{lem:cc-high-conn-teepath}}
We start with a short overview of the cycle cover algorithm from \cite{parter2019low}.
\paragraph{Overview of Algorithm $\GraphCover$ of \cite{parter2019low}.}
The algorithm starts by constructing a BFS tree $T$ of the graph $G$, and then proceeds in two stages. In the first stage, it uses Procedure $\NonTreeCover$ to construct a cycle cover $\cC_1$ for all non-tree edges -- the edges of $E(G)\setminus E(T)$.  In the second stage, Procedure $\TreeCover$ is employed in order to cover the remaining tree edges by a new collection $\cC_2$ of cycles. 

The initial collection $\cset_1$ of cycles has some 
useful properties that will be exploited in our algorithm.

\begin{fact}\label{fc:nontreecover}[Properties of Algorithm \ $\NonTreeCover$]
Let $T$ be the tree used in the algorithm.  Then: (i) each non-tree edge $e\in E(G)\setminus E(T)$ belongs to a 
single cycle in $\cC_1$; (ii) the length of each cycle in $\cset_1$ is $O(\depth(T)\cdot\log n)$; 
(iii) each cycle $C\in \cC_1$ contains at most $2\log n$ edges of $E(G)\setminus E(T)$; and (iv) each tree edge $e\in E(T)$ belongs to at most $O(\log n)$ cycles in $\cC_1$.
\end{fact}

\paragraph{From Low-Diameter Tree Packing to Highly Connected Cycle Cover.}
Let $\tset=\set{T_1,\ldots, T_{k}}$ be the given tree packing of congestion $\eta-1$ and diameter $D$. 

Our algorithm will compute, for each tree $T_i\in \tset$, a collection $\cset^i$ of cycles that covers all edges of $E(G)\setminus E(T_i)$ -- that is, the edges that are non-tree edges for $T_i$, using Algorithm  $\NonTreeCover$. Recall that an edge $e\in E(G)$ may belong to at most $\eta-1$ trees, and so $e$ is a non-tree edge for at least $k-\eta+1$ trees. For each such tree $T_i$, at least one cycle in $\cset^i$ will cover it.

Specifically, our algorithm processes each tree $T_i$ one-by-one, and for each such tree it computes a collection $\cC^i$ of cycles, using Algorithm $\NonTreeCover$. The final collection $\cC$ of cycles is obtained by taking the union over all resulting sets of cycles: $\cC=\bigcup_{i=1}^k \cC^i$. 
Consider now some index $1\leq i\leq k$. Every cycle $C \in \cC^i$ contains two types of edges: the non-tree edges, that lie in  $G\setminus T_i$ and the tree edges, that appear in $T_i$. When processing the tree $T_i$, for $i\geq 2$, our goal will be to compute a cycle-cover $\cC^i$ that covers each edge $e \notin T_i$ by a short cycle $C^i_e$, such that if $e'\in E(C^i_e)\setminus\set{e}$ is a non-tree edge (that is, $e'\not\in E(T_i)$,)  then it may not belong, as a non-tree edge, to any other cycle that covered the edge $e$ in previous iterations.

The key observation is that due to property (iii) of Fact \ref{fc:nontreecover}, there are at most $(i-1)\cdot 2\cdot \log n=O(k\log n)$ non-tree edges that need to be avoided when covering the edge $e$ using the tree $T_i$. In order to avoid such edges, we apply an algorithm that is inspired by a sampling procedure that is mostly used in the setting of fault tolerant network design \cite{weimann2013replacement,dinitz2011fault}. 


We now describe in details the $i^{th}$ phase, where the collection $\mathcal{C}^i$ of cycles is computed.

\textbf{Computing the cycle collection $\cC^i$ using the tree $T_i$.} 
We compute the collection $\cset^i$ of cycles in $\ell=\ceil{30k\log^2n}$ iterations, where in each iteration $j\in \set{1,\ldots,\ell}$, we compute a collection $\cset^i_j$ of cycles, as follows. We let $G_{i,j}$ be the subgraph of $G$ obtained by sampling each edge $e \in E(G)\setminus T_i$ with probability $p=1/(2k \cdot \log n)$, and adding all edges of $T_i$ to the resulting graph. In other words, $G_{i,j}=G[p]\cup T_i$. Let $\cC^i_j$ be the cycle collection obtained by applying Algorithm $\NonTreeCover$ to the graph $G_{i,j}$ and the tree $T_i$. After we complete the $\ell$th iteration, we set $\cC^i=\bigcup_{j=1}^\ell \cC^i_j$. This concludes the description of the algorithm. Lastly, we set $\cset=\bigcup_{i=1}^k\cset^i$. We now show that set $\cset$ of cycles has all required properties with high probability.

\textbf{Cycle Length:} It is easy to verify that the length of every cycle in $\cset$ is bounded by $O(D\log n)$, from Property (ii) of Fact \ref{fc:nontreecover}.

We next consider the covering property by showing that each edge $e$ is covered by at least $k-\eta+1$ cycles with overlap (at most) $\eta$.

\textbf{$(k-\eta+1,\eta)$ Covering:} We prove the $(k-\eta+1,\eta)$-covering property in the following claim.

\begin{claim}\label{claim: covering individual edges}
With high probability, for every edge $e\in E(G)$, there is a collection $\cset(e)\subseteq \cset$ of at least $k-\eta+1$ cycles, such that each edge $e'\in E(G)\setminus\set{e}$ appears on at most $\eta$ cycles of $\cset(e)$.
	
\end{claim}

\begin{proof}
%
For each iteration $1\leq i\leq \ell$, we call the edges of $E(G)\setminus E(T_i)$ \emph{non-tree edges for iteration $i$}, and we call the edges of $E(T_i)$ \emph{tree edges for iteration $i$}. Similarly, for a cycle $C\in \cset^i$, the edges of $C$ that lie in $T_i$ are called \emph{tree edges}, and the remaining edges of $C$ are called \emph{non-tree edges}.

Let $\set{T_{i_1},\ldots, T_{i_{k-\eta+1}}}$ be a set of $k-\eta+1$ trees of $\tset$, that do not contain the edge $e$, where $j_1<j_2 \ldots <j_{k-\eta+1}$. Such collection of trees must exist, since the trees in $\tset$ cause edge-congestion at most $\eta-1$.

We show by induction that for all $1\leq z\leq k-\eta+1$, set $\cset^{i_z}$ of cycles with high probability contains a cycle $C^z(e)$, such that $e\in C^z(e)$, and moreover, if $e'\in C^z(e)\setminus\set{e}$ is a non-tree edge, and $e'\in C^{z'}(e)$ for some $z'<z$, then $e'$ must be a tree edge for $C^{z'}(e)$. In other words, the non-tree edges of $C^z(e)$ are disjoint from the non-tree edges of $C^1(e),\ldots,C^{z-1}(e)$.
The proof proceeds by induction.

For $z=1$, we show that with high probability there is some cycle $C^1(e)\in \cset^{i_1}$ that contains $e$. 
Indeed, for all $1\leq j\leq \ell$, the edge $e$ is added to graph $G_{i_1,j}$ with probability $p=1/(2k\log n)$. Since $\ell=\ceil{30k\log^2n}$, with probability at least $1-1/n^{10}$, there exists an iteration $j$ with $e \in  G_{i_1,j}$. By property (i) of Fact \ref{fc:nontreecover}, the edge $e$ is covered by exactly one cycle in $\mathcal{C}^{i_1}_j$, that we denote by $C^1(e)$.



Assume now that we have defined the cycles $C^1(e), \ldots, C^{z-1}(e)$ in cycle sets $\cC^{i_1},\ldots, \cC^{i_{z-1}}$ respectively, such that the non-tree edges of all these cycles are disjoint. We show that w.h.p. $\cC^{i_z}$ contains a cycle $\cC^z(e)$, that covers $e$, such that its non-tree edges are disjoint from the non-tree edges (except for the mutual edge $e$) of all the previous cycles, namely, $C^1(e) \ldots, C^{z-1}(e)$. 
%
%

Let $F_z(e)$ be the collection of all non-trees edges of the cycles  $C^1(e), \ldots, C^{z-1}(e)$, excluding the mutual edge $e$. By property (iii) of Fact \ref{fc:nontreecover}, $|F_z(e)|\leq 2(z-1)\log n\leq 2\cdot k\log n$. An iteration $q$ in phase $i_z$ is said to be \emph{successful} for the edge $e$, if the following two events hold: (a) $e \in G_{i_z,q}$; and (b) $F_z(e)\cap G_{i_z,q}=\emptyset$. We next show that with probability at least $1-1/n^8$, at least one iteration of phase $i_z$ is successful for the edge $e$.

Indeed, 
the probability that a specific iteration $q$ is successful for the edge $e$ is at least:
$$(1-p)^{|F_z(e)|}\cdot p\leq (1-p)^{2\cdot k\log n}\cdot p\leq p/e=1/(2ek\log n)~.$$
Since there are at least $30\cdot k \log^2 n$ iterations, with probability at least $1-n^8$, there is at least one successful iteration for $e$ in phase $i_z$. (By applying the union bound over all $k$ phases, this holds for the edge $e$ with probability $1-n^7$).

Let $q$ be such a successful iteration for the edge $e$ in phase $i_z$. Since $e$ is a non-tree edge of $T_{i_z}$, when applying Algorithm\ $\NonTreeCover$ to the tree $T_{i_z}$ and the subgraph $G_{i_z,q}$, by property (i) of Fact  \ref{fc:nontreecover} the resulting cycle collection $\cC^{i_z}_q$ contains exactly one cycle covering $e$, that we denote by $C^z(e)$. By the definition of a successful iteration, $C^z(e)\cap F_z(e)=\emptyset$, as desired. 

We conclude that with high probability, there is a collection $C^1(e),\ldots, C^{k-\eta+1}(e)$ of cycles in $\cset$, that contain $e$ and are disjoint in their non-tree edges (except for sharing the edge $e$). We set $\cset(e)=\set{C^1(e),\ldots, C^{k-\eta+1}(e)}$

Finally, we claim that each edge $e' \neq e$ appears on at most $\eta$ cycles of $\cset(e)$. By definition, each edge $e' \neq e$ can serve as a non-tree edge of at most \emph{one} cycle. 
In addition, recall that each cycle $C^z(e)$ is computed in phase $i_z$ by applying Algorithm $\NonTreeCover$ to the tree $T_{i_z}$. Since each edge $e'$ may belong to at most $\eta-1$ trees of $\tset$, it may appear as a tree edge on at most $\eta-1$ cycles. 
\end{proof}

\textbf{Congestion:} 
We first bound the number of cycles in $\mathcal{C}$ that contains a given edge $e$ as a non-tree edge.
Fix a phase $i$ such that $e \notin T_i$. There are $\ell=O(k \log^2 n)$ iterations in phase $i$. In each iteration, every edge $e$ is sampled independently with probability $p=1/(2k\log n)$. Therefore, by the Chernoff bound, with high probability, edge $e$ is sampled in at most $O(\log n)$ iterations of a given phase. In each iteration $q$ of phase $i$, we apply Algorithm\ $\NonTreeCover$. By property (i) of Fact \ref{fc:nontreecover}, every non-tree edge appears on exactly one cycle in $\mathcal{C}^i_q$. 
Therefore, $e$ appears on at most $O(\log n)$ cycles in the cycle collection $\cC^i$. By summing over all $k$ phases, an edge $e$ may appear as a non-tree edge on at most $O(k\log n)$ cycles.

We next turn to bound the number of cycles that contains a fixed edge $e$ as a tree edge. 
Fix a tree $T_i$ where $e \in T_i$. By property (iv) of Fact \ref{fc:nontreecover}, $e$ appears on $O(\log n)$ cycles in each application of Algorithm $\NonTreeCover$.  Since there are $O(k\log^2 n)$ applications of this algorithm on $T_i$, and since $e$ appears on at most $\eta$ trees, overall it appears as a tree edge in $O(k \cdot \eta \cdot \log^3 n)$ cycles. 
%
%
This concludes the proof of Lemma \ref{lem:cc-high-conn-teepath}.
\qed


We conclude by showing an immediate application of $(\dilation,\congestion,\congedge,k)$-cycle cover to resilient computation in the presence of eavesdropper.
\begin{lemma}[Compiler for Eavesdropping in Highly Connected Graphs]\label{lem:ecomp}
Assume that a $(\dilation,\congestion,\congedge,k)$-cycle cover $\cset$ is computed in a pre-processing phase. Then, any distributed algorithm $\cA$ can be compiled into an  algorithm $\cA'$ that is \emph{resilient} to an eavesdropping adversary listening on at most $k'\leq \lfloor k/(2a\congedge\cdot \log n)-1 \rfloor$ edges in the graph, for some constant $a$, in every round. W.h.p. this incurs a multiplicative overhead of $\widetilde{O}(\dilation\cdot \congestion\cdot\congedge)$ in the number of rounds. 
\end{lemma}
\begin{proof}
The compiler works round-by-round. Fix a round $i$ of algorithm $\cA$. Observe that the round is fully specified by the collection of messages sent on the edges at this round. We will simulate this round in $\cA'$ using a total of $\widetilde{O}(\dilation\cdot \congestion\cdot\congedge)$ rounds. Consider an edge $e=(u,v)$\footnote{As before, we view the edge $e=(u,v)$ as a directed edge where the message is sent from $u$ to $v$. Thus we will use the cycles covering $e$ twice: to send the message from $u$ to $v$ and vice-verse.} and let $M=M_e$ be the message sent on the edge $e$ in this round from $u$ to $v$. 


For each edge $e=(u,v)$ we will have a subalgorithm $A_e$ that securely sends the original message $M$ from $u$ to $v$
in a total of $O(\dilation \cdot \congedge)$ rounds. The algorithm will run all these subalgorithm $A_e$ for every edge $e$ in parallel using the random delay approach. We first describe $A_e$ and then show that it is indeed secure even if the adversary listens over $k'$ edges in each round, possibly picking a different set of edges in every round. Then we will show how to run all these $A_e$ algorithm in parallel while maintaining the security of each of them.

Algorithm $A_e$ consists of $2\dilation$ phases, each containing $\eta$ rounds. Each of the $k$ cycles $C_1,\ldots, C_k$ covering $e$ corresponds to a distinct $u$-$v$ path, where the length of the path from cycle $C_i$ is denoted by $d_i$. By definition, $d_i \leq \dilation$. The message $M$ is secret shared by the sender $u$ into $\ell=\sum_{i=1}^k (2\dilation-d_i+1)$ shares, where on the $i^{th}$ path (i.e., $C_i \setminus \{(u,v)\}$), $u$ sends a distinct set of $2\dilation-d_i+1\leq 2\dilation$ shares for every $i \in \{1,\ldots, k\}$. These shares are sent on each path in a pipeline manner at a speed of one share per phase on each of the paths. Since each phase consists of $\congedge$ rounds, we will be able to pass the at most $\eta$ messages that need to go through an edge in a given phase (due to at most $\congedge$ different paths that go through this edge). In addition, since the messages are sent in a pipeline manner on each path, in $2\dilation$ phases, $v$ receives $2\dilation-d_i+1 \geq \dilation$ messages from $u$ via the $i^{th}$ path of length $d_i$. 
Overall, $v$ receives at least $\ell \geq \dilation \cdot k$ shares during the $2\dilation$ phases of algorithm $A_e$. Since the adversary can manipulate at most $k'$ messages in a round, it can listen to at most $k' \cdot 2\dilation \cdot \congedge < \ell$ messages, and therefore there is at least one missing share that it did not receive.  

Finally, we show that all these secure subalgorithms $A_e$ can be run in parallel without compromising  security. 
We use the scheduling algorithm of Ghaffari \cite{ghaffari2015near} (see Theorem \ref{thm:delay}), that proceeds in phases, where each phase has $a\cdot \log n$ rounds for some constant $a$. Each subalgorithm $A_e$ has a random starting point and proceeds at a speed of one phase at a time. The argument shows that due to the random start of each algorithm, there are at most $a \log n$ algorithms that send a message on a fixed edge in a given phase, hence all these messages sent through an edge fit within the phase window of $a\log n$ rounds. In the final scheduling, the total number of rounds in which the messages of $A_e$ are sent is at most $r'=a\cdot \log n \cdot 2\dilation\cdot \eta$ (this is because $A_e$ has $2\dilation \eta$ rounds, and thus simulated in $2\dilation \eta$ phases, each of length $a\log n$ rounds). Thus the adversary can listen to at most $r' \cdot k'$ messages of the $A_e$ algorithm. Since $r'\cdot k' < \ell$, the adversary did not receive at least one of the shares.
%
%

Finally we bound the running time of the whole algorithm. Each algorithm $A_e$ takes $O(\dilation \eta)$ rounds and the dominant part in the round complexity of the final compiler is dominated by the edge congestion of all these algorithms. A single subalgorithm $A_e$ sends $O(\dilation \cdot \eta)$ messages over each edge $e'$. Since each cycle $C$ is used by at most $O(\log n)$ edges (the non-tree edges of that cycle), and since each edge $e'$ appears on at most $\congestion$ other cycles, over all each edge $e'$ participates on $\widetilde{O}(\congestion)$ algorithms $A_e$. Thus overall at most $\widetilde{O}(\dilation \cdot \congestion\cdot \eta)$ messages are sent over a single edge. By Theorem \ref{thm:delay}, we get that running all the $A_e$ subalgorithms in parallel takes $\widetilde{O}(\congestion \cdot \dilation \cdot \eta)$ rounds using the random delay approach. 
\end{proof}

By combining with Corollary \ref{lem:cc-high-conn} we get the following:
\begin{corollary}\label{cor:cc-high-conn-dropper-stronger}
(1) Given a $k$-edge connected $D$-diameter graph $G$ with a $(\dilation,\congestion,\congedge,k/2)$ cycle cover satisfying $\congedge=O(1)$, $\dilation=O((101k \ln n)^D\cdot \log n)$, and $\congestion=O(k\log^3 n)$ computed in a preprocessing step. Any distributed algorithm $\cA$ can be compiled into an  algorithm $\cA'$ that is resilient to an eavesdropping adversary listening on at most $k'=O(k/\log n)$ edges in the graph, in every round. W.h.p. this incurs a multiplicative overhead of $\widetilde{O}((101k \ln n)^D\cdot k)$ in the number of rounds. 

\par (2) Given a $(k,D)$-connected graph $G$ with a $(\dilation,\congestion,\eta,k)$ cycle cover satisfying 
$\congedge=O(\log n), \dilation=O(D\log^2 n)$, and $\congestion=O(k\log^4 n)$ computed in a preprocessing step. 
Any distributed algorithm $\cA$ can be compiled into an  algorithm $\cA'$ that is resilient to an eavesdropping adversary listening on at most $k'=O(k/\log^2 n)$ edges in the graph. W.h.p. this incurs a multiplicative overhead of $\widetilde{O}(D\cdot k)$ in the number of rounds. 
\end{corollary}

\paragraph{Secure Broadcast.} The task of broadcasting a single message $M$ from a source vertex $s$ over a spanning tree is arguably one of the most fundamental communication primitives. We consider the task of secure broadcast in the  setting of \emph{store-and-forward} routing algorithms. In this setting all nodes, except for the source node $s$, can only forward the messages they receive and cannot send messages formed by any combination of the previously received messages. This class of algorithms follows the classical paradigm of message routing, that is used in telecommunications and wireless networks. The security of the broadcast scheme is with respect to an eavesdropper 
adversary who listens to messages sent over at most $k'$ edges in the graph. Our goal is to perform a simple store-and-forward broadcast with a guarantee that the adversary learns nothing -- in the information-theoretic sense -- about the original message $M$.

We show, using the low-diameter tree packing from Claim \ref{cl:basic-dist-tool}, that there exists a simple store-and-forward \emph{secure} broadcast algorithm that w.h.p. runs in $\widetilde{O}((101k\log n )^D)$ rounds. 
Let $\tset=\set{T_1,\ldots, T_{k}}$ be the collection of trees with congestion $c'\log n$ for some constant $c'$ obtained by using Claim \ref{cl:basic-dist-tool} with $\eta=1$. The diameter of these trees is $\widetilde{O}((101k\log n)^{D})$ with high probability. The sender $u$ secret shares the message $M$ to $k$ random shares $M_1,\ldots,M_{k}$ such that $M_1 \oplus \dots \oplus M_{k} = M$. Each share $M_i$ is broadcast on the tree $T_i$. 
The collection of all $k$ shares can be sent in $\widetilde{O}((101k\log n )^D)$ rounds over the trees.
Since all trees $T_1,\ldots, T_k$ are spanning, every vertex $v \in G$ receives the $k$ shares and can compute $M$. However, since each edge appears on at most $c'\log n$ trees, any adversary that listens on $k'=\lfloor k/(c'\log n+1)\rfloor$ edges can learn at most $k-1$ shares, and hence learns nothing on the original message $M$.  
Overall, by Claim 5.2 and the random delay approach of Theorem \ref{thm:delay}, the total round complexity is bounded by $\widetilde{O}((101k\log n)^D)$ with high probability.